%% file: main.tex
\RequirePackage{fix-cm}
\RequirePackage{snapshot}
\documentclass[smallcondensed]{svjour3}     
\smartqed  
\usepackage{graphicx}
\usepackage[bookmarks,bookmarksdepth=2]{hyperref}
\usepackage{stmaryrd,scalerel}
\usepackage{amsmath}
\usepackage{amsfonts}
\usepackage{pifont}
\usepackage{mathabx}
\usepackage{algorithm}
\usepackage{algpseudocode}
\usepackage{tikz}
\usetikzlibrary{arrows}
\usepackage{wrapfig}
\usepackage{kbordermatrix}
\usepackage{multicol}
\usepackage{url}
\usepackage{marginnote}
\usepackage{subfig}
\usepackage{thmtools, thm-restate}
\usepackage{grffile}
\usetikzlibrary{matrix} 
\usetikzlibrary{decorations.pathreplacing}
\usetikzlibrary{decorations.pathmorphing}
\usetikzlibrary{patterns}
\usepackage[font=small,labelfont=bf]{caption}

\allowdisplaybreaks

\newcommand{\dbmij}[3]{\ensuremath{\mathbf{#1}_{#2,#3}}}
\newcommand{\dbm}[1]{\ensuremath{\mathbf{#1}}}

\newcommand{\mine}{Min\'{e}}
\newcommand{\dbmstrong}[1]{\ensuremath{\dbm{m^{\textrm{\textbullet}}}}}
\newcommand{\gammadiff}{\ensuremath{\gamma_{\text{diff}}}}
\newcommand{\gammaoct}{\ensuremath{\gamma_{\text{oct}}}}

\newcommand{\floor}[1]{\ensuremath{\lfloor #1 \rfloor}}
\newcommand{\bari}{\ensuremath{\bar\imath}}
\newcommand{\bark}{\ensuremath{\bar{k}}}
\newcommand{\barj}{\ensuremath{\bar\jmath}}
\newcommand{\barb}{\ensuremath{\bar{b}}}
\newcommand{\bara}{\ensuremath{\bar{a}}}

\newcommand{\floorfrac}[2]{\ensuremath{\left \lfloor \frac{#1}{#2} \right \rfloor}}

\newcommand{\CheckConsistent}{\textsc{CheckConsistent}}
\newcommand{\CheckIntegerConsistent}{\textsc{CheckZConsistent}}
\newcommand{\NonIncrementalClosure}{\textsc{Close}}
\newcommand{\IncrementalClosureHoisting}{\textsc{IncCloseHoist}}
\newcommand{\IncrementalClosureInSitu}{\textsc{InplaceIncClose}}
\newcommand{\IncrementalClosure}{\textsc{IncClose}}
\newcommand{\IncrementalStrongClosure}{\textsc{IncStrongClose}}
\newcommand{\IncrementalStrongClosureReduce}{\textsc{IncStrongCloseMotion}}
\newcommand{\IncrementalIntegerClosure}{\textsc{IncZClose}}
\newcommand{\IncrementalIntegerClosureInSitu}{\textsc{InplaceIncZClose}}
\newcommand{\IncrementalStrongClosureInSitu}{\textsc{InplaceIncStrongClose}}
\newcommand{\IncrementalClosureMine}{\textsc{Min\'{e}IncClose}}

\newcommand{\Strengthen}{\textsc{Str}}
\newcommand{\Tighten}{\textsc{Tighten}}
\newcommand{\TightClosure}{\textsc{TightClose}}
\newcommand\sIf[2]{ \If{#1}#2\EndIf}

\numberwithin{theorem}{section}
\numberwithin{lemma}{section}
\numberwithin{corollary}{section}
\numberwithin{proposition}{section}
\numberwithin{definition}{section}
\begin{document}

\title{Incrementally Closing Octagons\thanks{This work was supported by EP/K031929/1 and EP/K032585/1
funded by GCHQ in
association with EPSRC and the EPSRC/NRF UK/Singaporean research grant EP/N020243/1}}


\author{Aziem Chawdhary \and Ed Robbins \and \mbox{Andy King}}


\institute{Aziem Chawdhary,
              School of Computing, University of Kent, Canterbury, CT2 7NF, UK.
              Tel.: +44-1227-827911
              Fax: +44-1227-762811
              \email{a.a.chawdhary@kent.ac.uk}
}

\date{Received: date / Accepted: date}

\maketitle

  \input{abstract}

\input{intro}
\input{related}
\input{prelim}
\input{incrementalclosure}
\input{coherence}
\input{strongclosure}
\input{integerclosure}
\input{inplace}
\input{experiments}
\input{conclusion}
\input{acks}

\bibliographystyle{plain}      
\bibliography{refs}   

\newpage
\appendix
\input{proofs}

\end{document}

%% file: abstract.tex

\begin{abstract}
  The octagon abstract domain is a widely used numeric abstract domain
  expressing relational information between variables whilst being
  both computationally efficient and simple to implement. Each element
  of the domain is a system of constraints where each constraint takes
  the restricted form $\pm x_i \pm x_j \leq c$. A key family of
  operations for the octagon domain are closure algorithms, which
  check satisfiability and provide a normal form for octagonal
  constraint systems. We present new quadratic incremental algorithms
  for closure, strong closure and integer closure and proofs of their
  correctness. We highlight the benefits and measure the performance
  of these new algorithms.
\end{abstract}


%% file: intro.tex

\section{Introduction}
\label{sec:intro}

The view that simplicity is a virtue in competing scientific theories
and that, other things being equal, simpler theories should be
preferred to more complex ones, is widely advocated by scientists and
engineers. Preferences for simpler theories are thought to have played
a role in many episodes in science, and the field of abstract domain
design is no exception. Abstract domains that have enduring appeal are
typically those that are conceptually simple. Of all the weakly
relational domains, for example, octagons \cite{mine_octagon_2006} are
arguably the most popular. One might claim that octagons are more
elegant than, say, the two variable per inequality (TVPI) domain
\cite{tvpi}, and certainly they are easier to understand and
implement. Yet one important operation for this popular domain has
remained elusive: incremental closure.

Inequalities in the octagon domain take the restricted form of
$\pm x_i \pm x_j \leq c$, where $x_i$ and $x_j$ are variables and $c$
is a numerical constant. Difference bound matrices (DBMs) can be
adapted to represent systems of octagonal constraints, for which a key
family of operations is closure. Closure, in its various guises,
provides normal forms for DBMs, allowing satisfiability to be observed
and equality to be checked. Closure also underpins operations such as
join and projection (the forget operator), hence the concept of closure
is central to the design of the whole domain. Closure uses shortest
path algorithms, such as Floyd-Warshall
\cite{Floyd:1962,Warshall:1962}, to check for satisfiability. However,
octagons can encode unary constraints, which require a stronger notion
of closure, known as strong closure, to derive a normal form.
Moreover, a refinement to strong closure, called integer closure, is
required to detect whether octagonal constraints have an integral solution.

A frequent use-case in program analysis is adding a single new
octagonal constraint to a closed DBM and then closing the augmented
system. This is incremental closure. Incremental closure not only
arises when an octagon for one line is adjusted to obtain an octagon
for the next: incremental closure also occurs in integer wrapping
\cite{simon_taming_2007} which involves repeatedly partitioning a
space into two (by adding a single constraint), closing and then
performing translation. Incremental closure is useful in access-based localisation \cite{oh11access}, which
analyses each procedure using abstractions defined over only those
variables it accesses. One way to adapt localisation to octagons
\cite{beckschulze12access} is to introduce fresh variables, called
anchors, that maintain the relationships which hold when a procedure is
entered. One anchor is introduced for each variable that is accessed
within the procedure. The body of the callee is analysed to capture
how a variable changes relative to its anchor, and then this change is
propagated into the caller. The abstraction of the callee is
amalgamated with that of the caller by replacing the variables in the
caller abstraction with their anchors, imposing the constraints from
the callee abstraction, and then eliminating the anchors. If there are
only a few non-redundant constraints in the callee
\cite{bagnara_weakly-relational_2009} then incremental closure
is attractive for combining caller and callee abstractions. Nevertheless, the
experimental results focus on the use-case of adding a single constraint
encountered on one line to an octagon that summaries the previous line.

In SMT solving, difference logic
\cite{nieuwenhuis05dpll} is widely supported, suggesting that an incremental
solver for the theory of octagons \cite{robbins15scp} would also be useful.
Furthermore afield in constraint solving, relational and mixed integer-real abstract domains show promise for
enhancing constraint solvers \cite{pelleau13constraint} and octagons
have been deployed for solving continuous constraints
\cite{pelleau14octagon}. In this context, a split operator is used to
divide the solution space into two sub-spaces by adding opposing
constraints such as $x_i - x_j \leq c$ and $x_j - x_i \leq -c$.
Splitting is repeatedly applied until a set of octagons is derived
that cover the entire solution space, within a given precision
tolerance. Propagation is applied after every split, which suggests
incremental closure, and a scheme in which incremental closure is
applied whenever a propagator updates a variable.  This use-case is also
examined experimentally.

Closing an augmented DBM is less general than closing an arbitrary DBM
and therefore one would expect incremental closure to be both
efficient and conceptually simple. However the running time of the
algorithm originally proposed for incremental closure \cite[Section
4.3.4]{mine-PhD04} is cubic in the number of variables (see Section~\ref{subsec:original} for an
explanation of the impact of row and column swaps).
The algorithms presented in this paper stem from the desire to understand
incremental closure by providing correctness proofs that would, in turn, provide a pathway to mechanisation. 
Yet the act of restructuring the proofs for \cite{ChawdharyRKAPLAS14}, exposed a degenerate form of
propagation and revealed fresh algorithmic insights. 
The resulting family of closure algorithms includes: 
a new algorithm for increment closure;
a new algorithm for strong closure that
performs strengthening on-the-fly, rather
than a separate pass over the whole DBM;
a further refinement to strong closure applicable when the input DBM is strongly closed;
and finally a new incremental closure algorithm for integer DBMs.
All algorithms
significantly outperform the incremental algorithm of \mine\/
\cite[Section 4.3.4]{mine-PhD04}, whilst entirely recovering closure, as
is demonstrated from their deployment in an off-the-shelf abstract interpretation
and a continuous constraint solver.  The dramatic speedups underscore
the importance of this domain operation.

\subsection{Contributions}
\label{sec:contributions}

We summarise the contributions of our work as follows:
\begin{itemize}

\item Using new insights, we present new incremental algorithms for
  closure, strong closure and integer closure
  (Section~\ref{sec:incrementalclosure},
  Section~\ref{sec:strongclosuremain} and
  Section~\ref{sec:integerclosure} respectively). We show how
  code hoisting can be applied to incremental closure and how strength
  reduction can be applied to strong incremental closure.
  
\item We prove our algorithms correct and show how proofs for
  existing closure algorithms can be simplified, paving the way for
  mechanised formalisation. (To keep the length of the paper manageable, the proofs
  are relegated to on-line appendix \cite{chawdhary16incrementing}.  The exception 
  is Lemma~\ref{lemma:closed} since the argument is itself a significant conceptual advance, hence is included in the body of the paper.)

\item We give detailed proofs for in-place versions of our algorithms
  (Section~\ref{sec:inplace}).

\item We implement these new algorithms which show significant
  performance improvements over existing closure algorithms in real-world setting
  (Section~\ref{sec:experiments}).

\end{itemize}
The paper is structured as follows: Section~\ref{sec:related}
contextualises this study and Section~\ref{sec:prelim} provides the
necessary preliminaries. Section~\ref{sec:incrementalclosure}
critiques the incremental algorithm of \mine\/, introduces a new
incremental quadratic algorithm.
Section~\ref{sec:strongclosuremain} shows how to recover strong
closure incrementally and do so, again, in a single DBM pass.
Section~\ref{sec:integerclosure} explains how to extend incrementally
to integer closure. Section~\ref{sec:inplace}
suggest various optimisations to the incremental algorithms including
in-place update. Experimental results are presented in
Section~\ref{sec:experiments} and Section~\ref{sec:conclusion}
concludes.



%% file: related.tex

\section{Related Work}\label{sec:related}

Since the thesis of \mine\/ \cite{mine-PhD04} and his subsequent
magnum opus \cite{mine_octagon_2006}, algorithms for manipulating
octagons, and even their representations, have continued to evolve.
Early improvements showed how strengthening, the act of combining
pairs of unary octagon constraints to improve binary octagon
constraints, need not be applied repeatedly, but instead can be left
to a single post-processing step
\cite{bagnara_weakly-relational_2009}. This result, which was
justified by an inventive correctness argument, led to a 
performance improvement of approximately 20\%
\cite{bagnara_weakly-relational_2009}. Showing that integer octagonal
constraints admit polynomial satisfiability represented another
significant advance \cite{bagnara_improved_2008}, especially since
dropping either the two variable or unary coefficient property makes
the problem NP-complete \cite{lagarias95computational}.

Octagonal representations have come under recent scrutiny
\cite[Chapter 8]{jourdan16verasco}. In Coq, it is natural to realise
DBMs as map from pairs of indices (represented as bit sequences) to
matrix entries. Look-up becomes logarithmic in the dimension of the
DBM, but the DBM itself can be sparse. Strengthening, which combine
bounds on different variables, can populate a DBM with entries for
binary constraints. Dropping strengthening thus improves sparsity,
albeit at the cost of sacrificing a canonical representation. Join can
be recovered by combining bounds during join itself, in effect,
strengthening on-the-fly. Quite independently, sparse representations
have recently been developed for differences \cite{gange16exploiting}.
Further field, $O(mn)$ decision procedures have been proposed for unit
two variable per inequality (UTVPI) constraints
\cite{lahiri05efficient} where $m$ and $n$ are the number of
constraints and variables respectively. Subsequently an incremental
version was proposed for UTVPI \cite{schutt10incremental} with time
complexity $O(m + n\log(n) + p)$ where $p$ is the number of
constraints tightened by the additional inequality. Certifying
algorithms have also been devised for UTVPI constraints
\cite{subramani15graphical}, supported by a graphical representation
of these constraints, which aids the extraction of a certificate for
validating unsatisfiability. DBMs, however, offer additional support
for other operations that arise in program analysis such as join and
projection. Moreover, there is no reason why each DBM entry could not be
augmented with a pair of row and column coordinates which records how
it was updated, allowing a proof for unsatisfiability to be extracted
from a negative diagonal entry.

Other recent work \cite{singh15making} has proposed factoring octagons into
independent sub-systems, which reduces the size of the DBM. Domain
operations are applied point-wise to the independent sub-matrices of the
DBM, echoing \cite{halbwachs06some}. The work also shows how the
regular access patterns of DBMs enable vectorisation, the step beyond
which is harnessing general purpose GPUs \cite{banterle07fast}.
Packs \cite{blanchet03static} have also been proposed as a factoring device in which
the set of programs variables is covered by a sets of variables called packs (or clusters).  
An octagon is computed for each pack
to abstract the DBM as a set of low-dimensional DBMs.  Recent work
has even explored how packs can be introduced automatically using preanalysis and machine
learning \cite{heo16learning}.

The alternative to simplifying the DBM representation is to assume
that the DBM satisfies some prerequisites so that a domain operation
need not be applied in full generality. \mine\/ \cite{mine-PhD04}
showed that an incremental version of the closure could be derived by
observing that a new constraint is independent of the first $c$
variables of the DBM. This paper stems from an earlier work
\cite{ChawdharyRKAPLAS14} that extends an incremental algorithm for
disjunctive spatial constraints which originates in planning \cite{baykan_spatial_1997}.
The work was motivated by the desire to augment
\cite{ChawdharyRKAPLAS14} with conceptually simple correctness proofs,
that revealed a deficiency in the propagation algorithm of \cite{ChawdharyRKAPLAS14} which prompted
a more thorough study of incrementality.

Further afield, closure of octagons echos path consistency in temporal constraint
networks \cite{dechter89temporal}, which also uses the Floyd-Warshall
algorithm to tighten constraints. 
Furthermore, \IncrementalStrongClosure, which processes key entries (staggered diagonal entries) first,
tallies with how
extremal values are first processed in constraint propagation \cite{rossi06handbook}.
Difference constraints can be generalised to Allen constraints \cite{roy16enforcing} to express
set theoretic properties, such as overlap.
Solving Allen constraints is also polynomial, but each variable can be updated many times
when calculating the fixpoint. By way of contrast, the restricted form of octagons means that
each element in the DBM is updated at most once, which is key to the efficiency of incremental closure. 




%% file: prelim.tex

\section{Preliminaries}
\label{sec:prelim}
This section serves as a self-contained introduction to the
definitions and concepts required in subsequent sections. For more
details, we invite the reader to consult both the seminal
\cite{mine-PhD04,mine_octagon_2006} and subsequent
\cite{bagnara_weakly-relational_2009,ChawdharyRKAPLAS14} works on the octagon abstract domain.

\subsection{The Octagon Domain and its Representation}
\label{subsec:domainrep}

An octagonal constraint is a two variable inequality of the form
$ \pm x_i \pm x_j \leq d $ where $x_i$ and $x_j$ are variables and $d$
is a constant.  An octagon is a set of points satisfying a system of
octagonal constraints.  The octagon domain is the set of all
octagons that can be defined over the variables $x_0, \ldots, x_{n-1}$.

Implementations of the octagon domain reuse the machinery developed
for solving difference constraints of the form $x_i - x_j \leq
d$. \mine\/ \cite{mine_octagon_2006} showed how to translate octagonal
constraints to difference constraints over an extended set of
variables $x'_0, \ldots, x'_{2n-1}$.  A single octagonal constraint
translates into a conjunction of one or more difference constraints as
follows:
\[
  \begin{array}{rcrllrl}
    x_i - x_j \leq d & \rightsquigarrow & x'_{2i} - x'_{2j} & \leq
                                                              d & \wedge & x'_{2j+1} - x'_{2i+1} & \leq d \\
    x_i + x_j \leq d & \rightsquigarrow & x'_{2i} - x'_{2j+1} & \leq
                                                                d & \wedge & x'_{2j} - x'_{2i+1} & \leq d \\
    -x_i - x_j \leq d & \rightsquigarrow & x'_{2i+1} - x'_{2j} & \leq
                                                                 d & \wedge & x'_{2j+1} - x'_{2i} & \leq d \\
    x_i \leq d & \rightsquigarrow & x'_{2i} - x'_{2i+1} & \leq 2d \\
    -x_i \leq d & \rightsquigarrow & x'_{2i+1} - x'_{2i} & \leq 2d \\
  \end{array}
  \label{def:octtranslation}
\]
A common representation for difference constraints is a difference
bound matrix (DBM) which is a square matrix of dimension $n \times n$,
where $n$ is the number of variables in the difference system. The
value of the entry $d = \dbmij{m}{i}{j}$ represents the constant $d$
of the inequality $x_i - x_j \leq d$ where the indices
$i,j \in \{ 0, \ldots, n - 1 \}$. An octagonal constraint system over
$n$ variables translates to a difference constraint system over $2n$
variables, hence a DBM representing an octagon has dimension
$2n \times 2n$.

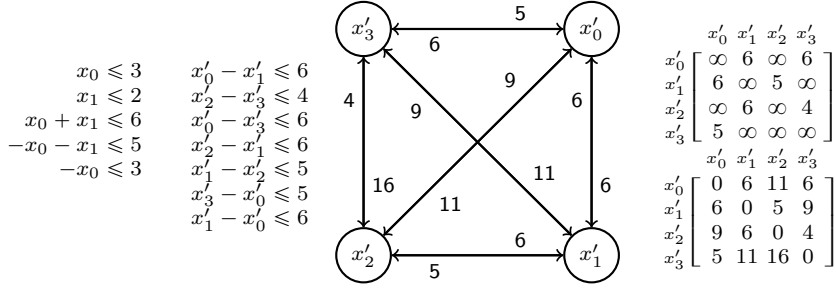
\begin{figure}[t]
\begin{tabular}{@{}c@{\qquad}c@{\quad}c@{\quad}c@{}}
$
\begin{array}{@{}rl@{}}
  x_0 \leq 3 \\
  x_1 \leq 2 \\
  x_0 + x_1 \leq 6 \\
  -x_0 - x_1 \leq 5 \\
  -x_0 \leq 3 \\
\\
\\
\end{array}
$
&
$
\begin{array}{@{}rl@{}}
  x'_0 - x'_1 \leq 6 \\
  x'_2 - x'_3 \leq 4 \\
  x'_0 - x'_3 \leq 6 \\
  x'_2 - x'_1 \leq 6 \\
  x'_1 - x'_2 \leq 5 \\
  x'_3 - x'_0 \leq 5 \\
  x'_1 - x'_0 \leq 6 \\
\end{array}
$
&
\begin{tabular}{@{}c@{}}
\begin{tikzpicture}[->,auto,thick,main node/.style={circle,draw,font=\sffamily\small}]

  \node[main node] (0) at (3,3) {$x'_0$};
  \node[main node] (1) at (3,0) {$x'_1$};
  \node[main node] (2) at (0,0) {$x'_2$};
  \node[main node] (3) at (0,3) {$x'_3$};

  \path[every node/.style={font=\sffamily\small}]
(0) edge node [near end] {6} (1)
(0) edge node [near end] {11} (2)
(0) edge node [near end] {6} (3)

(1) edge node [near end] {9} (3)
(1) edge node [near end] {5} (2)
(1) edge node [near end] {6} (0)

(2) edge node [near end] {9} (0)
(2) edge node [near end] {6} (1)
(2) edge node [near end] {4} (3)

(3) edge node [near end] {5} (0)
(3) edge node [near end] {11} (1)
(3) edge node [near end] {16} (2);
\end{tikzpicture}
\end{tabular}
&
\begin{tabular}{c}
$\kbordermatrix{
             & x'_0   & x'_1   & x'_2   & x'_3 \cr
        x'_0 & \infty & 6      & \infty & 6      \cr
        x'_1 & 6      & \infty & 5      & \infty      \cr
        x'_2 & \infty & 6      & \infty & 4      \cr
        x'_3 & 5      & \infty & \infty & \infty      \cr
      }
$
\\
$\kbordermatrix{
             & x'_0 & x'_1 & x'_2 & x'_3 \cr
        x'_0 & 0    & 6    & 11   & 6      \cr
        x'_1 & 6    & 0    & 5    & 9      \cr
        x'_2 & 9    & 6    & 0    & 4      \cr
        x'_3 & 5    & 11   & 16   & 0      \cr
      }
$
\end{tabular}
\end{tabular}
   \caption{Example of an octagonal system and its DBM representation}
  \label{fig:oct_examples}
\end{figure}

\begin{example} 
  Figure~\ref{fig:oct_examples} serves as an example of how an octagon
  translates to a system of differences.  The entries of the upper DBM
  correspond to the constants in the difference constraints.  Note how
  differences which are (syntactically) absent from the system lead to
  entries which take a symbolic value of $\infty$.  Observe too how
  that DBM represents an adjacency matrix for the illustrated graph
  where the weight of a directed edge abuts its arrow.
\end{example}
The interpretation of a DBM representing an octagon is different to a DBM
representing difference constraints. Consequently there are two
concretisations for DBMs: one for interpreting differences and another
for interpreting octagons, although the latter is defined in terms
of the former:

\begin{definition}\label{def:octconcretization} Concretisation for rational $(\mathbb{Q}^n)$ solutions:
 \[
 \begin{array}{@{}r@{\;}c@{\;}l@{}}
\gammadiff(\dbm{m}) & = & \{ \langle v_0, \ldots, v_{n-1} \rangle \in \mathbb{Q}^n \;\mid\;\forall
  i,j . v_i - v_j \leq \dbmij{m}{i}{j}\}
\\[1ex]
  \gammaoct(\dbm{m}) & = & \{ \langle v_0, \ldots, v_{n-1} \rangle \in \mathbb{Q}^n \;\mid\; \langle
  v_0 , -v_0, \ldots, v_{n-1}, -v_{n-1} \rangle \in \gammadiff(\dbm{m})\}
\end{array}
\]
where the concretisation
for integer $(\mathbb{Z}^n)$ solutions can be defined analogously. 
\end{definition}

\begin{example}
  Since octagonal inequalities are modelled as two related
  differences, the upper DBM contains duplicated entries, for
  instance, $\dbmij{m}{1}{2} = \dbmij{m}{3}{0}$.
\end{example}

\noindent
Operations on a DBM representing an octagon must maintain equality
between the two entries that share the same constant of an octagonal
inequality. This requirement leads to the definition of coherence:

\begin{definition}[Coherence]
\label{def:coherence} 
  A DBM $\dbm{m}$ is coherent iff
  $\forall i.j. \dbmij{m}{i}{j} = \dbmij{m}{\barj}{\bari}$ where
  $\bari = i+1$ if $i$ is even and $i-1$ otherwise.
\end{definition}

\begin{example}
  For the upper DBM observe
  $\dbmij{m}{0}{3} = 6 = \dbmij{m}{2}{1} =
  \dbmij{m}{\bar{3}}{\bar{0}}$.  Coherence holds in a degenerate way for
  unary inequalities, note
  $\dbmij{m}{2}{3} = 4 = \dbmij{m}{2}{3} =
  \dbmij{m}{\bar{3}}{\bar{2}}$.
\end{example}

\noindent
The bar operation can be realised without a branch using
$\bari = i\;\textbf{xor}\;1$ \cite[Section~4.2.2]{mine-PhD04}.
Care should be taken to
preserve coherence when manipulating DBMs, either by carefully
designing algorithms or by using a data structure that enforces
coherence \cite[Section~4.5]{mine-PhD04}. For clarity, we abstract
away from the question of how to represent a DBM by presenting all
algorithms for square matrices, rather than triangular matrices as
introduced
in \cite[Section~4.5]{mine-PhD04}. One final property is necessary for
satisfiability: 

\begin{definition}[Consistency] 
  A DBM $\dbm{m}$ is consistent iff
  $\forall i . \dbmij{m}{i}{i} \geq 0$.
\end{definition}
Intuitively, consistency means that there is not negative cycle in the
DBM, which corresponds to unsatisfiability \cite{citeulike:3116035}.

\subsection{Definitions of Closure}
\label{sec:defnclosure}

Closure properties define canonical representations of DBMs, and can
decide satisfiability and support operations such as join and
projection. Bellman \cite{citeulike:3116035} showed that the
satisfiability of a difference system can be decided using shortest
path algorithms on a graph representing the differences. If the graph
contains a negative cycle (a cycle whose edge weights sum to a negative value)
then the difference system is unsatisfiable.
The same applies for DBMs representing octagons. Closure propagates
all the implicit (entailed) constraints in a system, leaving each
entry in the DBM with the sharpest possible constraint entailed
between the variables. Closure is formally defined below:

\begin{definition}[Closure]
A DBM $\dbm{m}$ is closed iff \begin{itemize}

\item $\forall i. \dbmij{m}{i}{i} = 0$

\item $\forall i,j,k. \dbmij{m}{i}{j} \leq \dbmij{m}{i}{k} + \dbmij{m}{k}{j}$

\end{itemize}
\label{def:closure}
\end{definition}

\begin{example}
  The top right DBM in Figure~\ref{fig:oct_examples} is not closed. By
  running an all-pairs shortest path algorithm we get the following DBM:
  \[
    \kbordermatrix{
             & x'_0 & x'_1 & x'_2 & x'_3 \cr
        x'_0 & 11   & 6    & 11   & 6      \cr
        x'_1 & 6    & 11   & 5    & 9      \cr
        x'_2 & 9    & 6    & 11   & 4      \cr
        x'_3 & 5    & 11   & 16   & 11      \cr
        }
      \]
      Notice that the diagonal values have non-negative elements
      implying that the constraint system is satisfiable. Running
      shortest path closure algorithms propagates all constraints and
      makes every explicit all constraints implied by the original
      system. Once satisfiability has been established, we can set the
      diagonal values to zero to satisfy the definition of closure.
\end{example}

\begin{figure}[t!]
  \centering
  \begin{tikzpicture}
    \draw (0,0) node [circle, style={draw}] (xminus) {$x^{-}$};
    \draw (0,2) node [circle, style={draw}] (xplus) {$x^{+}$};
    
    \draw (2,0) node [circle, style={draw}] (yplus) {$y^{+}$};
    \draw (2,2) node [circle, style={draw}] (yminus) {$y^{-}$};

    \draw (5,0) node [circle, style={draw}] (xminus') {$x^{-}$};
    \draw (5,2) node [circle, style={draw}] (xplus') {$x^{+}$};
    
    \draw (7,0) node [circle, style={draw}] (yplus') {$y^{+}$};
    \draw (7,2) node [circle, style={draw}] (yminus') {$y^{-}$};

    \draw[->] (xminus) to node[left] {$4$} (xplus);
    \draw[->] (yminus) to node[right] {$8$} (yplus);

    \draw[->] (xminus') to node[left] {$4$} (xplus');
    \draw[->] (yminus') to node[right] {$8$} (yplus');
    \draw[->] (xminus') to node[below] {$3$} (yplus');
    \draw[->] (yminus') to node[above] {$3$} (xplus');
  \end{tikzpicture}
  \caption{Intuition behind strong closure: Two closed graphs representing the same octagon: $x \leq 2 \wedge y \leq 4$}
  \label{fig:intuitionstrongclosure}
\end{figure}
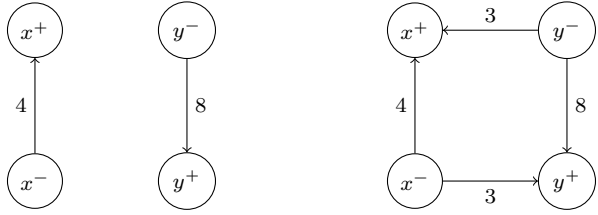
Closure is not enough to provide a canonical form for DBMs
representing octagons. \mine\/ defined the notion of strong closure in
\cite{mine-PhD04,mine_octagon_2006} to do so:
\begin{definition}[Strong closure]
\label{def:strongclosure}
\setcounter{equation}{1}
A DBM $\dbm{m}$ is strongly closed iff 
\begin{itemize}


\item $\dbm{m}$ is closed

\item
$\forall i,j .\dbmij{m}{i}{j} \leq \dbmij{m}{i}{\bari}/2 + \dbmij{m}{\barj}{j}/2$

\end{itemize}
\end{definition}
\noindent
The strong closure of DBM $\dbm{m}$ can be computed by
$\Strengthen(\dbm{m})$, the code for which is given in
Figure~\ref{fig:standardclosure}. The algorithm propagates the property
that if $x'_j - x'_{\barj} \leq c_1$ and $x'_{\bari} - x'_i \leq c_2$
both hold then $x'_j - x'_i \leq (c_1 +c_2)/2$ also holds. This
sharpens the bound on the difference $x'_j - x'_i$ using the two unary
constraints encoded by $x'_j - x'_{\barj} \leq c_1$ and
$x'_{\bari} - x'_i \leq c_1$, namely, $2x'_j \leq c_1$ and
$-2x'_i \leq c_2$. Note that this constraint propagation is not
guaranteed to occur with a shortest path algorithm since there is not
necessarily a path from a $\dbmij{m}{i}{\bari}$ and
$\dbmij{m}{\barj}{j}$. An example in
Figure~\ref{fig:intuitionstrongclosure} shows such a situation: the two
graphs represent the octagon, but a shortest path
algorithm will not propagate constraints on the left graph; hence
strengthening is needed to bring the two graphs to the same normal form.
Strong closure yields a canonical representation: there is a unique
strongly closed DBM for any (non-empty) octagon
\cite{mine_octagon_2006}. Thus any semantically equivalent octagonal
constraint systems are represented by the same strongly closed DBM.
Strengthening is the act of computing strong closure.

\begin{example}
  The lower right DBM in Figure~\ref{fig:oct_examples} gives the
  strong closure of the upper right DBM. Strengthening is
performed after the shortest path algorithm.
\end{example}

\noindent For octagonal constraints over integers, the strong closure
property may result in non-integer values due to the division by two.
The definition of strong closure for integer octagonal constraints
thus needs to be refined. If $x_i$ is integral then $x_i \leq c$
tightens to $x_i \leq \lfloor c \rfloor$.  Since $x_i \leq c$
translates to the difference \mbox{$x'_{2i} - x'_{2i + 1} \leq 2c$},
tightening the unary constraint is achieved by tightening the
difference to \mbox{$x'_{2i} - x'_{2i + 1} \leq 2\lfloor c/2 \rfloor$}.

\begin{definition}[Tight closure]
\label{def:tightclosure}
A DBM $\dbm{m}$ is tightly closed iff
\begin{itemize}
\item $\dbm{m}$ is strongly closed

\item
$\forall i. \dbmij{m}{i}{\bari}\;\textrm{is even}$

\end{itemize}
\end{definition}

\noindent 
For the integer case, a tightening step is required before
strengthening. Tightening a closed DBM results in a weaker form of
closure, called weak closure. Strong closure can be recovered from
weak closure by strengthening \cite{bagnara_improved_2008}. Note,
however, that we introduce the property for completeness of exposition
because our formalisation and proofs do not make use of this notion.

\begin{definition}[Weak closure]  \label{def:weakclosure}
  A DBM $\dbm{m}$ is weakly closed iff
\begin{itemize}

\item
$\forall i . \dbmij{m}{i}{i} = 0$

\item
$\forall i,j,k. \dbmij{m}{i}{k} + \dbmij{m}{k}{j} \geq
  \min(\dbmij{m}{i}{j}, \dbmij{m}{i}{\bari}/2 + \dbmij{m}{\barj}{j}/2)$

\end{itemize}
\end{definition}

\input{highleveldiagram}

\subsection{High-level Overview} \label{subsec:nonincremental}

Figure~\ref{fig:closurealgsoverview} gives a high-level overview of
closure calculation.  First a closure algorithm is applied to a
DBM. Next, consistency is checked by observing the diagonal has
non-negative entries indicating the octagon is satisfiable.  If
satisfiable, then the DBM is strengthened, resulting in a strongly
closed DBM. Note that consistency does not need to be checked again after
strengthening. The dashed lines in the figure show the alternative
path taken for integer problems: to ensure that the DBM entries are
integral, a tightening step is applied which is then followed by an
integer consistency check and strengthening.

Figure~\ref{fig:standardclosure} shows how this architecture can be
instantiated with algorithms for non-incremental strong closure.  A
Floyd-Warshall all-pairs shortest path algorithm
\cite{Floyd:1962,Warshall:1962} can be applied to a DBM to compute
closure, which is cubic in $n$.  The check for consistency involves a
pass over the matrix diagonal to check for a strictly negative entry,
as illustrated in the figure. (Note that \CheckConsistent\/
resets a strictly positive diagonal entry to zero as in
\cite{bagnara_weakly-relational_2009,mine_octagon_2006}, but the
incremental algorithms presented in this paper never relax a zero
diagonal entry to a strictly positive value. Hence the reset is actually redundant for
the incremental algorithms that follow.)  The consistency check is linear in $n$.
Strong closure can be additionally obtained by following closure with
a single call to \Strengthen\/, the code for which is also listed in the
figure.  This is quadratic in $n$.

\begin{figure}[t]
\begin{tabular}{@{}cc@{}}
\begin{tabular}{@{}p{7cm}@{}}
    \begin{algorithmic}[1]
    \Function{\NonIncrementalClosure}{$\dbm{m}$} 
    \For{$k \in \{ 0, \ldots, 2n-1 \}$} 
      \For{$i \in \{ 0, \ldots, 2n-1 \}$} 
        \For{$j \in \{ 0, \ldots, 2n-1 \}$} 
           \State $\dbmij{m'}{i}{j} \gets \min( \dbmij{m}{i}{j}, \dbmij{m}{i}{k} + \dbmij{m}{k}{j})$
    \EndFor
    \EndFor
    \EndFor
    \State $\textbf{return} \; \dbm{m'}$
    \EndFunction
  \end{algorithmic}
\\[-3ex]
  \begin{algorithmic}[1]
    \Function{\Strengthen}{$\dbm{m}$} \For{$i \in \{ 0, \ldots, 2n-1 \}$}
    \For{$j \in \{ 0, \ldots, 2n-1 \}$}
	\State $\dbmij{m'}{i}{j}\; \gets \min(\dbmij{m}{i}{j}, (\dbmij{m}{i}{\bari} + \dbmij{m}{\barj}{j})/2)$
    \EndFor
    \EndFor
    \State $\textbf{return} \; \dbm{m'}$
    \EndFunction
  \end{algorithmic}
\end{tabular}
&
\begin{tabular}{@{}p{6cm}@{}}
  \begin{algorithmic}[1]
    \label{alg:checkconsistent}
      \Function{\CheckConsistent}{$\dbm{m}$}
      \For{$i \in \{ 0, \ldots, 2n-1 \}$}
      \If{$\dbmij{m}{i}{i} < 0$}
	\State $\textrm{\textbf{return} false}$
      \Else 
        \State $\dbmij{m}{i}{i} = 0$ 
      \EndIf
      \EndFor
      \State $\textrm{\textbf{return} true}$
      \EndFunction
    \end{algorithmic}
\\[25ex]
\end{tabular}
\end{tabular}
  \caption{Non-incremental closure and strengthening}
  \label{fig:standardclosure}
\end{figure}


%% file: highleveldiagram.tex

 \begin{figure}[t!]
    \centering
    \begin{tikzpicture}[scale=0.75]
      \node (dbm) at (0,0) [] {DBM $\dbm{m}$};
     \node (shortestpath) at (3,0) [draw,align=center,thin,minimum width=1cm,minimum
      height=1cm] {Closure \\ \textbf{(Sec.~\ref{sec:prelim})}};
      \node (consistency) at (6,0) [draw,align=center,thin,minimum width=1cm,minimum
      height=1cm] {Consistent \\\textbf{(Fig.~\ref{fig:standardclosure})}};

      \node (unsat) at (6,-2) [] {UNSAT};

      \node (tightening) at (6,3) [draw,align=center,thin,minimum width=1cm,minimum
      height=1cm] {Tighten \\ \textbf{(Sec.~\ref{sec:integerclosure})}};

      \node (strengthen) at (12,0) [draw,align=center,thin,minimum width=1cm,minimum
      height=1cm] {Strengthen \\ \textbf{(Sec.~\ref{sec:strongclosuremain})}};

      \node (z-consistent) at (12,3) [draw,align=center,thin,minimum width=1cm,minimum
      height=1cm] {ZConsistent \\ \textbf{(Sec.~\ref{sec:integerclosure})}};
      \node (unsat2) at (12,5) [] {UNSAT};
      \node[text width=3.4cm, align=center] (sat) at (12,-2) [] {Strongly Closed DBM/ Tightly Closed DBM};
        
      \draw[->] (dbm) -- (shortestpath);
      \draw[->] (consistency) to node[left] {} (unsat);
      \draw[->] (shortestpath) to node[above] {} (consistency);
      \draw[->] (consistency) to node[above] {consistent} node[below] {closed DBM} (strengthen);
      \draw[->] (strengthen) -- (sat);

      \draw[dashed,->] (consistency) to node[left] {consistent closed
        DBM} (tightening);
      \draw[dashed,->] (tightening) to node[text width=2.5cm,
      align=center, above] {weakly closed DBM} (z-consistent);
      \draw[dashed,->] (z-consistent) to node[text width=2.5cm, align=center, left] {weakly closed
        consistent DBM} (strengthen);
      \draw[dashed,->] (z-consistent) -- (unsat2);
    \end{tikzpicture}

    \caption{High-Level Overview of Closure Algorithms for Octagons}
    \label{fig:closurealgsoverview}
  \end{figure}
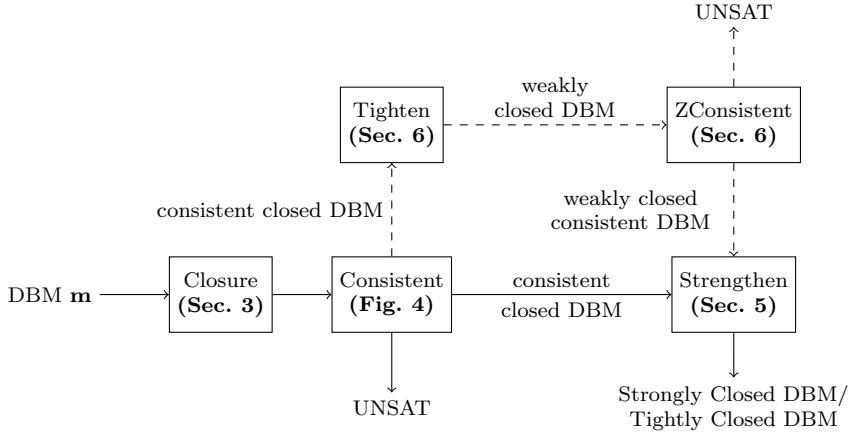


%% file: incrementalclosure.tex

\section{Incremental Closure}\label{sec:incrementalclosure}

We are interested in the specific use case of adding a new octagonal
constraint to an existing closed octagon. \mine\/ designed an incremental
algorithm for this very task, which can be refactored into computing
closure and then separately strengthening, as depicted in
Figure~\ref{fig:closurealgsoverview}. His incremental algorithm, and a
refinement, are discussed in Section~\ref{subsec:original}.
Section~\ref{subsec:new} presents our new incremental algorithm with
improved performance.

\input{mineincrementalclosure}
\input{newincrementalclosure}



%% file: mineincrementalclosure.tex

\subsection{Classical Incremental Closure}
\label{subsec:original}

\mine\/ designed an incremental algorithm based on the observation
that a new constraint will not affect all the variables of the octagon
\cite[Section 4.3.4]{mine-PhD04}. Without loss of generality, suppose
the inequality $x'_a - x'_b \leq d$ is added to the DBM (unary
constraints are supported by putting $b = \bar{a}$). Adding
$x'_a - x'_b \leq d$ implies that the equivalent constraint
$x'_{\bar{b}} - x'_{\bar{a}} \leq d$ is added too, and the entries
$\dbmij{m}{a}{b}$ and $\dbmij{m}{\bar{b}}{\bar{a}}$ are updated to $d$
to reflect this. 
Figure~\ref{fig:mine_incremental_extended} presents a version of the incremental algorithm of \mine\/, specialised
for adding $x'_a - x'_b \leq d$ to a closed DBM.
The algorithm
relies on the observation that updating
$\dbmij{m}{a}{b}$ and $\dbmij{m}{\bar{b}}{\bar{a}}$
will only (initially) mutate the rows
and columns for the $x'_a,x'_b,x'_{\bara}, x'_{\barb}$ variables.
Put $v = \min(a, b, \bara, \barb)$.
Since $\dbm{m}$ was closed, despite the updates, it still follows that
if $k < v$ then
$\dbmij{m}{i}{j} \leq \dbmij{m}{i}{k} + \dbmij{m}{k}{j}$
for all $0 \leq i < 2n$ and $0 \leq j < 2n$.
This is the inductive property which is established after
the first $v$ iterations of the outermost for loop of 
the standard Floyd-Warshall algorithm.
Therefore, to restore closure
it only necessary to apply the remaining $2n - v$ iterations of Floyd-Warshall, which leads
to the algorithm of Figure~\ref{fig:mine_incremental_extended}.

The incremental closure of Figure~\ref{fig:mine_incremental_extended}
reduces the number of $\min$ operations from $8n^3$ to
$(2n - v)4n^2$ (notwithstanding those in \Strengthen\/).  
Prior to the updates,
one could conceivably 
reconfigure the DBM by swapping rows and columns
so that, say, $a = 2n - 4, \bara = 2n - 3, b = 2n - 2, \barb = 2n - 1$.
Then $v = 2n - 4$ reducing incremental closure
to $16n^2$.  However, after closure, the rows and columns would need
to be swapped back to maintain a consistent representation.
Observe too that $x'_{a} - x'_{b} \leq d$ and $x'_{e} - x'_{f} \leq d$ 
can be added to the DBM simultaneously by putting
$v = \min(a, b, \bar{a}, \bar{b}, e, f, \bar{e}, \bar{f})$ and then applying
incremental closure once.  


\begin{figure}[t]
\begin{tabular}{@{}cc@{}}
  \begin{tabular}{@{}p{12cm}@{}}
        \begin{algorithmic}[1]
          \Function{\IncrementalClosureMine}{$\dbm{m}$, $x'_a - x'_b \leq d$}
          \State{$\dbmij{m}{a}{b} \gets \min(\dbmij{m}{a}{b}, d)$}
          \State{$\dbmij{m}{\barb}{\bara} \gets \min(\dbmij{m}{\barb}{\bara}, d)$}
          \State{$v \gets \min(a, b, \bara, \barb)$};
          \For{$k \in \{ v, \ldots, 2n-1 \}$}
          \For{$i \in \{ 0, \ldots, 2n-1 \}$}
          \For{$j \in \{ 0, \ldots, 2n-1 \}$}
          \State{$\dbmij{m}{i}{j} \gets \min(\dbmij{m}{i}{j}, \dbmij{m}{i}{k}+\dbmij{m}{k}{j})$}
          \EndFor
          \EndFor
          \EndFor
          \State \Return{$\dbm{m}$}
          \EndFunction
        \end{algorithmic}
\end{tabular}
\end{tabular}
  \caption{Incremental Closure of \mine\/}
 \label{fig:mine_incremental_extended}
\end{figure}


\begin{figure}[t]
\end{figure}


%% file: newincrementalclosure.tex

\subsection{Improved Incremental Closure}
\label{subsec:new}

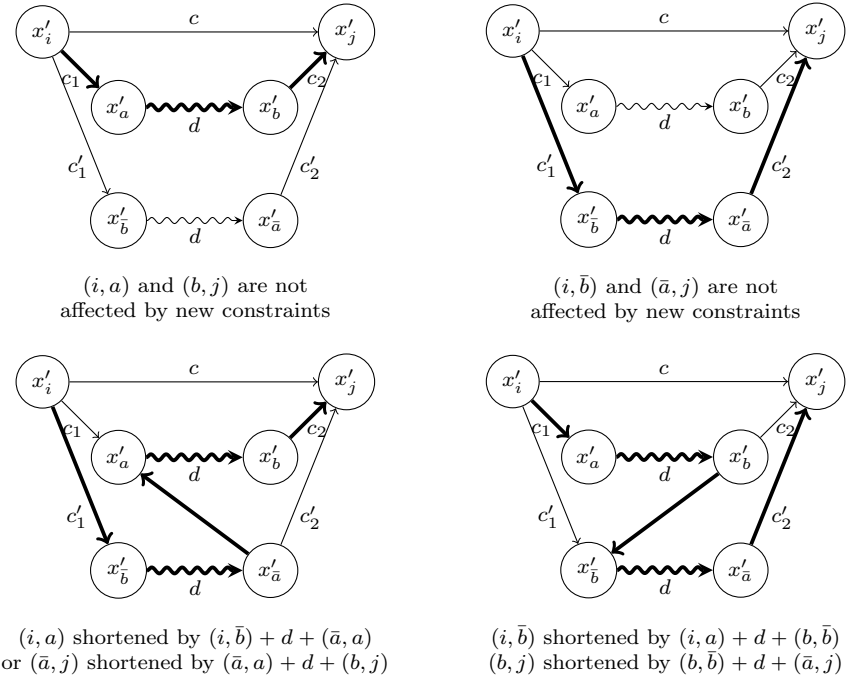
\begin{figure}[t]\centering
    \begin{minipage}[b]{0.45\linewidth}
      \[
      \begin{array}{c}
    \begin{tikzpicture}[baseline={(b.base)}]
      \draw (0,2) node [circle, style={draw}] (i) {$x'_i$}; \draw
      (1,1) node [circle, style={draw}] (a) {$x'_a$}; \draw (3,1) node
      [circle, style={draw}] (b) {$x'_b$}; \draw (4,2) node [circle,
      style={draw}] (j) {$x'_j$}; \draw (1,-0.5) node [circle,
      style={draw}] (a') {$x'_{\bar{b}}$}; \draw (3,-0.5) node [circle,
      style={draw}] (b') {$x'_{\bar{a}}$};

      \draw[->] (i) to node[above] {$c$} (j); \draw[line
      width=0.5mm,->] (i) to node[left, pos=0.8] {$c_1$}(a);
      \draw[line width=0.5mm,->] (b) to node[right, pos=0.2] {$c_2$}
      (j); \draw[->] (b') to node[right, pos=0.2] {$c'_2$} (j);

      \draw[->] (i) to node[left, pos=0.8] {$c'_1$} (a');

      \draw[-stealth, decoration={snake,amplitude=.4mm,segment
        length=2mm,post length=0.9mm},decorate] (a') to node[below]
      {$d$} (b'); \draw[line width=0.5mm,-stealth,
      decoration={snake,amplitude=.4mm,segment length=2mm,post
        length=0.9mm},decorate] (a) to node[below] {$d$} (b);
    \end{tikzpicture}  
        \\
        \\
        \textrm{$(i,a)$ and $(b,j)$ are not} 
        \\
        \textrm{affected by new constraints}
      \end{array}
      \]
  \end{minipage} 
    \hspace{0.5cm}
    \begin{minipage}[b]{0.45\linewidth}
      \[
      \begin{array}{c}
      \begin{tikzpicture}[baseline={(b.base)}]
      \draw (0,2) node [circle, style={draw}] (i) {$x'_i$}; \draw
      (1,1) node [circle, style={draw}] (a) {$x'_a$}; \draw (3,1) node
      [circle, style={draw}] (b) {$x'_b$}; \draw (4,2) node [circle,
      style={draw}] (j) {$x'_j$}; \draw (1,-0.5) node [circle,
      style={draw}] (a') {$x'_{\bar{b}}$}; \draw (3,-0.5) node [circle,
      style={draw}] (b') {$x'_{\bar{a}}$};

      \draw[->] (i) to node[above] {$c$} (j); \draw[->] (i) to
      node[left, pos=0.8] {$c_1$}(a); \draw[->] (b) to node[right,
      pos=0.2] {$c_2$} (j); \draw[line width=0.5mm,->] (b') to
      node[right, pos=0.2] {$c'_2$} (j);

      \draw[line width=0.5mm,->] (i) to node[left, pos=0.8] {$c'_1$}
      (a');

      \draw[line width=0.5mm,-stealth,
      decoration={snake,amplitude=.4mm,segment length=2mm,post
        length=0.9mm},decorate] (a') to node[below] {$d$} (b');
      \draw[-stealth, decoration={snake,amplitude=.4mm,segment
        length=2mm,post length=0.9mm},decorate] (a) to node[below]
      {$d$} (b);
    \end{tikzpicture}  
        \\
        \\
        \textrm{$(i,\bar{b})$ and $(\bar{a},j)$ are not}
        \\
        \textrm{affected by new constraints}
      \end{array}
      \]
    \end{minipage}
    \begin{minipage}[b]{0.45\linewidth}
      \[
      \begin{array}{c}
    \begin{tikzpicture}[baseline={(b.base)}]
      \draw (0,2) node [circle, style={draw}] (i) {$x'_i$}; \draw
      (1,1) node [circle, style={draw}] (a) {$x'_a$}; \draw (3,1) node
      [circle, style={draw}] (b) {$x'_b$}; \draw (4,2) node [circle,
      style={draw}] (j) {$x'_j$}; \draw (1,-0.5) node [circle,
      style={draw}] (a') {$x'_{\bar{b}}$}; \draw (3,-0.5) node [circle,
      style={draw}] (b') {$x'_{\bar{a}}$};

      \draw[->] (i) to node[above] {$c$} (j); \draw[->] (i) to
      node[left, pos=0.8] {$c_1$}(a); \draw[line width=0.5mm,->] (b)
      to node[right, pos=0.2] {$c_2$} (j); \draw[->] (b') to
      node[right, pos=0.2] {$c'_2$} (j); \draw[line width=0.5mm, ->]
      (b') to node[right,pos=0.2] {} (a);

      \draw[line width=0.5mm, ->] (i) to node[left, pos=0.8] {$c'_1$}
      (a');

      \draw[line width=0.5mm,-stealth,
      decoration={snake,amplitude=.4mm,segment length=2mm,post
        length=0.9mm},decorate] (a') to node[below] {$d$} (b');
      \draw[line width=0.5mm,-stealth,
      decoration={snake,amplitude=.4mm,segment length=2mm,post
        length=0.9mm},decorate] (a) to node[below] {$d$} (b);
    \end{tikzpicture} 
        \\
        \\
        \textrm{$(i,a)$ shortened by $(i,\bar{b}) + d + (\bar{a},a)$}
        \\
        \textrm{or $(\bar{a},j)$ shortened by $(\bar{a},a) + d + (b,j)$}
        \end{array}
        \]
    \end{minipage}
    \hspace{0.5cm}
    \begin{minipage}[b]{0.45\linewidth}
      \[
      \begin{array}{c}
    \begin{tikzpicture}[baseline={(b.base)}]
      \draw (0,2) node [circle,style={draw}] (i)
      {$x'_i$}; \draw (1,1) node [circle,style={draw}] (a)
      {$x'_a$}; \draw (3,1) node [circle,style={draw}] (b)
      {$x'_b$}; \draw (4,2) node [circle,style={draw}] (j)
      {$x'_j$}; \draw (1,-0.5) node [circle,style={draw}] (a')
      {$x'_{\bar{b}}$}; \draw (3,-0.5) node [circle,style={draw}] (b')
      {$x'_{\bar{a}}$};

      \draw[->] (i) to node[above]
      {$c$} (j); \draw[line width=0.5mm,->] (i) to node[left, pos=0.8]
      {$c_1$}(a); \draw[->] (b) to node[right, pos=0.2]
      {$c_2$} (j); \draw[line width=0.5mm,->] (b') to node[right,
      pos=0.2]
      {$c'_2$} (j); \draw[line width=0.5mm,->] (b) to
      node[right,pos=0.2] {} (a');

      \draw[->] (i) to node[left, pos=0.8] {$c'_1$} (a');

      \draw[line width=0.5mm,-stealth,
      decoration={snake,amplitude=.4mm,segment length=2mm,post
        length=0.9mm},decorate] (a') to node[below] {$d$} (b');
      \draw[line width=0.5mm, -stealth,
      decoration={snake,amplitude=.4mm,segment length=2mm,post
        length=0.9mm},decorate] (a) to node[below] {$d$} (b);
    \end{tikzpicture}
        \\
        \\
        \textrm{$(i,\bar{b})$ shortened by $(i,a) + d + (b,\bar{b})$}
        \\
        \textrm{$(b,j)$ shortened by $(b,\bar{b}) + d + (\bar{a},j)$}
      \end{array}
      \]
    \end{minipage}
  \caption{Four ways to reduce the distance between $x'_i$ and $x'_j$}
  \label{fig:casesmin}
\end{figure}

To give the intuition behind our new incremental closure algorithm,
consider adding the constraint $x'_a - x'_b \leq d$ and
$x'_{\bar{b}} - x'_{\bar{a}} \leq d$ to the closed DBM $\dbm{m}$. The
four diagrams given in Figure~\ref{fig:casesmin} illustrate how the path
between variables $x'_i$ and $x'_j$ can be shortened. The distance
between $x'_i$ and $x'_j$ is $c$ ($\dbmij{m}{i}{j} = c$), the distance
between $x'_i$ and $x'_a$ is $c_1$ ($\dbmij{m}{i}{a} = c_1$), etc. The
wavy lines denote the new constraints $x'_a - x'_b \leq d$ and
$x'_{\bar{b}} - x'_{\bar{a}} \leq d$ and the heavy lines indicate
short-circuiting paths between $x'_i$ and $x'_j$.
The bottom left path of the figure illustrates how the distance
between $x'_i$ and $x'_a$ can be reduced from $c_1$ by the
$x'_{\bar{b}} - x'_{\bar{a}} \leq d$ constraint. The same path
illustrates how to shorten the distance between $x'_{\bar{a}}$ and
$x'_j$ from $c'_2$ using the $x'_a - x'_b \leq d$ constraint. The
bottom right path of the figure gives two symmetric cases in which
$c'_1$ and $c_2$ are sharpened by the addition of $x'_a - x'_b \leq d$
and $x'_{\bar{b}} - x'_{\bar{a}} \leq d$ respectively. Note that we
cannot have the two paths from $x'_i$ to $x'_a$ and from $x'_b$ to
$x'_j$ both shortened: at most one of them
can change. The same holds for the two paths from $x'_i$ to
$x'_{\barb}$ and $x'_{\bara}$ to $x'_j$. These extra paths lead to the
following strategy for updating $\dbmij{m'}{i}{j}$:
\[
\dbmij{m'}{i}{j} \gets \min \left( \begin{array}{l}
                             \dbmij{m}{i}{j},\\
                             \dbmij{m}{i}{a} + d + \dbmij{m}{b}{j},\\
                             \dbmij{m}{i}{\bar{b}} + d + \dbmij{m}{\bar{a}}{j},\\
                             \dbmij{m}{i}{\bar{b}} + d + \dbmij{m}{\bar{a}}{a} + d + \dbmij{m}{b}{j} \\
                             \dbmij{m}{i}{a} + d + \dbmij{m}{b}{\bar{b}} + d + \dbmij{m}{\bar{a}}{j} \\
                           \end{array} \right)
\]
This leads to the incremental closure algorithm listed in top of
Figure~\ref{fig:newincrementalclosurewithoutqueue}. Quintic $\min$ can be
realised as four binary $\min$ operations, hence the total number
of binary $\min$ operations required for \IncrementalClosure\/ is
$16n^2$, which is quadratic in $n$.   The listing in the
bottom of the figure shows how commonality can be factored out so that
each iteration of the inner loop requires a single ternary $\min$ to
be computed. Factorisation reduces the number of binary $\min$ operations to $2n(2 + 4n)$ = $8n^2 + 4n$ 
in \IncrementalClosureHoisting\/.  Moreover, this form of code hoisting
is also applicable algorithms that follow (though this optimisation is not elaborated in the sequel).
Furthermore, like \IncrementalClosure\/,
\IncrementalClosureHoisting\/ is not sensitive to the specific traversal order
of the DBM, hence has potential for parallelisation.   In additional, both
\IncrementalClosure\/ and
\IncrementalClosureHoisting\/ do not incur any checks. 

\begin{figure}[t]
\begin{tabular}{p{12cm}}
  \begin{algorithmic}[1]
        \Function{\IncrementalClosure}{$\dbm{m}$, $x'_{a} - x'_{b} \leq d$}
    \For{$i \in \{ 0, \ldots, 2n-1 \}$} 
    \For{$j \in \{ 0, \ldots, 2n-1 \}$}
    \State $\dbmij{m'}{i}{j} \gets \min \left(
      \begin{array}{l}
        \dbmij{m}{i}{j}, \\
        \dbmij{m}{i}{a}+ d + \dbmij{m}{b}{j},\\
        \dbmij{m}{i}{\bar{b}} + d + \dbmij{m}{\bar{a}}{j},\\
        \dbmij{m}{i}{\bar{b}} + d + \dbmij{m}{\bar{a}}{a} + d + \dbmij{m}{b}{j}, \\
        \dbmij{m}{i}{a} + d + \dbmij{m}{b}{\bar{b}} + d + \dbmij{m}{\bar{a}}{j} \\
      \end{array}
    \right)
    $
    \EndFor
    \EndFor
        \If{\CheckConsistent($\dbm{m'}$)}
        \State \Return{$\dbm{m'}$}
	\Else
        \State \Return{$false$}
	\EndIf
    \EndFunction
  \end{algorithmic}
  \\
    \begin{algorithmic}[1]
        \Function{\IncrementalClosureHoisting}{$\dbm{m}$, $x'_{a} - x'_{b} \leq d$}
           \State $t_1 \gets d + \dbmij{m}{\bar{a}}{a} + d$;
           \State $t_2 \gets d + \dbmij{m}{b}{\bar{b}} + d$; 
        \For{$i \in \{ 0,  \ldots,  2n-1 \}$} 
           \State $t_3 \gets \min(\dbmij{m}{i}{a} + d, \dbmij{m}{i}{\bar{b}} + t_1)$;
           \State $t_4 \gets \min(\dbmij{m}{i}{\bar{b}} + d, \dbmij{m}{i}{a} + t_2)$; 
        \For{$j \in \{ 0, \ldots, 2n-1 \}$}
           \State $\dbmij{m'}{i}{j} \gets \min(\dbmij{m}{i}{j}, t_3 + \dbmij{m}{b}{j}, t_4 + \dbmij{m}{\bar{a}}{j})$
        \EndFor
         \If{$\dbmij{m'}{i}{i} < 0$}
        \State \Return{$false$}
	\EndIf        
        \EndFor
        \State \Return{$\dbm{m}'$}
        \EndFunction
      \end{algorithmic}
\end{tabular}
  \caption{Incremental Closure (without and with code hoisting)}
  \label{fig:newincrementalclosurewithoutqueue}
\end{figure}

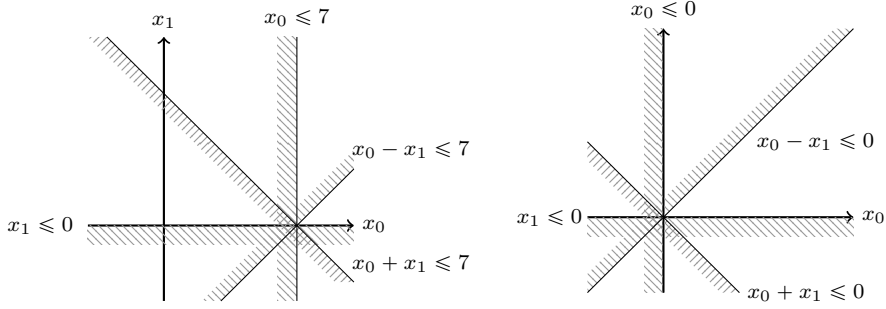
\begin{figure}[t]
  \begin{tabular}{cc}
  \begin{tikzpicture}[scale=0.25, axis/.style={thick,->}]
    \draw[axis] (-4,0) -- (10,0) node(xline)[right] {$x_0$};
    \draw[axis] (0,-4) -- (0,10) node(yline)[above] {$x_1$};

    \draw[draw=none, pattern=north east lines, pattern color=black!40] (10,-3) -- (9,-3) -- (-4,10) -- (-3,10) --cycle;
    \draw[draw=none, pattern=north west lines, pattern color=black!40] (-4,0) -- (10,0) -- (10,-1) -- (-4,-1)-- cycle;
    \draw[draw=none, pattern=north west lines, pattern color=black!40] (10,3) -- (10,4) -- (2,-4) -- (3,-4) -- cycle;

    \draw[draw=none, pattern=north west lines, pattern color=black!40] (7,-4) -- (6,-4) -- (6,10) -- (7,10) -- cycle;
    
    \node[] at (13,-2) {$x_0 + x_1 \leq 7$};
    \node[] at (-6.5,0) {$ x_1 \leq 0$};
    \node[] at (13,4) {$x_0 - x_1 \leq 7$};
    \node[] at (7, 11) {$x_0 \leq 7$};
    \draw (10,-3) -- (-3,10) node(ypoint)[left] {};
    \draw (10,3) -- (3,-4) node(ypoint2)[left] {};
    \draw (7,-4) -- (7,10) node(x)[left] {};
  \end{tikzpicture} 
    &
 \begin{tikzpicture}[scale=0.25, axis/.style={thick,->}]
    \draw[axis] (-4,0) -- (10,0) node(xline)[right] {$x_0$};
    \draw[axis] (0,-4) -- (0,10) node(yline)[above] {};

    \draw[draw=none, pattern=north east lines, pattern color=black!40] (-4,4) -- (-4,3) -- (3,-4) -- (4,-4) -- cycle;

    \draw[draw=none, pattern=north west lines, pattern color=black!40] (-4,0) -- (10,0) -- (10,-1) -- (-4,-1)-- cycle;
    \draw[draw=none, pattern=north west lines, pattern color=black!40] (10,10) -- (9,10) -- (-4,-3) -- (-4,-4) -- cycle;

    \draw[draw=none, pattern=north west lines, pattern color=black!40] (0,-4) -- (-1,-4) -- (-1,10) -- (0,10) -- cycle;
    
    \node[] at (7.5,-4) {$x_0 + x_1 \leq 0$};
    \node[] at (-6,0) {$ x_1 \leq 0$};
    \node[] at (8,4) {$x_0 - x_1 \leq 0$};
    \node[] at (0, 11) {$x_0 \leq 0$};
    \draw (-4,4) -- (4,-4) node(ypoint)[left] {};
    \draw (10,10) -- (-4,-4) node(ypoint2)[left] {};
    \draw (0,-4) -- (0,10) node(x)[left] {};
  \end{tikzpicture} 
  \end{tabular}
  \caption{Before and after adding $x_0 - x_1 \leq 0$}
  \label{fig:before-and-after}
\end{figure}

\begin{example}
  \label{example:aplasiswrong}
To illustrate how the incremental closure algorithm of
  \cite{ChawdharyRKAPLAS14}, from which the above is derived, omits a form of propagation, consider adding
  $x_0 - x_1 \leq 0$, or equivalently $x'_0 - x'_2 \leq 0$, to the 
  system on the left
  \[
\begin{array}{l}
\\
    x_0 \leq 7, \\
    x_1 \leq 0,  \\
    x_0 - x_1 \leq 7, \\
    x_0 + x_1 \leq 0 
\end{array}
\qquad\qquad
\dbm{m} =
      \begin{array}{c}
        \kbordermatrix{
        & x'_0   & x'_1 & x'_2   & x'_3 \\
        x'_0 & 0      & 14   & 7      & 7 \\
        x'_1 & \infty & 0    & \infty & \infty \\
        x'_2 & \infty & 7    & 0      & 0 \\
        x'_3 & \infty & 7    & \infty & 0 \cr} 
      \end{array}
  \]
  whose DBM $\dbm{m}$ is given on right. The system is illustrated
  spatially on the left hand side of
  Figure~\ref{fig:before-and-after}; the right hand side of the same
  figure shows the effect of adding the constraint $x_0 - x_1 \leq 0$.
  Adding $x_0 - x_1 \leq 0$ using the incremental closure algorithm
  from \cite{ChawdharyRKAPLAS14} gives the DBM $\dbm{m'}$;
  \IncrementalClosure\/ gives the DBM $\dbm{m''}$:
\[
    \dbm{m'} = \begin{array}{@{}r@{\qquad\qquad}r@{}}
             \begin{array}{c}
               \kbordermatrix{
               & x'_0   & x'_1 & x'_2   & x'_3 \\
               x'_0 & 0      & 7   & 0      & 0 \\
               x'_1 & \infty & 0    & \infty & \infty \\
               x'_2 & \infty & 0    & 0      & 0 \\
               x'_3 & \infty & 0    & \infty & 0 
               }  
             \end{array}
           &
     \dbm{m''} =  \begin{array}{c}
               \kbordermatrix{
               & x'_0   & x'_1 & x'_2   & x'_3 \\
               x'_0 & 0      & 0   & 0      & 0 \\
               x'_1 & \infty & 0    & \infty & \infty \\
               x'_2 & \infty & 0    & 0      & 0 \\
               x'_3 & \infty & 0    & \infty & 0 
               }  
             \end{array}
    \end{array}
  \]  
  The DBM $\dbm{m'}$ represents the constraint $x \leq \frac{7}{2}$ but
  $\dbm{m''}$ encodes the tighter constraint $x \leq 0$. The reason
  for the discrepancy between entries $\dbmij{m'}{0}{1}$ and
  $\dbmij{m''}{0}{1}$ is shown by the following calculations:
  \[
    \dbmij{m'}{0}{1} = \min \left (
                      \begin{array}{l}
                        \dbmij{m}{0}{1} \\
                        \dbmij{m}{0}{0} + 0 + \dbmij{m}{2}{1} \\
                        \dbmij{m}{0}{\bar 2} + 0 + \dbmij{m}{\bar 0}{1} \\
                      \end{array} \right ) 
                    = \min \left (
                      \begin{array}{l}
                        14, \\
                        0 + 0 + 7 \\
                        7 + 0 + 0 \\
                      \end{array} \right) = 7 \\
  \]
  \[
    \dbmij{m''}{0}{1} = \min \left (
                      \begin{array}{l}
                        \dbmij{m}{0}{1} \\
                        \dbmij{m}{0}{0} + 0 + \dbmij{m}{2}{1} \\
                        \dbmij{m}{0}{\bar 2} + 0 + \dbmij{m}{\bar 0}{1} \\
                        \dbmij{m}{0}{0} + 0 + \dbmij{m}{2}{\bar 2} + 0 + \dbmij{m}{\bar 0}{1} \\
                        \dbmij{m}{0}{\bar 2} + 0 + \dbmij{m}{\bar 0}{0} + 0 + \dbmij{m}{2}{1} \\ 
                      \end{array} \right )
    = \min \left (
    \begin{array}{l}
      14 \\
      0 + 0 + 7 \\
      7 + 0 + 0 \\
      0 + 0 + 0 + 0 + 0 \\
      7 + 0 + \infty + 0 + 7 \\
    \end{array}
    \right ) = 0
    \]
    The entry at $\dbmij{m'}{0}{1}$ is calculated using
    $\dbmij{m}{2}{1}$, but this entry will itself reduce to 0;
    $\dbmij{m'}{0}{1}$ must take into account the change that occurs
    to $\dbmij{m}{2}{1}$. More generally, when calculating 
    $\dbmij{m'}{i}{j}$, the $\min$ expression of
    \cite{ChawdharyRKAPLAS14} overlooks how the added constraint can tighten 
$\dbmij{m}{i}{a}$, $\dbmij{m}{i}{b}$,
    $\dbmij{m}{i}{\barb}$ or $\dbmij{m}{\bara}{j}$. 
\qed
\end{example}
The new incremental algorithm is justified by
Theorem~\ref{thm:incrclosureclosed} which, in turn, is supported by
the following lemma:

\begin{lemma}
\label{lemma:nonzero}
Suppose $\dbm{m}$ is a closed DBM, $\dbm{m'}$ = $\IncrementalClosure(\dbm{m}, o)$
and  $o = (x'_a - x'_b \leq d)$.  Then $\dbm{m'}$ is consistent if and only if:
\begin{itemize}

\item $\dbmij{m}{b}{a} + d \geq 0$

\item $\dbmij{m}{\bar{a}}{\bar{b}} + d \geq 0$

\item $\dbmij{m}{\bar{a}}{a} + d + \dbmij{m}{b}{\bar{b}} + d \geq 0$

\end{itemize}
\end{lemma}


\begin{theorem}[\rm Correctness of $\IncrementalClosure$]
  \label{thm:incrclosureclosed}
  Suppose $\dbm{m}$ is a closed DBM, $\dbm{m'}$ =
  $\IncrementalClosure(\dbm{m}, o)$ and $o = (x'_a - x'_b \leq d)$.
  Then $\dbm{m'}$ is either closed or it is not consistent.
\end{theorem}


Note that unsatisfiability can be detected without applying any $\min$
operations at all, though for brevity this is omitted in the presentation of
the algorithms. Fast unsatisfiability checking is
justified by the following corollary of Lemma~\ref{lemma:nonzero}:

\begin{corollary}\label{cor:nonzero1}
Suppose $\dbm{m}$ is a closed DBM, $\dbm{m'}$ = $\IncrementalClosure(\dbm{m}, o)$
and  $o = (x'_a - x'_b \leq d)$.   If
\begin{itemize}

\item $\dbmij{m}{b}{a} + d < 0$ or

\item $\dbmij{m}{\bar{a}}{\bar{b}} + d < 0$ or

\item $\dbmij{m}{b}{\bar{b}} + d + \dbmij{m}{\bar{a}}{a} + d < 0$

\end{itemize}
then $\dbm{m'}$ is not consistent.
\end{corollary}


%% file: coherence.tex

\subsection{Properties of Incremental Closure}

By design \IncrementalClosure\/ recovers closure, but it should also
be natural for the algorithm to preserve and enforce other properties
too. These properties are not just interesting within themselves; they
provide scaffolding for results that follow:

\begin{proposition}\label{lemma-monotonicity}
Suppose $\dbm{m} \leq \dbm{m}'$ (pointwise) and $o = (x'_a - x'_b \leq d)$.  Then
$\IncrementalClosure(\dbm{m}, o)
\leq
\IncrementalClosure(\dbm{m}', o)$.
\end{proposition}

\begin{proposition}\label{lemma:coherence}
Suppose $\dbm{m}$ is coherent, $\dbm{m'}$ = $\IncrementalClosure(\dbm{m}, o)$
and  $o = (x'_a - x'_b \leq d)$.  Then $\dbm{m'}$ is coherent.
\end{proposition}


An important property of \IncrementalClosure\/ is idempotence: it
formalises the idea that an octagon should not change shape if it is
repeatedly intersected with the same inequality. If idempotence did
not hold then there would exist
$\dbm{m'} = \IncrementalClosure(\dbm{m}, o)$ and
$\dbm{m''} = \IncrementalClosure(\dbm{m'}, o)$ for which
$\dbm{m'} \neq \dbm{m''}$. This would suggest that
\IncrementalClosure\/ did not properly tighten $\dbm{m}$ using the
inequality $o$, but overlooked some propagation, which is the form of
suboptimal behaviour we are aiming to avoid.


\begin{proposition}
\label{lemma-idempotence}
Suppose that $\dbm{m}$ is a closed DBM,  
$\dbm{m'} = \IncrementalClosure(\dbm{m}, o)$,
$\dbm{m''} = \IncrementalClosure(\dbm{m'}, o)$ and
$o = (x'_a - x'_b \leq d)$. Then either $\dbm{m'}$ is consistent and
$\dbm{m''} = \dbm{m'}$ or $\dbm{m''}$ is not consistent.
\end{proposition}

%

%


%% file: strongclosure.tex

\section{Incremental Strong Closure}
\label{sec:strongclosuremain}

We now turn our attention from recovering closure to recovering strong
closure, which generates a canonical representation for any
(non-empty) octagon.

\subsection{Classical Strong Closure}

The classical strong closure by \mine\ repeatedly invokes $\Strengthen$ within
the main Floyd-Warshall loop, but it was later shown by Bagnara et al.
\cite{bagnara_weakly-relational_2009} that this was equivalent to
applying  $\Strengthen$ just once after the main loop. The following theorem
\cite[Theorem 3]{bagnara_weakly-relational_2009} justifies this tactic, though the
proofs we present have been revisited and streamlined:

\begin{theorem}
\label{thm:strongclosurestrengthen}
Suppose $\dbm{m}$ is a closed, coherent DBM and
$\dbm{m'} = \Strengthen(\dbm{m})$.
Then $\dbm{m'}$ is a strongly closed DBM.
\end{theorem}

\subsection{Properties of Strong Closure}
We establish a number of properties about \Strengthen\/
which will be useful when we prove in-place versions of our
incremental strong (and tight) closure algorithms. 


\begin{proposition}\label{lemma-stengthen-idempotent}
Suppose $\dbm{m}$ be a DBM and 
$\dbm{m'} = \Strengthen(\dbm{m})$.
Then $\dbm{m'} = \Strengthen(\dbm{m'})$.
\end{proposition}


\begin{proposition}\label{prop-str-monotonicity}
Suppose $\dbm{m}^1 \leq \dbm{m}^2$ (pointwise).
Then
$\Strengthen(\dbm{m}^1) \leq \Strengthen(\dbm{m}^2)$.
\end{proposition}


\begin{proposition}
  \label{prop:strongreductive}
  Suppose $\dbm{m}$ is a DBM and
  $\dbm{m'} = \Strengthen(\dbm{m})$. Then
  $\dbm{m'} \leq \dbm{m}$.
\end{proposition}


\begin{proposition}
  \label{prop:strongcoherence}
  Suppose $\dbm{m}$ is a closed, coherent DBM. Then $\dbm{m'} = \Strengthen(\dbm{m})$ is a coherent DBM. 
\end{proposition}

\subsection{Incremental Strong Closure}
\label{sec:incrstrongclosure}

\begin{figure}[t]
\begin{tabular}{p{12cm}}
      \begin{algorithmic}[1]
        \Function{\IncrementalStrongClosure}{$\dbm{m}, x'_a - x'_b \leq d }$
        \For{$i \in \{ 0, \ldots, 2n-1 \}$}
       \State $\dbmij{m'}{i}{\bari} \gets \min\left(
         \begin{array}{l}
           \dbmij{m}{i}{\bari}, \\
           \dbmij{m}{i}{a} + d + \dbmij{m}{b}{\bari}, \\
           \dbmij{m}{i}{\bar{b}} + d + \dbmij{m}{\bar{a}}{\bari},\\
           \dbmij{m}{i}{\bar{b}} + d + \dbmij{m}{\bar{a}}{a} + d + \dbmij{m}{b}{\bari}, \\
           \dbmij{m}{i}{a} + d + \dbmij{m}{b}{\bar{b}} + d + \dbmij{m}{\bar{a}}{\bari} \\
         \end{array}
       \right )$
       \EndFor

        \For{$i \in \{ 0, \ldots, 2n-1 \}$} 
        \For{$j \in \{ 0, \ldots, 2n-1 \}$}
        \sIf{$j \neq \bari$}{
        \State $\dbmij{m'}{i}{j} \gets \min\left (
          \begin{array}{l}
            \dbmij{m}{i}{j}, \\
            \dbmij{m}{i}{a}+ d + \dbmij{m}{b}{j},\\
            \dbmij{m}{i}{\bar{b}} + d + \dbmij{m}{\bar{a}}{j},\\
            \dbmij{m}{i}{\bar{b}} + d + \dbmij{m}{\bar{a}}{a} + d + \dbmij{m}{b}{j}, \\
            \dbmij{m}{i}{a} + d + \dbmij{m}{b}{\bar{b}} + d + \dbmij{m}{\bar{a}}{j}, \\
            (\dbmij{m'}{i}{\bari} + \dbmij{m'}{\barj}{j})/2
          \end{array}
        \right )$
      }
        \EndFor
        \If{$\dbmij{m'}{i}{i} < 0$}
        \State \Return{$false$}
	\EndIf
        \EndFor
        \State \Return{$\dbm{m}'$}
        \EndFunction
      \end{algorithmic}
\end{tabular}
  \caption{Incremental Strong Closure}
  \label{fig:newincrementalstrongclosurewithoutqueue}
\end{figure}

Theorem~\ref{thm:strongclosurestrengthen} states that a strongly
closed DBM can be obtained by calculating closure and then
strengthening. This is realised by calling $\IncrementalClosure$, from
Figure~\ref{fig:newincrementalclosurewithoutqueue}, followed by a
call to $\Strengthen$. Although this is conventional wisdom, it
incurs two passes over the DBM: one by $\IncrementalClosure$ and the
other by $\Strengthen$.
The two passes can be unified by observing that strengthening
$\dbm{m}'$ critically depends on the entries $\dbmij{m}{i}{\bari}'$
where $i \in \{ 0, \ldots, 2n - 1 \}$. Furthermore, these entries,
henceforth called key entries, are themselves not changed by
strengthening because:
\[
\min(\dbmij{m}{i}{\bari}', (\dbmij{m}{i}{\bari}' + \dbmij{m}{\bar{\bari}}{\bari}')/2)
=
\min(\dbmij{m}{i}{\bari}', (\dbmij{m}{i}{\bari}' + \dbmij{m}{i}{\bari}')/2)
=
\dbmij{m}{i}{\bari}'
\]
This suggests precomputing the key entries up front and then using
them in the main loop of $\IncrementalClosure$ to strengthen
on-the-fly. This insight leads to the algorithm listed in
Figure~\ref{fig:newincrementalstrongclosurewithoutqueue}. Line~3
generates the key entries which are closed by construction and
unchanged by strengthening. Once the key entries are computed, the
algorithm iterates over the rest of the DBM, closing
and simultaneously strengthening each entry $\dbmij{m}{i}{j}$  at line~8. 

The total number of binary $\min$ operations required for
\IncrementalStrongClosure\/ is $8n + 10n(2n - 1) = 20n^2 - 2n$, which
improves on following \IncrementalClosure\/ by \Strengthen\/, which
requires $16n^2 + 4n^2 = 20n^2$. Furthermore, since $\dbm{m}$ is
coherent
$\dbmij{m}{i}{a}+ d + \dbmij{m}{b}{\bari} = \dbmij{m}{\bar{a}}{\bari}
+ d + \dbmij{m}{i}{\bar{b}} = \dbmij{m}{i}{\bar{b}} + d +
\dbmij{m}{\bar{a}}{\bari}$ so that the quintic $\min$ on line~4
becomes quartic, reducing the $\min$ count for $\IncrementalClosure$
to $20n^2 - 4n$. Furthermore, the entry $\dbmij{m}{i}{\bari}$ can be
cached in a linear array $\dbm{a}_i$ of dimension $2n$ and the
expression $(\dbmij{m'}{i}{\bari} + \dbmij{m'}{\barj}{j})/2$ in line~8
can be replaced with $(\dbm{a}_{i} + \dbm{a}_{\barj})/2$, thereby
avoiding two lookups in a two-dimensional matrix. We omit the
algorithm using array caching for space reasons as this is a simple
change to Figure~\ref{fig:newincrementalstrongclosurewithoutqueue}.

The following theorem
justifies the correctness of the new incremental strong closure
algorithm:
\begin{theorem}[\rm Correctness of \IncrementalStrongClosure]
\label{thm:incrstrongclosure}
Suppose $\dbm{m}$ is a DBM, \linebreak
$\dbm{m'} = \IncrementalStrongClosure(\dbm{m}, o)$,
$\dbm{m^{\dagger}} = \IncrementalClosure(\dbm{m}, o)$,
$\dbm{m^{*}} = \Strengthen(\dbm{m^{\dagger}})$ and \linebreak
\mbox{$o = (x'_a - x'_b \leq d)$}. Then
$\dbm{m'} = \dbm{m^{*}}$.
\end{theorem}

Code is duplicated in \IncrementalStrongClosure\/ in
the assignments of $\dbmij{m'}{i}{\bari}$ and $\dbmij{m'}{i}{j}$ on lines~3 and 8 respectively.  
Fig~\ref{fig:newincrementalstrongclosurewithstrengthreduction} shows how this
can be factored out in that line~3 
of \IncrementalStrongClosureReduce\/ need only consider updates
stemming from $\dbmij{m}{i}{a} + d + \dbmij{m}{b}{\bari}$.  Moreover, the guard
on line~7 of Fig~\ref{fig:newincrementalstrongclosurewithoutqueue} is eliminated but moving
the remainder of the $\dbmij{m'}{i}{\bari}$ calculation
into the main loop.  This increases the $\min$ count by $2n$ but reduces code size.
This can potentially be a good exchange because $\min$ is itself essentially a check (though it
can be implemented as straight-line code for machine integers \cite{warren02hackers}), and eliminating
the guard from the main loop avoids $4n^2$ checks, giving a saving overall.  
However, putting asymptotic arguments aside, whether
\IncrementalStrongClosureReduce\/  outperforms \IncrementalStrongClosure\/ depends
on the relative cost of the integer comparison on line~7 of 
Fig~\ref{fig:newincrementalstrongclosurewithoutqueue}  to
the comparison implicit in line~3 of Fig~\ref{fig:newincrementalstrongclosurewithstrengthreduction}, which is 
performed in the underlying number system. The following result justifies this form of code motion:

\begin{theorem}[\rm Correctness of \IncrementalStrongClosureReduce]
  \label{lemma:incrstrongclosurereduce}
  Suppose $\dbm{m}$ is a strongly closed, coherent DBM and let
  $\dbm{m^{*}} = \IncrementalStrongClosure(\dbm{m},o)$ where \linebreak
  $o = (x'_a - x'_b \leq d)$ and
\[
    \dbmij{m''}{i}{j} = \min\left (
          \begin{array}{l}
            \dbmij{m}{i}{j}, \\
            \dbmij{m}{i}{a}+ d + \dbmij{m}{b}{j},\\
            \dbmij{m}{i}{\bar{b}} + d + \dbmij{m}{\bar{a}}{j},\\
            \dbmij{m}{i}{\bar{b}} + d + \dbmij{m}{\bar{a}}{a} + d + \dbmij{m}{b}{j}, \\
            \dbmij{m}{i}{a} + d + \dbmij{m}{b}{\bar{b}} + d + \dbmij{m}{\bar{a}}{j}, \\
            (\dbmij{m}{i}{a} + d + \dbmij{m}{b}{\bari} + \dbmij{m}{\barj}{j})/2, \\
            (\dbmij{m}{i}{\bari} + \dbmij{m}{\barj}{a} + d + \dbmij{m}{b}{j})/2
          \end{array} \right )
  \]
Then either $\dbm{m^{*}} = \dbm{m''}$ or $\dbm{m^{*}}$ is not consistent and $\dbm{m''}$ is not inconsistent.
\end{theorem}

\noindent The force of the above result is that $\dbmij{m'}{i}{j}$ is only
affected by a change to $\dbmij{m'}{i}{\bari}$ via
$\dbmij{m}{i}{a} + d + \dbmij{m}{b}{\bari}$ or a change to $\dbmij{m'}{\barj}{j}$
via $\dbmij{m}{\barj}{a} + d + \dbmij{m}{b}{j}$. Thus the initial
loop on line 3, need only check whether
$\dbmij{m}{i}{\bari}$ is shortened by
$\dbmij{m}{i}{a} + d + \dbmij{m}{b}{\bari}$ in order to correctly
update an arbitrary entry $\dbmij{m}{i}{j}$ in the loop on line 8.  Note
that $\dbm{m}$ is not just required to be closed, but also strongly closed and coherent.


\begin{figure}[t]
\begin{tabular}{p{12cm}}
      \begin{algorithmic}[1]
        \Function{\IncrementalStrongClosureReduce}{$\dbm{m}, x'_a - x'_b \leq d }$
        \For{$i \in \{ 0, \ldots, 2n-1 \}$}
       \State $\dbmij{m'}{i}{\bari} \gets \min\left(
         \begin{array}{l}
           \dbmij{m}{i}{\bari}, \\
           \dbmij{m}{i}{a} + d + \dbmij{m}{b}{\bari} \\
         \end{array}
       \right )$
       \EndFor

        \For{$i \in \{ 0, \ldots, 2n-1 \}$} 
        \For{$j \in \{ 0, \ldots, 2n-1 \}$}

        \State $\dbmij{m'}{i}{j} \gets \min\left (
          \begin{array}{l}
            \dbmij{m}{i}{j}, \\
            \dbmij{m}{i}{a}+ d + \dbmij{m}{b}{j},\\
            \dbmij{m}{i}{\bar{b}} + d + \dbmij{m}{\bar{a}}{j},\\
            \dbmij{m}{i}{\bar{b}} + d + \dbmij{m}{\bar{a}}{a} + d + \dbmij{m}{b}{j}, \\
            \dbmij{m}{i}{a} + d + \dbmij{m}{b}{\bar{b}} + d + \dbmij{m}{\bar{a}}{j}, \\
            (\dbmij{m'}{i}{\bari} + \dbmij{m'}{\barj}{j})/2
          \end{array}
        \right )$

        \EndFor
        \If{$\dbmij{m'}{i}{i} < 0$}
        \State \Return{$false$}
	\EndIf
        \EndFor
        \State \Return{$\dbm{m}'$}
        \EndFunction
      \end{algorithmic}
\end{tabular}
  \caption{Incremental Strong Closure with code motion}
  \label{fig:newincrementalstrongclosurewithstrengthreduction}
\end{figure}


%% file: integerclosure.tex

\section{Incremental Tight Closure}
\label{sec:integerclosure}

The strong closure algorithms previously presented have to be modified
to support integer octagonal constraints. If $x_i$ is integral then
$x_i \leq c$ can be tightened to $x_i \leq \lfloor c \rfloor$. Since
$x_i \leq c$ is represented as the difference
$x'_{2i} - x'_{2i + 1} \leq 2c$, tightening is achieved by sharpening
the difference to $x'_{2i} - x'_{2i + 1} \leq 2\lfloor c/2 \rfloor$,
so that the constant $2\lfloor c/2 \rfloor$ is even. This is achieved
by applying $\textsc{Tighten}(\dbm{m})$, the code for which is given
in Figure~\ref{fig:tightclosure}. As suggested by
Figure~\ref{fig:closurealgsoverview}, closure does not need to be
reapplied after tightening to check for consistency; it is sufficient
to check that \mbox{$\dbmij{m}{i}{\bari} + \dbmij{m}{\bari}{i} < 0$}
\cite{bagnara_weakly-relational_2009}, which is the role of
$\CheckIntegerConsistent(\dbm{m})$. One subtlety that is worthy of
note is that after running $\textsc{tighten}(\dbm{m})$ on a closed DBM
$\dbm{m}$, the resulting DBM will not necessarily be closed but will
instead satisfy a weaker property, namely weak closure. Strong closure
can be recovered from weak closure, however, by strengthening
\cite{bagnara_weakly-relational_2009}. However, we do not use this
approach in the sequel: instead we use tightening and strengthening
together to avoid having to work with weakly closed DBMs. First we
prove that tightening followed by strengthening will return a closed
DBM when the resulting system is satisfiable:

\begin{lemma}
\label{lemma:closed}
Suppose $\dbm{m}$ is a closed, coherent integer DBM. Let $\dbm{m'}$ be defined as follows:
\[
  \dbmij{m'}{i}{j} = \min(\dbmij{m}{i}{j}, \floorfrac{\dbmij{m}{i}{\bari}}{2} + \floorfrac{\dbmij{m}{\barj}{j}}{2})
\]
Then $\dbm{m'}$ is either closed or it is not consistent. 
\end{lemma}

\begin{proof} Suppose $\dbm{m'}$ is consistent.
Because $\dbm{m}$ is closed $\dbmij{m'}{i}{i} \leq \dbmij{m}{i}{i} = 0$
and since $\dbm{m'}$ is consistent $0 \leq \dbmij{m'}{i}{i}$
hence $\dbmij{m'}{i}{i} = 0$.
Now to show $\dbmij{m'}{i}{k} + \dbmij{m'}{k}{j} \geq \dbmij{m'}{i}{j}$. 
   \begin{enumerate}

     \item Suppose $\dbmij{m'}{i}{k} = \dbmij{m}{i}{k}$ and $\dbmij{m'}{k}{j} = \dbmij{m}{k}{j}$.
 Because $\dbm{m}$ is closed:
\begin{align*}
         \dbmij{m'}{i}{k} + \dbmij{m'}{k}{j} = \dbmij{m}{i}{k} + \dbmij{m}{k}{j} \geq \dbmij{m}{i}{j} \geq \dbmij{m'}{i}{j} 
\end{align*}

\item Suppose $\dbmij{m'}{i}{k} \neq \dbmij{m}{i}{k}$ and $\dbmij{m'}{k}{j} = \dbmij{m}{k}{j}$. 
\begin{enumerate}

       \item Suppose $\dbmij{m}{\bark}{k}$ is even. Because $\dbm{m}$ is closed and coherent:
         \begin{align*}
\dbmij{m'}{i}{k} + \dbmij{m'}{k}{j}  & = \floorfrac{\dbmij{m}{i}{\bari}}{2}  +  \floorfrac{\dbmij{m}{\bark}{k}}{2} + \dbmij{m}{k}{j} 
	= \floorfrac{\dbmij{m}{i}{\bari}}{2} + \frac{\dbmij{m}{\bark}{k} + 2\dbmij{m}{k}{j}}{2} \\
           & \geq \floorfrac{\dbmij{m}{i}{\bari}}{2} + \frac{\dbmij{m}{\bark}{j} + \dbmij{m}{k}{j}}{2} 
           = \floorfrac{\dbmij{m}{i}{\bari}}{2} + \frac{\dbmij{m}{\barj}{k} + \dbmij{m}{k}{j}}{2} \\
           & \geq  \floorfrac{\dbmij{m}{i}{\bari}}{2} +  \frac{\dbmij{m}{j}{\barj}}{2} \geq  \floorfrac{\dbmij{m}{i}{\bari}}{2}  +  \floorfrac{\dbmij{m}{j}{\barj}}{2} \geq \dbmij{m'}{i}{j} 
         \end{align*}   
        
       \item Suppose $\dbmij{m}{\bark}{k}$ is odd. Then
           \begin{align*}
\dbmij{m'}{i}{k} + \dbmij{m'}{k}{j}  & = \floorfrac{\dbmij{m}{i}{\bari}}{2}  +  \floorfrac{\dbmij{m}{\bark}{k}}{2} + \dbmij{m}{k}{j} = \floorfrac{\dbmij{m}{i}{\bari}}{2} + \frac{(\dbmij{m}{\bark}{k} - 1) + 2\dbmij{m}{k}{j}}{2} 
         \end{align*}        
Because $\dbm{m}$ is closed and coherent:
         \begin{align*}
         \frac{(\dbmij{m}{\bark}{k} - 1) + 2\dbmij{m}{k}{j}}{2} 
           & \geq \frac{\dbmij{m}{\bark}{j} + \dbmij{m}{k}{j} - 1}{2} 
           = \frac{\dbmij{m}{\barj}{k} + \dbmij{m}{k}{j} - 1}{2} \geq  \frac{\dbmij{m}{\barj}{j} - 1}{2}  
         \end{align*}
\begin{enumerate}

\item
Suppose $\dbmij{m}{\bark}{k} + 2\dbmij{m}{k}{j} =  \dbmij{m}{\barj}{j}$.  
Since $\dbmij{m}{\bark}{k}$ is odd $\dbmij{m}{\barj}{j}$ is odd thus
\[
\frac{\dbmij{m}{\barj}{j} - 1}{2}  = \floorfrac{\dbmij{m}{\barj}{j}}{2} 
\text{ and }
\dbmij{m'}{i}{k} + \dbmij{m'}{k}{j}  \geq \floorfrac{\dbmij{m}{i}{\bari}}{2} + \floorfrac{\dbmij{m}{\barj}{j}}{2} 
\geq \dbmij{m'}{i}{j}
\]

\item
Suppose $\dbmij{m}{\bark}{k} + 2\dbmij{m}{k}{j} > \dbmij{m}{\barj}{j}$.  
Thus 
$(\dbmij{m}{\bark}{k} - 1) + 2\dbmij{m}{k}{j} \geq \dbmij{m}{\barj}{j}$ 
\[
\dbmij{m'}{i}{k} + \dbmij{m'}{k}{j}  \geq \floorfrac{\dbmij{m}{i}{\bari}}{2} + \frac{\dbmij{m}{\barj}{j}}{2}  \geq \floorfrac{\dbmij{m}{i}{\bari}}{2} + \floorfrac{\dbmij{m}{\barj}{j}}{2} 
\geq \dbmij{m'}{i}{j}
\]

\end{enumerate}        

       \end{enumerate}

     \item Suppose $\dbmij{m'}{i}{k} = \dbmij{m}{i}{k}$ and $\dbmij{m'}{k}{j} \neq \dbmij{m}{k}{j}$.
       Symmetric to the previous case.
       
     \item Suppose $\dbmij{m'}{i}{k} \neq \dbmij{m}{i}{k}$ and $\dbmij{m'}{k}{j} \neq \dbmij{m}{k}{j}$. 
Then
   \begin{align*}
     \dbmij{m'}{i}{k} + \dbmij{m'}{k}{j} & = \floorfrac{\dbmij{m}{i}{\bari}}{2} +  \floorfrac{\dbmij{m}{\bark}{k}}{2} + \floorfrac{\dbmij{m}{k}{\bark}}{2} +  \floorfrac{\dbmij{m}{\barj}{j}}{2}    
   \end{align*}

Since $\dbm{m}$ is closed and $\dbm{m'}$ is consistent:
   \[
0 \leq \dbmij{m'}{\bark}{\bark} = \min(\dbmij{m}{\bark}{\bark}, \floorfrac{\dbmij{m}{\bark}{k}}{2} + \floorfrac{\dbmij{m}{k}{\bark}}{2}) = \min(0, \floorfrac{\dbmij{m}{\bark}{k}}{2} + \floorfrac{\dbmij{m}{k}{\bark}}{2})
   \]
Therefore
\[
     \floorfrac{\dbmij{m}{\bark}{k}}{2} + \floorfrac{\dbmij{m}{k}{\bark}}{2} \geq 0
\text{ and }
\dbmij{m'}{i}{k} + \dbmij{m'}{k}{j} \geq \floorfrac{\dbmij{m}{i}{\bari}}{2} + \floorfrac{\dbmij{m}{\barj}{j}}{2} \geq \dbmij{m'}{i}{j}
\]
\qed
\end{enumerate}
\end{proof}

\noindent It should be noted that the above proof by-passes the notion of weak closure which was
previously thought to be necessary
\cite[pages 28--31]{bagnara_weakly-relational_2009} greatly simplifying the proofs.
Using the proof that tighten and strengthening gives a closed DBM, it can now be shown
that the resulting DBM is also tightly closed:

\begin{theorem}(\cite[Theorem 4]{bagnara_weakly-relational_2009})
\label{lemma:tightlyclosed}
Suppose $\dbm{m}$ is a closed, coherent integer DBM. Let $\dbm{m'}$ be
defined as follows:
\[
  \dbmij{m'}{i}{j} = \min(\dbmij{m}{i}{j}, \floorfrac{\dbmij{m}{i}{\bari}}{2} + \floorfrac{\dbmij{m}{\barj}{j}}{2})
\]
Then $\dbm{m'}$ is either tightly closed or it is not consistent.
\end{theorem}
 
 Notice that the proof of tight closure does not use the concept
 of weak closure as advocated in \cite{bagnara_weakly-relational_2009}. 
The above proof goes
 directly from a closed DBM to a tightly closed DBM relying only on simple algebra;
 it is not based on showing
 that tightening gives a weakly closed (intermediate) DBM which can be subsequently strengthen
 to give a tightly closed DBM (see
 Figure~\ref{fig:closurealgsoverview}).
\begin{figure}[t] 
    \centering
    \begin{multicols}{2}
      \begin{algorithmic}[1]
        \Function{\Tighten}{$\dbm{m}$}
        \For{$i \in \{ 0, \ldots, 2n-1 \}$} \State
        $\dbmij{m}{i}{\bari} \gets 2\lfloor \dbmij{m}{i}{\bari} / 2 \rfloor$
        \EndFor
        \EndFunction
      \end{algorithmic}
      \qquad
      \begin{algorithmic}[1]
    \Function{\TightClosure}{$\dbm{m}$}
      \State $\textsc{ShortestPathClosure(\dbm{m})}$
      \If{$\CheckConsistent(\dbm{m})$}
      \State $\dbm{m} \gets \textsc{Tighten(\dbm{m})}$
      \If{$\CheckIntegerConsistent(\dbm{m})$}
        \State $\textrm{\textbf{return} \Strengthen(\dbm{m})}$
      \Else
        \State $\textrm{\textbf{return} false}$
      \EndIf
      \Else
        \State $\textrm{\textbf{return} false}$
      \EndIf
    \EndFunction
  \end{algorithmic}
      \columnbreak
    \begin{algorithmic}[1]
      \Function{\CheckIntegerConsistent}{$\dbm{m}$}
      \For{$i \in \{ 0, \ldots, 2n-1 \}$}
      \If{$\dbmij{m}{i}{\bari} + \dbmij{m}{\bari}{i} < 0$} \State $\textrm{\textbf{return} false}$
      \EndIf
      \EndFor
      \State $\textrm{\textbf{return} true}$
      \EndFunction
    \end{algorithmic} 
  \end{multicols}
    \caption{Tight Closure}
  \label{fig:tightclosure}
\end{figure}

\begin{figure}[t]
      \begin{algorithmic}[1]
        \Function{\IncrementalIntegerClosure}{$\dbm{m},\;x'_a - x'_b \leq d }$
        \For{$i \in \{ 0, \ldots, 2n-1 \}$}
       \State $\dbmij{m'}{i}{\bari} \gets 2 \left \lfloor \min\left(
         \begin{array}{l}
           \dbmij{m}{i}{\bari}, \\
           \dbmij{m}{i}{a} + d + \dbmij{m}{b}{\bari}, \\
           \dbmij{m}{i}{\bar{b}} + d + \dbmij{m}{\bar{a}}{\bari},\\
           \dbmij{m}{i}{\bar{b}} + d + \dbmij{m}{\bar{a}}{a} + d + \dbmij{m}{b}{\bari}, \\
           \dbmij{m}{i}{a} + d + \dbmij{m}{b}{\bar{b}} + d + \dbmij{m}{\bar{a}}{\bari} \\
         \end{array}
       \right ) / 2 \right \rfloor$
       \EndFor
      \If{$\CheckIntegerConsistent(\dbm{m'})$}
       \For{$i \in \{ 0, \ldots, 2n-1 \}$} 
        \For{$j \in \{ 0, \ldots, 2n-1 \}$}
        \sIf{$j \neq \bari$}{
        \State $\dbmij{m'}{i}{j} \gets \min\left (
          \begin{array}{l}
            \dbmij{m}{i}{j}, \\
            \dbmij{m}{i}{a}+ d + \dbmij{m}{b}{j},\\
            \dbmij{m}{i}{\bar{b}} + d + \dbmij{m}{\bar{a}}{j},\\
            \dbmij{m}{i}{\bar{b}} + d + \dbmij{m}{\bar{a}}{a} + d + \dbmij{m}{b}{j}, \\
            \dbmij{m}{i}{a} + d + \dbmij{m}{b}{\bar{b}} + d + \dbmij{m}{\bar{a}}{j}, \\
            (\dbmij{m'}{i}{\bari} + \dbmij{m'}{\barj}{j})/2
          \end{array}
        \right )$
      }
        \EndFor

      \If{$\dbmij{m'}{i}{i} < 0$}
        \State \Return{$false$}
      \EndIf
              \EndFor
      \Else
        \State \Return{$false$}
      \EndIf
        \State \Return{$\dbm{m}'$}
        \EndFunction
      \end{algorithmic}
  \caption{Incremental Tight Closure}
  \label{fig:newincrementalintegerclosure}
\end{figure}

Tight closure requires the key entries, and only these, to
be tightened. This suggests tightening the key entries on-the-fly
immediately after they have been computed by closure. This leads to
the algorithm given in Figure~\ref{fig:newincrementalintegerclosure}
which coincides with $\IncrementalStrongClosure(\dbm{m})$ except in
one crucial detail: line~4 tightens the key entries as they are
computed. Moreover the key entries are strengthened, with the other
entries of the DBM, in the main loop in tandem with the closure
calculation, thereby ensuring strong closure. Thus tightening can be
accommodated, almost effortlessly, within incremental strong closure.

\begin{theorem}[\rm Correctness of \IncrementalIntegerClosure]
  Suppose $\dbm{m}$ is an integer DBM and
  $\dbm{m'} = \IncrementalIntegerClosure(\dbm{m},o)$ where
  $o = x'_a - x'_b \leq d$. Let
  $\dbm{m^{\dagger}} = \IncrementalClosure(\dbm{m}, o)$,
  $\dbm{m^{\ddag}} = \Tighten(\dbm{m^{\dagger}})$ and
  $\dbm{m^*} = \Strengthen(\dbm{m^{\ddag}})$. Then $\dbm{m^*} = \dbm{m'}$.
\end{theorem}\label{theorem:correct-inc-tight-closure}

\subsection{Properties of Tight Closure}
\label{sec:propintclosure}

We prove a number of properties about \Tighten\/ which will be
useful when we justify the in-place versions of our incremental tight closure algorithm.


\begin{proposition}
  \label{prop:tighten-idempotence}
  Suppose $\dbm{m}$ is a DBM and $\dbm{m'} = \Tighten(\dbm{m})$.
  Then $\dbm{m'} = \Tighten(\dbm{m'})$.
\end{proposition}


\begin{proposition}
  \label{prop:intmonotonicity}
  Suppose $\dbm{m}^1 \leq \dbm{m}^2$ (pointwise).
  Then $\Tighten(\dbm{m}^1) \leq \Tighten(\dbm{m}^2)$.
\end{proposition}


\begin{proposition}
  \label{prop:tightenreductiveness}
  Suppose $\dbm{m}$ is a DBM and
  $\dbm{m'} = \Tighten(\dbm{m})$. Then $\dbm{m'} \leq \dbm{m}$.
\end{proposition}


\begin{proposition}
  \label{prop:tightencoherence}
  Let $\dbm{m}$ be a coherent DBM and
  $\dbm{m'} = \Tighten(\dbm{m})$. Then $\dbm{m'}$ is coherent.
\end{proposition}


%% file: inplace.tex

\section{In-place Update}
\label{sec:inplace}

Closure algorithms are traditionally formulated in a way that is
simple to reason about mathematically (see
\cite[Def~3.3.2]{mine-PhD04}), typically using a series of
intermediate DBMs and then present the algorithm itself using in-place
update (see \cite[Def~3.3.3]{mine-PhD04}). An operation on a DBM will conceptually calculate an output DBM from the input DBM.
Since this requires two DBMs, the input and the output, to be stored simultaneously, it is attractive to mutate the input DBM to derive the output DBM. This is called in-place update. 
The subtlety of in-place update, in the context of a DBM operation, is that one element can be calculated in terms of others, some of which may have already been updated. 
 The question of equivalence
between the mathematical formulation and the practical in-place
implementation is arguably not given the space it should. Min\'e, in
his magnus opus \cite{mine-PhD04}, merely states that equivalence can
be shown by using an argument for the Floyd-Warshall algorithm
\cite[Section 26.2]{cormen90introduction}. However that in-place
argument is itself informal. Later editions of the book do not help,
leaving the proof as an exercise for the reader. But the question of
equivalence is more subtle again for incremental closure. Correctness
is therefore argued for incremental closure in
Section~\ref{sect:in-place-closure}, incremental strong closure in
Section~\ref{sect:in-place-strong} and incremental tight closure in
Section~\ref{sect:in-place-integer}, one correctness argument
extending another.

\subsection{In-place Incremental Closure}\label{sect:in-place-closure}

Figure~\ref{fig:inplaceinsitu} gives an in-place version of
\IncrementalClosure\/ algorithm listed in
Figure~\ref{fig:newincrementalclosurewithoutqueue}. At first glance
one might expect that mutating the entries $\dbmij{m}{i}{a}$,
$\dbmij{m}{b}{\bari}$, $\dbmij{m}{i}{\barb}$,
$\dbmij{m}{\bara}{\bari}$, $\dbmij{m}{\bara}{a}$ or
$\dbmij{m}{b}{\barb}$ could potentially perturb those entries of
$\dbm{m}$ which are updated later. The following theorem asserts that
this is not so. Correctness follows from
Corollary~\ref{corollary:idempotence-facts} which is stated below:

\begin{figure}
  \centering
     \begin{algorithmic}[1]
        \Function{\IncrementalClosureInSitu}{$\dbm{m}$, $x'_{a} - x'_{b} \leq d$}
    \For{$i \in \{ 0, \ldots, 2n-1 \}$} 
    \For{$j \in \{ 0, \ldots, 2n-1 \}$}
    \State $\dbmij{m}{i}{j} \gets \min \left(
      \begin{array}{l}
        \dbmij{m}{i}{j}, \\
        \dbmij{m}{i}{a}+ d + \dbmij{m}{b}{j},\\
        \dbmij{m}{i}{\bar{b}} + d + \dbmij{m}{\bar{a}}{j},\\
        \dbmij{m}{i}{\bar{b}} + d + \dbmij{m}{\bar{a}}{a} + d + \dbmij{m}{b}{j}, \\
        \dbmij{m}{i}{a} + d + \dbmij{m}{b}{\bar{b}} + d + \dbmij{m}{\bar{a}}{j} \\
      \end{array}
    \right)
    $
    \EndFor
        \If{$\dbmij{m}{i}{i} < 0$}
        \State \Return{$false$}
	\EndIf
    \EndFor
    \State \Return{$\dbm{m}$}
    \EndFunction
  \end{algorithmic}
  \caption{In-place Incremental Closure}
  \label{fig:inplaceinsitu}
\end{figure}

\begin{corollary}
Suppose that $\dbm{m}$ is a closed DBM,
$\dbm{m}' = \IncrementalClosure(\dbm{m}, o)$, \linebreak
\mbox{$o = (x'_a - x'_b \leq d)$}
and $\dbm{m'}$ is consistent.  Then the following hold:
\begin{itemize}

\item $\dbmij{m'}{i}{j} \leq \dbmij{m'}{i}{a} + d + \dbmij{m'}{b}{j}$

\item $\dbmij{m'}{i}{j} \leq \dbmij{m'}{i}{\bar{b}} + d + \dbmij{m'}{\bar{a}}{j}$

\item $\dbmij{m'}{i}{j} \leq \dbmij{m'}{i}{\bar{b}} + d + \dbmij{m'}{\bar{a}}{a} + d + \dbmij{m'}{b}{j}$

\item $\dbmij{m'}{i}{j} \leq \dbmij{m'}{i}{a} + d + \dbmij{m'}{b}{\bar{b}} + d + \dbmij{m'}{\bar{a}}{j}$

\end{itemize}
\label{corollary:idempotence-facts}
\end{corollary}

\noindent The following theorem asserts that in-place update does not compromise correctness.  It is telling that the correctness
argument does not refer to the
entries $\dbmij{m}{i}{a}$, 
$\dbmij{m}{b}{\bari}$,
$\dbmij{m}{i}{\barb}$, 
$\dbmij{m}{\bara}{\bari}$,
$\dbmij{m}{\bara}{a}$
or
$\dbmij{m}{b}{\barb}$
at all.    This is because the corollary on which the theorem is founded
follows from the high-level property of idempotence. Notice too that the theorem
is parameterised by the traversal order over $\dbm{m}$ and therefore is independent of it. 

\begin{theorem}[\rm Correctness of \IncrementalClosureInSitu]
\label{lemma:insitu}
Suppose $\rho : \{ 0, \ldots, 2n - 1 \}^2 \to \{ 0, \ldots, 4n^2 - 1 \}$ is a bijective map,
$\dbm{m}$ is a closed DBM,
$\dbm{m'} = \IncrementalClosure(\dbm{m}, o)$, 
$o = (x'_a - x'_b \leq d)$, 
$\dbm{m}^{0} = \dbm{m}$ and
\[
\dbmij{m^{k+1}}{i}{j} =
\left\{
\begin{array}{rl}
 \dbmij{m^{k}}{i}{j} & \text{ if } \rho(i,j) \neq k \\[1.5ex]
\min \left( \begin{array}{l}
                                      \dbmij{m^{k}}{i}{j},\\
                                      \dbmij{m^{k}}{i}{a} + d +\dbmij{m^{k}}{b}{j}, \\
                                      \dbmij{m^{k}}{i}{\barb} + d + \dbmij{m^{k}}{\bara}{j}, \\
                                      \dbmij{m^{k}}{i}{a} + d + \dbmij{m^{k}}{b}{\barb} + d + \dbmij{m^{k}}{\bara}{j},\\
                                      \dbmij{m^{k}}{i}{\barb} + d + \dbmij{m^{k}}{\bara}{a} + d + \dbmij{m^{k}}{b}{j}
\end{array} \right) & \text{ if } \rho(i,j) = k
\end{array}
\right.
\]
Then either $\dbm{m'}$ is consistent and 
\begin{itemize}

\item
$\forall 0 \leq \ell < k. \dbm{m^{k}_{\rho^{-1}(\ell)}} = \dbm{m'_{\rho^{-1}(\ell)}}$ 

\item
$\forall k \leq \ell < 4n^2. \dbm{m^{k}_{\rho^{-1}(\ell)}} = \dbm{m_{\rho^{-1}(\ell)}}$ 

\end{itemize}
or $\dbm{m^{4n^2}}$ is inconsistent. 
\end{theorem}

\subsection{In-place Incremental Strong Closure}\label{sect:in-place-strong}

The in-place version of the incremental strong closure algorithm is
presented in Figure~\ref{fig:inplaceincrstrongclosure}.
\begin{figure}
  \centering
  \begin{algorithmic}[1]
        \Function{\IncrementalStrongClosureInSitu}{$\dbm{m}, x'_a - x'_b \leq d }$
        \For{$i \in \{ 0, \ldots, 2n-1 \}$}
        \State $\dbmij{m}{i}{\bari} \gets \min\left(
         \begin{array}{l}
           \dbmij{m}{i}{\bari}, \\
           \dbmij{m}{i}{a} + d + \dbmij{m}{b}{\bari}, \\
           \dbmij{m}{i}{\bar{b}} + d + \dbmij{m}{\bar{a}}{\bari},\\
           \dbmij{m}{i}{\bar{b}} + d + \dbmij{m}{\bar{a}}{a} + d + \dbmij{m}{b}{\bari}, \\
           \dbmij{m}{i}{a} + d + \dbmij{m}{b}{\bar{b}} + d + \dbmij{m}{\bar{a}}{\bari} \\
         \end{array}
       \right )$
       \EndFor

        \For{$i \in \{ 0, \ldots, 2n-1 \}$} 
        \For{$j \in \{ 0, \ldots, 2n-1 \}$}
                \If{$j \neq \bari$}
        \State $\dbmij{m}{i}{j} \gets \min\left (
          \begin{array}{l}
            \dbmij{m}{i}{j}, \\
            \dbmij{m}{i}{a}+ d + \dbmij{m}{b}{j},\\
            \dbmij{m}{i}{\bar{b}} + d + \dbmij{m}{\bar{a}}{j},\\
            \dbmij{m}{i}{\bar{b}} + d + \dbmij{m}{\bar{a}}{a} + d + \dbmij{m}{b}{j}, \\
            \dbmij{m}{i}{a} + d + \dbmij{m}{b}{\bar{b}} + d + \dbmij{m}{\bar{a}}{j}, \\
            (\dbmij{m}{i}{\bari} + \dbmij{m}{\barj}{j})/2
          \end{array}
        \right )$
        \EndIf
        \EndFor
        \If{$\dbmij{m}{i}{i} < 0$}
        \State \Return{$false$}
	\EndIf
	        \EndFor
    \State \Return{$\dbm{m}$}
        \EndFunction
      \end{algorithmic}
  \caption{In-place Incremental Strong Closure}
  \label{fig:inplaceincrstrongclosure}
\end{figure}
The following lemma shows that running incremental closure followed by
strengthening refines the entries in the DBM to their tightest
possible value with respect to the new octagonal constraint.

\begin{lemma}\label{lemma-str-close-idempotent} Suppose $\dbm{m}$ is a closed, coherent DBM and 
$\dbm{m'} = \IncrementalClosure(\dbm{m}, o)$,
$\dbm{m''} = \Strengthen(\dbm{m'})$,
$\dbm{m'''} = \IncrementalClosure(\dbm{m''}, o)$
and
$o = (x'_a - x'_b \leq d)$.
Then either $\dbm{m'}$ is consistent and $\dbm{m'''} = \dbm{m''}$ or $\dbm{m'''}$ is not consistent.
\end{lemma}

Now we move onto the theorem showing the correctness of
\IncrementalStrongClosureInSitu. We show that the in-place version of
the algorithm produces the same DBM as the non-in place version of the
algorithm. A bijective map used in the proof to process key entries
first before processing non-key entries: the condition
$\forall 0 \leq i < 2n. \rho(i,\bari) < 2n$ ensures this property.
Note that this is the only caveat on the order produced by the map:
the order in which key entries themselves are ordered is irrelevant
and similarly for non-key entries.
\begin{theorem}[\rm Correctness of \IncrementalStrongClosureInSitu] \label{lemma:stronginsitu} Suppose
  $\dbm{m}$ is a closed, coherent DBM,
  $\dbm{m'} = \IncrementalClosure(\dbm{m}, o)$, 
  $\dbm{m''} = \Strengthen(\dbm{m'})$, 
  $o = (x'_a - x'_b \leq d)$, 
  \mbox{$\rho : \{ 0, \ldots, 2n - 1 \}^2 \to \{ 0, \ldots, 4n^2 - 1 \}$} is a bijective map with \mbox{$\forall 0 \leq i < 2n . \rho(i,\bari) < 2n$},
  $\dbm{m}^{0} = \dbm{m}$ and
  \[
    \dbmij{m^{k+1}}{i}{j} =
    \left\{
      \begin{array}{ll}
        \dbmij{m^{k}}{i}{j} & \text{ if } \rho(i,j) \neq k \\[1.5ex]
        
        \min \left (\begin{array}{l}
                      \dbmij{m^{k}}{i}{\bari}, \\
                      \dbmij{m^{k}}{i}{a} + d + \dbmij{m^{k}}{b}{\bari}, \\
                      \dbmij{m^{k}}{i}{\bar{b}} + d + \dbmij{m^{k}}{\bar{a}}{\bari},\\
                      \dbmij{m^{k}}{i}{\bar{b}} + d + \dbmij{m^{k}}{\bar{a}}{a} + d + \dbmij{m^{k}}{b}{\bari}, \\
                      \dbmij{m^{k}}{i}{a} + d + \dbmij{m^{k}}{b}{\bar{b}} + d + \dbmij{m^{k}}{\bar{a}}{\bari} \\
                    \end{array} \right ) & \text{ if } \rho(i,j) = k \wedge j = \bari \\[1.5ex]
        \\
        \min \left( \begin{array}{l}
                      \dbmij{m^{k}}{i}{j},\\
                      \dbmij{m^{k}}{i}{a} + d +\dbmij{m^{k}}{b}{j}, \\
                      \dbmij{m^{k}}{i}{\barb} + d + \dbmij{m^{k}}{\bara}{j}, \\
                      \dbmij{m^{k}}{i}{a} + d + \dbmij{m^{k}}{b}{\barb} + d + \dbmij{m^{k}}{\bara}{j},\\
                      \dbmij{m^{k}}{i}{\barb} + d + \dbmij{m^{k}}{\bara}{a} + d + \dbmij{m^{k}}{b}{j}, \\
                      (\dbmij{m^{k}}{i}{\bari} + \dbmij{m^{k}}{\barj}{j})/2
                    \end{array} \right) & \text{ if } \rho(i,j) = k \wedge j \neq \bari
      \end{array}
    \right.
  \]
  Then either $\dbm{m'}$ is consistent and 
  \begin{itemize}
    
  \item
    $\forall 0 \leq \ell < k. \dbm{m^{k}_{\rho^{-1}(\ell)}} = \dbm{m''_{\rho^{-1}(\ell)}}$ 
    
  \item
    $\forall k \leq \ell < 4n^2. \dbm{m^{k}_{\rho^{-1}(\ell)}} = \dbm{m_{\rho^{-1}(\ell)}}$ 
    
  \end{itemize}
  or $\dbm{m^{4n^2}}$ is inconsistent. 
\end{theorem}

\subsection{In-place Incremental Tight Closure}\label{sect:in-place-integer}

\begin{figure}[t]
  \centering
  \begin{algorithmic}[1]
        \Function{\IncrementalIntegerClosureInSitu}{$\dbm{m}, x'_a - x'_b \leq d }$
        \For{$i \in \{ 0, \ldots, 2n-1 \}$}
       \State $\dbmij{m}{i}{\bari} \gets 2 \left \lfloor \min\left(
         \begin{array}{l}
           \dbmij{m}{i}{\bari}, \\
           \dbmij{m}{i}{a} + d + \dbmij{m}{b}{\bari}, \\
           \dbmij{m}{i}{\bar{b}} + d + \dbmij{m}{\bar{a}}{\bari},\\
           \dbmij{m}{i}{\bar{b}} + d + \dbmij{m}{\bar{a}}{a} + d + \dbmij{m}{b}{\bari}, \\
           \dbmij{m}{i}{a} + d + \dbmij{m}{b}{\bar{b}} + d + \dbmij{m}{\bar{a}}{\bari} \\
         \end{array}
       \right ) / 2 \right \rfloor $
       \EndFor
        \If{\CheckIntegerConsistent($\dbm{m'}$)}
        \For{$i \in \{ 0, \ldots, 2n-1 \}$} 
        \For{$j \in \{ 0, \ldots, 2n-1 \}$}
                \If{$j \neq \bari$}
        \State $\dbmij{m}{i}{j} \gets \min\left (
          \begin{array}{l}
            \dbmij{m}{i}{j}, \\
            \dbmij{m}{i}{a}+ d + \dbmij{m}{b}{j},\\
            \dbmij{m}{i}{\bar{b}} + d + \dbmij{m}{\bar{a}}{j},\\
            \dbmij{m}{i}{\bar{b}} + d + \dbmij{m}{\bar{a}}{a} + d + \dbmij{m}{b}{j}, \\
            \dbmij{m}{i}{a} + d + \dbmij{m}{b}{\bar{b}} + d + \dbmij{m}{\bar{a}}{j}, \\
            (\dbmij{m}{i}{\bari} + \dbmij{m}{\barj}{j})/2
          \end{array}
        \right )$
        \EndIf
        \EndFor
        \If{$\dbmij{m}{i}{i} < 0$}
        \State \Return{$false$}
	\EndIf
	\EndFor
	\Else
        \State \Return{$false$}
	\EndIf
    \State \Return{$\dbm{m}$}
        \EndFunction
      \end{algorithmic}
      \caption{In-place Incremental Tight Closure}
      \label{fig:inplaceincrtightclosure}
\end{figure}

The in-place version of the incremental tight closure algorithm is
presented in Figure~\ref{fig:inplaceincrstrongclosure}, the only
difference with incremental strong closure is that for key entries we
also run a tightening step (line~3).
As in the previous section, we have a helper lemma for the main
theorem, showing that incremental closure followed by tightening and
strengthening refines the entries in the DBM to the tightest value
with respect to the new octagonal constraint.

\begin{lemma}
  \label{lemma:tightencloseidempotent}
  Suppose $\dbm{m}$ is a closed, coherent DBM and
  $\dbm{m'} = \IncrementalClosure(\dbm{m},o)$,
  $\dbm{m''} = \Tighten(\dbm{m'},o)$,
  $\dbm{m'''} = \Strengthen(\dbm{m''},o)$,
  $\dbm{m^{*}} = \IncrementalClosure(\dbm{m'''},o)$
  and \linebreak
  $\dbm{m^{*}} = \dbm{m'''}$ or $\dbm{m^{*}}$ is inconsistent.
\end{lemma}

The following theorem is analogous to the theorem for in-place strong closure:
\begin{theorem}[\rm Correctness of
  \IncrementalIntegerClosureInSitu] \label{lemma:tightinsitu} Suppose
  $\dbm{m}$ is a closed, coherent DBM,
  $\dbm{m'} = \IncrementalClosure(\dbm{m}, o)$,
  $\dbm{m''} = \Tighten(\dbm{m'})$,
  $\dbm{m'''} = \Strengthen(\dbm{m'})$, $o = (x'_a - x'_b \leq d)$,
  that $\rho : \{ 0, \ldots, 2n - 1 \}^2 \to \{ 0, \ldots, 4n^2 - 1 \}$ is
  a bijective map with
  \mbox{$\forall 0 \leq i < 2n. \rho(i,\bari) < 2n$},
  $\dbm{m}^{0} = \dbm{m}$ and 
  \[
    \dbmij{m^{k+1}}{i}{j} =
    \left\{
      \begin{array}{ll}
        \dbmij{m^{k}}{i}{j} & \text{ if } \rho(i,j) \neq k \\[1.5ex]
        
        2 \left\lfloor \min \left (
        \begin{array}{l}
          \dbmij{m^{k}}{i}{\bari}, \\
          \dbmij{m^{k}}{i}{a} + d + \dbmij{m^{k}}{b}{\bari}, \\
          \dbmij{m^{k}}{i}{\bar{b}} + d + \dbmij{m^{k}}{\bar{a}}{\bari},\\
          \dbmij{m^{k}}{i}{\bar{b}} + d + \dbmij{m^{k}}{\bar{a}}{a} + d + \dbmij{m^{k}}{b}{\bari}, \\
          \dbmij{m^{k}}{i}{a} + d + \dbmij{m^{k}}{b}{\bar{b}} + d + \dbmij{m^{k}}{\bar{a}}{\bari} \\
        \end{array}
        \right )/2 \right\rfloor & \text{ if } \rho(i,j) = k \wedge j = \bari \\[1.5ex]\\
        \min \left( \begin{array}{l}
                      \dbmij{m^{k}}{i}{j},\\
                      \dbmij{m^{k}}{i}{a} + d +\dbmij{m^{k}}{b}{j}, \\
                      \dbmij{m^{k}}{i}{\barb} + d + \dbmij{m^{k}}{\bara}{j}, \\
                      \dbmij{m^{k}}{i}{a} + d + \dbmij{m^{k}}{b}{\barb} + d + \dbmij{m^{k}}{\bara}{j},\\
                      \dbmij{m^{k}}{i}{\barb} + d + \dbmij{m^{k}}{\bara}{a} + d + \dbmij{m^{k}}{b}{j}, \\
                      (\dbmij{m^{k}}{i}{\bari} + \dbmij{m^{k}}{\barj}{j})/2
                    \end{array} \right) & \text{ if } \rho(i,j) = k \wedge j \neq \bari
      \end{array}
    \right.
  \]
  Then either $\dbm{m'}$ is consistent and 
  \begin{itemize}
    
  \item
    $\forall 0 \leq \ell < k. \dbm{m^{k}_{\rho^{-1}(\ell)}} = \dbm{m'''_{\rho^{-1}(\ell)}}$ 
    
  \item
    $\forall k \leq \ell < 4n^2. \dbm{m^{k}_{\rho^{-1}(\ell)}} = \dbm{m_{\rho^{-1}(\ell)}}$ 
    
  \end{itemize}
  or $\dbm{m^{4n^2}}$ is inconsistent. 
\end{theorem}


%% file: experiments.tex

\section{Experimental Evaluation}
\label{sec:experiments}

For a fair and robust evaluation, the algorithms were
implemented using machinery provided in the
Apron library \cite{mine-CAV09}. The library provides 
implementations of the box, polyhedra and octagon abstract domains, the latter used for purposes of comparison.  Apron
is realised in C, and provides bindings for OCaml, C++ and Java. 
 \IncrementalClosure\/ and  \IncrementalStrongClosure\/
where then compared against the optimised implementation of incremental
closure provided by Apron. Three sets of experiments were performed.  First, the closure
algorithms were applied to randomly generated DBMs, subject to various
size constraints, to systematically exercise the algorithms on a range
of problem size. Henceforth these randomly generated problems will be
referred to as the micro-benchmarks. Second, to investigate the performance of
the algorithms in a real-world setting, the algorithms were integrated
into Frama-C, which is an industrial-strength static analysis tool for C code.
The tool was then applied to a
collection of C programs drawn from the Frama-C benchmarks repository. Third,
the algorithms were integrated
into AbSolute solver \cite{pelleau13constraint} and evaluated
against benchmarks drawn from
continuous constraint
programming.   

All experiments were performed on a
32-core Intel Xeon workstation with 128GB of memory. 

\subsection{Apron Library}

\newcommand{\svdots}{\raisebox{3pt}{\scalebox{.75}{\vdots}}}
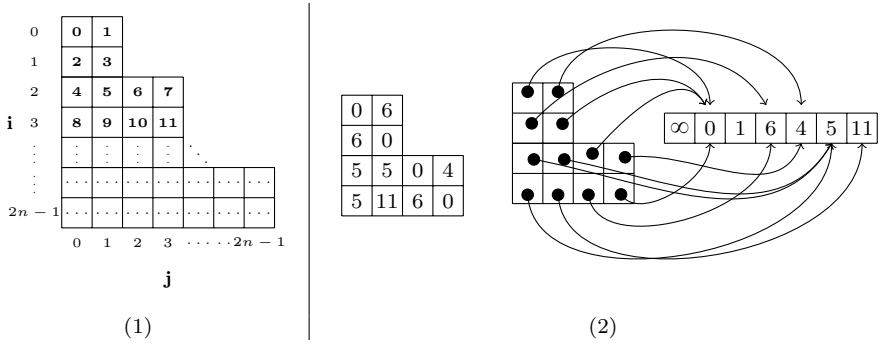
\begin{figure}[!t]
  \begin{tabular}{c|cc}
    \begin{tabular}{c@{}}
 \begin{tikzpicture}[scale=0.4]
   \draw (0,0) grid (7,2);
   \draw (0,2) grid (4,5);
   \draw (0,5) grid (2,7);

   \node at (-1.7, 3.5) {\scriptsize $\bf{i}$};
   \node at (3.5, -1.7) {\scriptsize $\bf{j}$};
   \foreach \c [count=\x from 0] in {{0}, {1}, {2}, {3}, {\ldots}, {\ldots}, {2n-1}} \node at (\x+0.5,-0.5) {\tiny $\c$};
   \foreach \c [count=\x from 0] in {{2n-1},{\svdots}, {\svdots}, {3}, {2}, {1}, {0}} \node at (-0.9, \x+0.5) {\tiny $\c$};

   \foreach \c [count=\x from 0] in {{0}, {2}, {4}, {8}, {\svdots}, {\ldots}, {\ldots}} \node at (0.5, 6.5-\x) {\tiny $\bf{\c}$};
   \foreach \c [count=\x from 0] in {{1}, {3}, {5}, {9}, {\svdots}, {\ldots}, {\ldots}} \node at (1.5, 6.5-\x) {\tiny $\bf{\c}$};
   \foreach \c [count=\x from 0] in {{6}, {10}, {\svdots}, {\ldots}, {\ldots}} \node at (2.5, 4.5-\x) {\tiny $\bf{\c}$};
   \foreach \c [count=\x from 0] in {{7}, {11}, {\svdots}, {\ldots}, {\ldots}} \node at (3.5, 4.5-\x) {\tiny $\bf{\c}$};
   \foreach \c [count=\x from 0] in {{\ldots}, {\ldots}} \node at (4.5, 1.5-\x) {\tiny $\bf{\c}$};
   \foreach \c [count=\x from 0] in {{\ldots}, {\ldots}} \node at (5.5, 1.5-\x) {\tiny $\bf{\c}$};
   \foreach \c [count=\x from 0] in {{\ldots}, {\ldots}} \node at (6.5, 1.5-\x) {\tiny $\bf{\c}$};
   \node at (4.5, 2.7) {\tiny $\ddots$};
 \end{tikzpicture} 
    \end{tabular}
    &
      \begin{tabular}{c@{}}
 \begin{tikzpicture}[scale=0.4]
   \draw (0,0) grid (4,2);
   \draw (0,2) grid (2,4);

   \foreach \c [count=\x from 0] in {{0},{6},{5},{5}} \node at (0.5, 3.5-\x) {$\c$};
   \foreach \c [count=\x from 0] in {{6},{0},{5},{11}} \node at (1.5, 3.5-\x) {$\c$};
   \foreach \c [count=\x from 0] in {{0},{6}} \node at (2.5, 1.5-\x) {$\c$};
   \foreach \c [count=\x from 0] in {{4},{0}} \node at (3.5, 1.5-\x) {$\c$};
   
 \end{tikzpicture}
        \end{tabular}
    &
      \begin{tabular}{c@{}}
  \begin{tikzpicture}[scale=0.4]
   \draw (0,0) grid (4,2);
   \draw (0,2) grid (2,4);

   \draw (5,2) grid (12,3);
   \foreach \c [count=\x from 0] in {{\infty}, {0}, {1}, {6}, {4}, {5}, {11}} \node at (\x + 5.5, 2.5) {$\c$};
   \node at (5.5,3) (0) {};
   \node at (6.5,3) (1) {};
   \node at (7.5,3) (2) {};
   \node at (8.5,3) (3) {};
   \node at (9.5,3) (4) {};
   \node at (10.5,3) (5) {};
   \node at (11.5,3) (6) {};
   
   \draw[*->] (0.5,3.5) to [out=90, in=90] (1);
   \draw[*->] (1.5,3.5) to [out=90, in=90] (4);
   \draw[*->] (0.5,2.5) to [out=50, in=120] (3);
   \draw[*->] (1.5,2.5) to [in=120] (1);

   \draw[*->] (0.5,1.5) to [out=-10, in=-120] (10.5,2);
   \draw[*->] (1.5,1.5) to [out=-10, in=-120] (10.5,2);
   \draw[*->] (2.5,1.5) to [out=50, in=120] (1);
   \draw[*->] (3.5,1.5) to [out=10, in=-110] (9.5,2);

   \draw[*->] (0.5, 0.5) to [out=-90, in=-90] (10.5,2);
   \draw[*->] (1.5, 0.5) to [out=-90, in=-90] (11.5,2);
   \draw[*->] (2.5, 0.5) to [out=-90, in=-90] (8.5,2);
   \draw[*->] (3.5, 0.5) to [out=-70, in=-90] (6.5,2);
 \end{tikzpicture}
      \end{tabular} \\
   
    (1) & \multicolumn{2}{c}{(2)} \\
  \end{tabular}
   \caption{(1) Representing a DBM as an array; (2) Representing a DBM as a CoDBM}
   \label{fig:cdbms}
 \end{figure}

 The Apron library \cite{mine-CAV09} supports various number systems,
 such as single precision floating-point numbers and GNU
 multiple-precision (GMP) rationals. The default number system for the
 OCaml bindings is rationals, but it must be appreciated that the
 computational overhead of allocating memory for the rationals dominates the
 runtime, potentially masking the benefits of
 \IncrementalStrongClosure\/ over \IncrementalClosure\/. (Recall that
 \IncrementalStrongClosure\/ saves a separate pass over the DBM
 relative to \IncrementalClosure\/, avoiding counter increments and
 integer comparisons.)
 
 In Apron, numbers are represented by a type \texttt{bound\_t}, which
 depending on compile-time options, will select a specific header file
 containing concrete implementations of operations involving numbers
 extended to the symbolic values of $-\infty$ and $+\infty$. Every
 \texttt{bound\_t} object has to be initialised via a call to
 \texttt{bound\_init}, which in the case of rationals will invoke
 \texttt{malloc} and initialise space for the rational number. DBMs
 are stored taking advantage of the half-matrix nature of octagonal
 DBMs which follows by the definition of coherence. An array of
 \texttt{bound\_t} objects is then used to represent the half-matrix,
 as shown in figure~\ref{fig:cdbms}, subfigure (1). If $i\geq j$ or
 $i=\barj$ then the entry at $(i,j)$ in the DBM is stored at index
 $j + \lfloor i^2 / 2\rfloor$ in the array. Otherwise $(i,j)$ is
 stored at the array element reserved for entry $(\barj, \bari)$. A
 DBM of size $n$ requires an array of size $2n(n+1)$ which gives a
 significant space reduction over a naive representation of size
 $4n^2$.
 
 \subsection{Compact DBMs}

 Unexpectedly, initial experiments with Frama-C suggested that much of
 its runtime was spent in memory management rather than the domain
 operations themselves. Further investigation using Callgrind showed
 that 36\% of all function calls emanated from \texttt{malloc}-like routines.
 In response, the underlying DBM data structure was refactored to ensure that this
 undesirable memory management feature did not
 artificially perturb the experiments. The
 refactoring is fully described in a separate work
 \cite{ChawdharyKAPLAS17}, but to keep the paper self-contained the
 main idea is summarised in the following paragraph.

 The DBM representation was changed from a matrix storing numbers to a
 matrix storing pointers to numbers stored in a cache. This reduces
 the amount of memory used by a DBM as shown in
 figure~\ref{fig:cdbms}. The modified data structure has been dubbed a
 compact DBM or CoDBM for short. The cache is an array initialised to
 contain $\infty$ as its first entry, augmented with an ordered table
 which enables the pointer for any given number to be found (if it
 exists) in the cache using the bisection search method. As new
 numbers are created they are added to the cache, and the table is
 extended in sync. This representation which, crucially, factors out
 the overhead of storing a number repeatedly, has a significant impact
 on the memory usage of the Apron library. It also rebalances the
 proportion of time spend in domain operations. Further performance
 debugging of Frama-C, for instance to speed up parsing, would only
 increase the fraction of time spend on the domain operations and
 closure in particular.

\subsection{Micro-benchmarks}

\begin{figure}
\centering
\begin{tabular}{@{}c@{}}
\includegraphics[width=\linewidth]{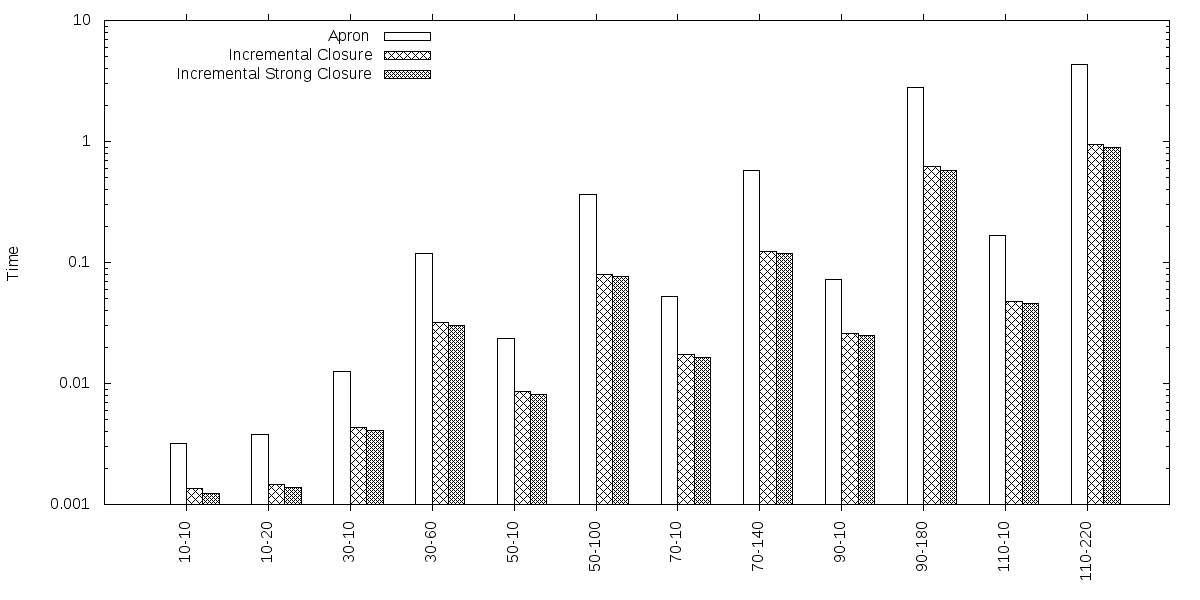}
\end{tabular}
\caption{Micro-benchmark timings for rationals}\label{fig:microbenchmarks}
\end{figure}

Each micro-benchmark suite was a collection of 10 problems, each
problem consisting of a random octagon and a randomly generated
octagonal constraint. Each random octagon was generated from a
prescribed number of octagonal constraints, so as to always contain
the origin, for a given number for variables. Each octagon was then
closed. A single randomly generated octagonal constraint, not
necessarily containing the origin, was then added to the closed
octagon using incremental closure. \IncrementalClosure\/ and
\IncrementalStrongClosure\/ where then timed and compared against the
Apron version for DBMs over rationals. The resulting DBMs
were then all checked for equality against \NonIncrementalClosure\/.
All timings were averaged over 10 runs and, moreover, all algorithms
were exercised on exactly the same collection of problems.
Fig.~\ref{fig:microbenchmarks} presents timings for the
micro-benchmark suites. The results show that \IncrementalClosure\/ outperforms the original
Apron implementation by a factor of 3--4 and 
\IncrementalStrongClosure\/ offers an additional 4--9\% speedup over \IncrementalClosure\/.

\subsection{Frama-C Benchmarks}

\begin{table}[t!]
\begin{center}
  \begin{tabular}{l|l|r|l}
    Name   & Benchmark         & LOC   & Description                                       \\
    \hline
    lev     & levenstein        & 187   & Levenstein string distance library  \\
    sol     & solitaire         & 334   & card cipher                   \\
    2048    & 2048              & 435   & 2048 game                                         \\
    kh      & khash             & 652   & hash code from klib C library                     \\
    taes    & Tiny-AES          & 813   & portable AES-128 implementation            \\
    qlz     & qlz               & 1168  & fast compression library                          \\
    mod     & libmodbus         & 7685  & library to interact with Modbus protocol          \\
   mgmp     & mini-gmp          & 11787 & subset of GMP library                     \\
    unq     & unqlite           & 64795 & embedded NoSQL DB                                 \\
    bzip    & bzip-single-file  & 74017 & bzip single file for static analysis benchmarking 
  \end{tabular}
\caption{Benchmark suite of C programs}
\label{tbl:benchmarks}
\end{center}
\end{table}

\begin{figure}[t!]
  \includegraphics[width=\linewidth]{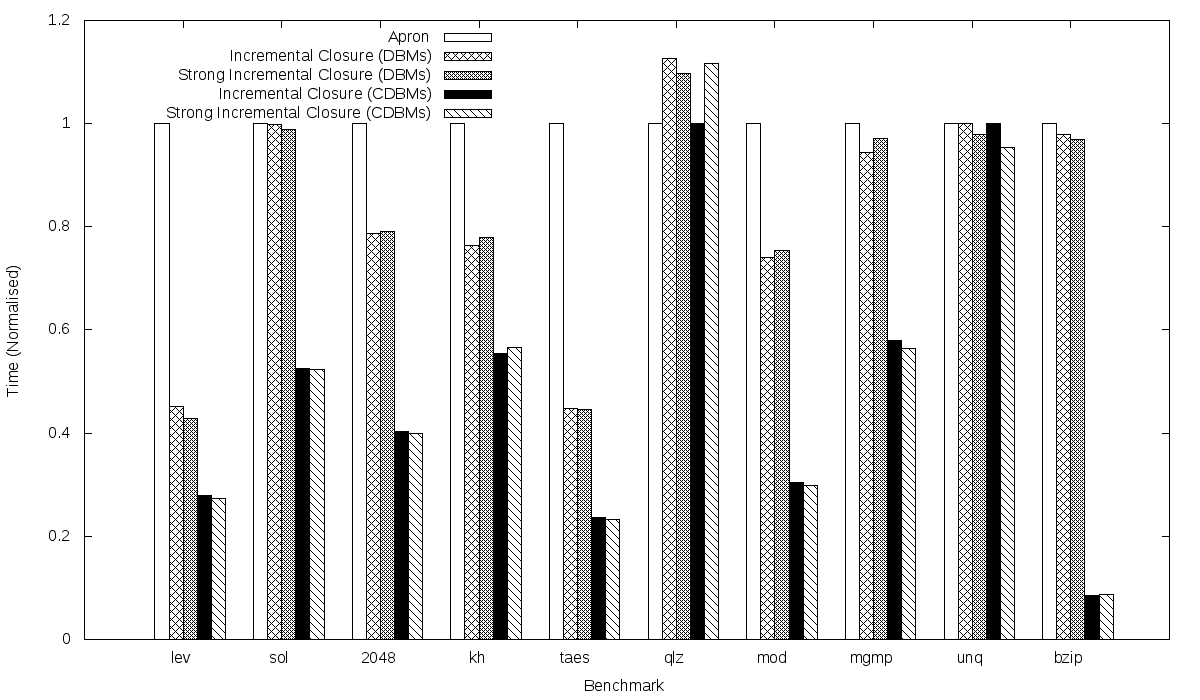}
  \includegraphics[width=\linewidth]{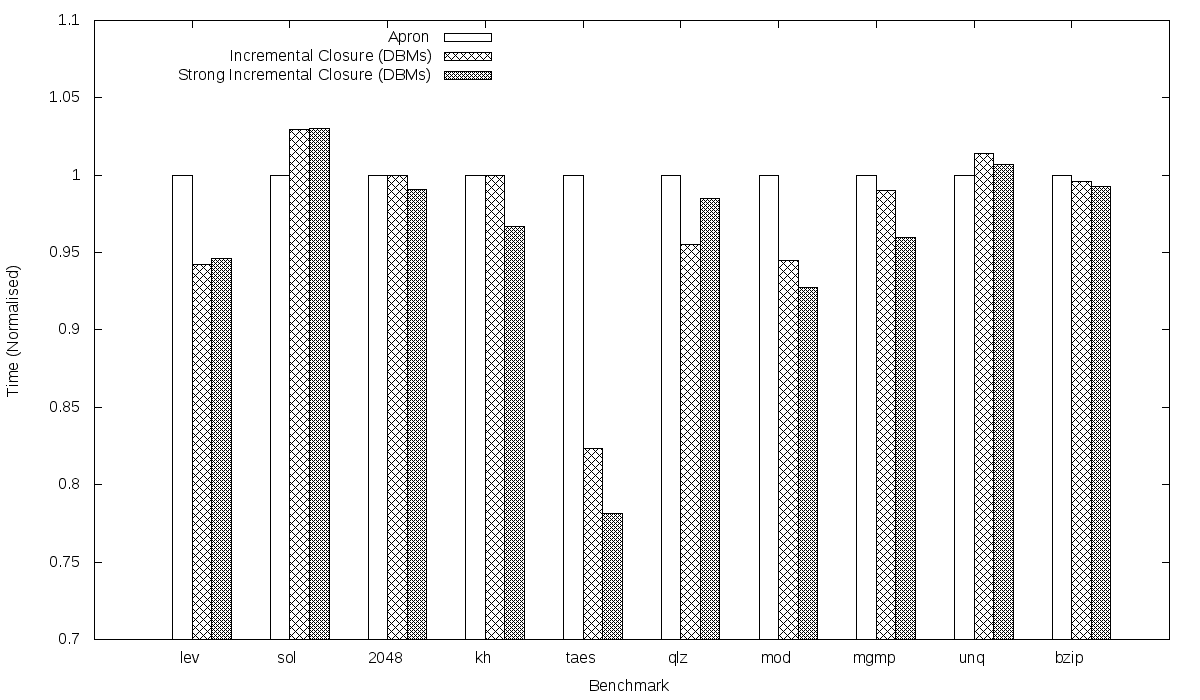}
\caption{Normalised timings of Frama-C for rationals (above) and floating point (below)}
\label{fig:normalised-timings}
\end{figure}

The EVA plugin of Frama-C implements an abstract interpreter over the
internal intermediate language used by Frama-C. The plugin uses the
Apron library to perform an octagon domain analysis, and so by
modifying Apron, Frama-C can make direct
use of \IncrementalClosure\/ and \IncrementalStrongClosure\/ (and specially their
\IncrementalClosureHoisting\/ and \IncrementalStrongClosureReduce\/ variants).
Table~\ref{tbl:benchmarks} lists the benchmark programs
passed to EVA to interpret the programs over the octagon domain. It should be noted
that EVA is a prototype
(which may explain its memory behaviour) and as such does not use widely
used heuristics and optimisations such as variable packing \cite{blanchet03static,heo16learning} or
localisation techniques \cite{beckschulze12access,oh11access} to enable the analysis to scale. Nonetheless,
the octagon analysis successfully terminated over the selected benchmarks.

Figure~\ref{fig:normalised-timings} gives the timings of benchmarks for rational (above) and floating point arithmetic (below),
normalised to the time required by the Apron implementation.  For rationals, normalised
timings are given for both DBMs and CoDBMs.  The relative speedup obtained from deploying
\IncrementalClosure\/ and \IncrementalStrongClosure\/ over Apron algorithm is variable,
ranging from a large speedup for taes to a modest slowdown for qlz.  
Table~\ref{table:unnormalised-timings} amplifies
the relative timings presented in the bar chart, giving the exact timings in seconds. The table shows that the longest
running analyses (which correspond to those employing the largest DBMs) are best served by
\IncrementalClosure\/ and \IncrementalStrongClosure\/.  

Cachegrind \cite{nethercote04dynamic} profiling
sheds light on qlz: some of the refined incremental algorithms actually increase
the number of first-level data cache misses, giving a net slowdown.  This cache anomaly might
arise because the DBMs generated by qlz are tiny. Cachegrind also suggests this
is the exception, revealing that the large speedup on bzip, mod and taes for
CoDBMs over DBMs stems from a reduction in the number of misses to level 3 unified data and instruction cache.
In fact, for bzip, mod and taes, the number of level 3 cache misses is reduced to zero.  This validates
the CoDBM data-structure.  It also illustrates that optimisations which match the architecture
can have surprising impact.

Floating point arithmetic is much faster than rationals, so the proportion
of the overall execution time spent in closure is decreased, hence one would expect the relative speedup from
\IncrementalClosure\/ and \IncrementalStrongClosure\/ over Apron to be likewise reduced.
Figure~\ref{fig:normalised-timings} and
table~\ref{table:unnormalised-timings} shows that this is the general pattern.  CoDBMs timings are
not given for floats because floats have a much denser representation than GMP rationals. Nevertheless,
the longest running analysis, which arises on taes, significantly benefits from both
\IncrementalClosure\/ and \IncrementalStrongClosure\/.  


\begin{table}[t!]
\centering
\begin{tabular}{c}
\begin{tabular}{@{}l|r|rr|rr@{}}
&   & \multicolumn{2}{c|}{DBM} & \multicolumn{2}{c}{CoDBM}                                                              \\
  Benchmark& Apron    & \IncrementalClosure     & \IncrementalStrongClosure & \IncrementalClosure & \IncrementalStrongClosure \\ \hline           
    lev    & 33.16    & 14.98                   & 14.21                     & 9.23                & 9.05                      \\
    sol    & 49.80    & 49.76                   & 49.19                     & 26.17               & 26.03                     \\
    2048   & 33.16    & 26.10                   & 26.26                     & 13.39               & 13.23                     \\
    kh     & 1.80     & 1.37                    & 1.40                       & 1.00                & 1.02                      \\
    taes   & 1817.91  & 814.77                  & 810.00                    & 430.60              & 421.32                    \\
    qlz    & 1.08     & 1.21                    & 1.18                      & 1.08                & 1.20                      \\
    mod    & 463.46   & 343.05                  & 349.62                    & 141.17              & 138.60                    \\
    mgmp   & 2.09     & 1.97                    & 2.03                      & 1.21                & 1.18                      \\
    unq    & 1.49     & 1.49                    & 1.46                      & 1.49                & 1.42                      \\
    bzip   & 621.53   & 607.88                  & 602.78                    & 53.51               & 52.63                     \\ \hline
    cumulative & 3025.48 & 1862.58 & 1858.13 & 678.85 & 665.68 \\ 
\multicolumn{6}{c}{} \\
            &        & \multicolumn{2}{c|}{DBM} \\
  Benchmark & Apron  & \IncrementalClosure      & \IncrementalStrongClosure \\ \cline{1-4}
lev         & 2.61   & 2.46                     & 2.47                      \\ 
sol         & 12.62  & 12.99                    & 13.00                    \\ 
2048        & 4.48   & 4.48                     & 4.44                     \\ 
kh          & 0.60   & 0.60                     & 0.58                      \\ 
taes        & 113.26 & 93.26                    & 88.47                  \\ 
qlz         & 1.35   & 1.29                     & 1.33                      \\ 
mod         & 57.59  & 54.43                    & 53.41                 \\
mgmp        & 1.00   & 0.99                     & 0.96                    \\
unq         & 1.44   & 1.46                     & 1.45                     \\
bzip        & 22.69  & 22.60                    & 22.52                 \\
\cline{1-4}
cumulative  &  217.64  & 194.56                   & 188.63        
\end{tabular}
\end{tabular}
  \caption{Absolute timings of Frama-C for rationals (above) and floating point (below)}
  \label{table:unnormalised-timings}
\end{table}

\subsection{AbSolute Constraint Solver Benchmarks}

\begin{figure}[t!]
  \includegraphics[width=\linewidth]{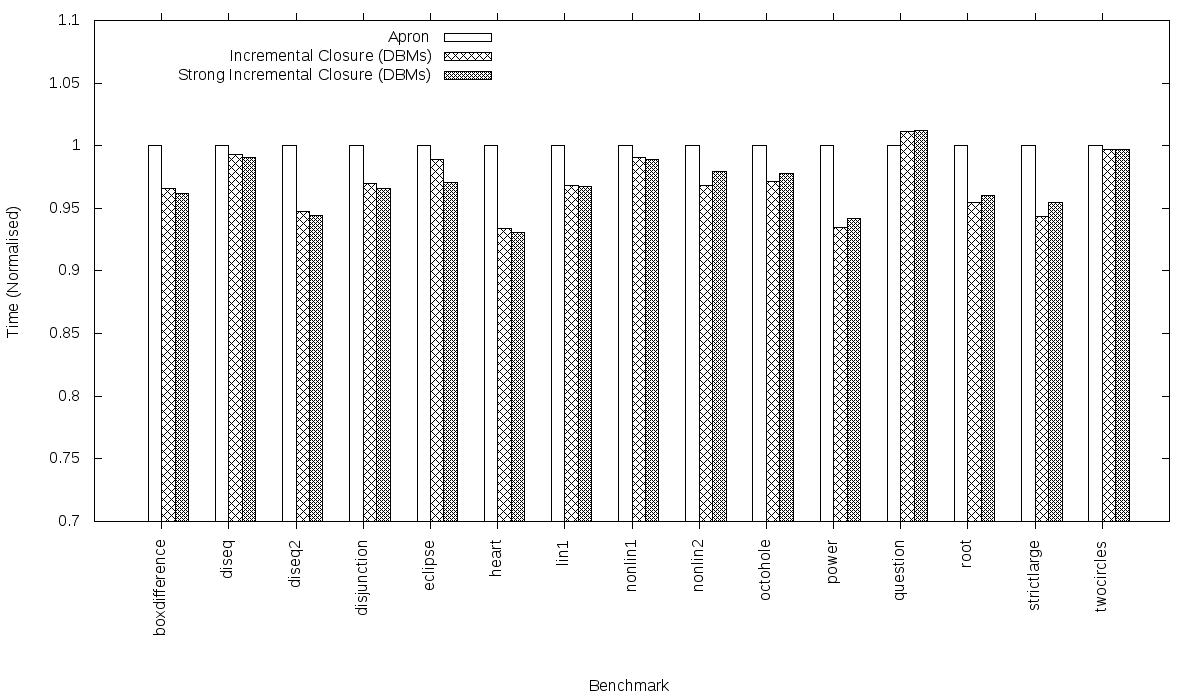}
\caption{Normalised timings for the AbSolute constraint solver (doubles)}
\label{fig:absolute-normalised-timings}
\end{figure}

\begin{table}[t!]
\centering
\begin{tabular}{l|r|r|r}
Benchmark     & Apron   & \IncrementalClosure & \IncrementalStrongClosure \\
\hline
boxdifference & 8.72    & 8.42                & 8.39                      \\
diseq         & 18.25   & 18.11               & 18.07                     \\
diseq2        & 15.62   & 14.80               & 14.75                     \\
disjunction   & 4.30    & 4.17                & 4.15                      \\
eclipse       & 42.39   & 41.91               & 41.14                     \\
heart         & 1014.13 & 947.32              & 944.04                    \\
lin1          & 12.15   & 11.77               & 11.76                     \\
nonlin1       & 2.38    & 2.35                & 2.35                      \\
nonlin2       & 4.05    & 3.92                & 3.97                      \\
octo\_hole    & 2.11    & 2.05                & 2.06                      \\
power         & 24.57   & 22.97               & 23.13                     \\
question      & 5.30    & 5.36                & 5.37                      \\
root          & 13.53   & 12.92               & 13.00                     \\
strict\_large & 4.52    & 4.26                & 4.31                      \\
two\_circles  & 11.99   & 11.95               & 11.95                     \\
  \hline
cumulative    & 1192.38 & 1112.28             & 1108.44                   \\
\end{tabular}
\caption{Absolute timing for the AbSolute constraint solver (doubles)}
\label{tbl:absolute}
\end{table}

The AbSolute constraint solver \cite{pelleau13constraint} applies
principles from abstract interpretation to continuous constraint
programming. Continuous constraint programming uses interval
approximations to approximate solutions to continuous constraints: in
essence a solution enclosed by a single interval is successively
refined to a set of intervals covering the solution (provided one
exists).

The AbSolute solver deploys octagons rather than intervals to obtain a
more precise and scalable solver. It uses Apron to implement its
abstract domain operations, working over floats rather than rationals.
The benchmarks selected to exercise AbSolute are a strict subset of
those contained in the AbSolute repository (some problems fail to
parse while others contain trigonometric functions not supported by
the Apron library).

Figure~\ref{fig:absolute-normalised-timings} summaries the relative
performance of Apron, \IncrementalClosure\/ and
\IncrementalStrongClosure\/; Table~\ref{tbl:absolute} gives the exact
timings in seconds. All but one benchmarks show an improvement with
\IncrementalClosure\/ and \IncrementalStrongClosure\/, even though the
size of the DBMs are small compared to those that arise in the Frama-C
benchmarks.



%% file: conclusion.tex

\section{Concluding Discussion}
\label{sec:conclusion}

The octagon domain is used for many applications due to its
expressiveness and its easy of implementation, relative to other relational
abstract domains. Yet the elegance of their domain operations is at
odds with the subtlety of the underlying ideas, and the reasoning
needed to justify refinements that appear to be straightforward, such
as tightening and in-place update.

This paper has presented novel algorithms to incrementally update an
octagonal constraint system. More specifically, we have developed new
incremental algorithms for closure, strong closure and integer
closure, and their in-place variants. Experimental results with a
prototype implementation demonstrate significant speedups
over existing closure algorithms. We leave as future work the
generalisation of the in-place update results to parallel evaluation
and the application of our incremental algorithms for modelling machine
arithmetic \cite{simon_taming_2007} in binary analysis which,
incidentally, was the problem that motivated this thread of work.


%% file: acks.tex

\paragraph{Acknowledgements}
We thank Jacques-Henri~Jourdan for
stimulating discussions on Verasco at POPL'16 in St. Petersburg and
the ``Verified Trustworthy Software Systems'' Royal Society event in London.


%% file: proofs.tex

\section{Proof Appendix}\label{sec:proofs}

\input{proofs-incrementalclosure}
\input{proofs-coherence}
\input{proofs-strongclosure}
\input{proofs-integerclosure}
\input{proofs-inplace}

%% file: proofs-incrementalclosure.tex

\subsection{Proofs for the Correctness of Incremental Closure}

\begin{proof}[for lemma~\ref{lemma:nonzero}]
We first prove the if case: since $\dbm{m'}$ is consistent $\dbmij{m'}{\bar{a}}{\bar{a}} \geq 0$ hence
\[
\min \left(
      \begin{array}{l}
        \dbmij{m}{\bar{a}}{\bar{a}}, \\
        \dbmij{m}{\bar{a}}{a}+ d + \dbmij{m}{b}{\bar{a}},\\
        \dbmij{m}{\bar{a}}{\bar{b}} + d + \dbmij{m}{\bar{a}}{\bar{a}},\\
        \dbmij{m}{\bar{a}}{\bar{b}} + d + \dbmij{m}{\bar{a}}{a} + d + \dbmij{m}{b}{\bar{a}}, \\
        \dbmij{m}{\bar{a}}{a} + d + \dbmij{m}{b}{\bar{b}} + d + \dbmij{m}{\bar{a}}{\bar{a}} \\
      \end{array}
    \right) = \dbmij{m'}{\bar{a}}{\bar{a}} \geq 0
\] 
Therefore
$\dbmij{m}{\bar{a}}{\bar{b}} + d + \dbmij{m}{\bar{a}}{\bar{a}} \geq 0$
and
$\dbmij{m}{\bar{a}}{a} + d + \dbmij{m}{b}{\bar{b}} + d + \dbmij{m}{\bar{a}}{\bar{a}} \geq 0$.
Since $\dbm{m}$ is closed $\dbmij{m}{\bar{a}}{\bar{a}} = 0$ hence
$\dbmij{m}{\bar{a}}{\bar{b}} + d \geq 0$ 
and
$\dbmij{m}{\bar{a}}{a} + d + \dbmij{m}{b}{\bar{b}} + d \geq 0$.

Repeating the argument $\dbmij{m'}{b}{b} \geq 0$ hence
\[
      \min\left (
      \begin{array}{l@{\hspace{0pt}}l@{\hspace{0pt}}l}
        &\dbmij{m}{b}{b},\\
        & \dbmij{m}{b}{a}+ d + \dbmij{m}{b}{b}, \\
        & \dbmij{m}{b}{\bar{b}} + d + \dbmij{m}{\bar{a}}{b},\\
        & \dbmij{m}{b}{\bar{b}} + d + \dbmij{m}{\bar{a}}{a} + d + \dbmij{m}{b}{b}, \\
        & \dbmij{m}{b}{a} + d + \dbmij{m}{b}{\bar{b}} + d + \dbmij{m}{\bar{a}}{b} \\
      \end{array}
    \right ) = \dbmij{m'}{b}{b} \geq 0
    \]
    Therefore $\dbmij{m}{b}{a} + d + \dbmij{m}{b}{b} \geq 0$. Since
    $\dbmij{m}{b}{b} = 0$ it follows that
    $\dbmij{m}{b}{a} + d \geq 0$.
    \\
    Now suppose that $\dbmij{m}{b}{a} + d \geq 0$,
    $\dbmij{m}{\bara}{\barb} + d \geq 0$ and
    $\dbmij{m}{\bara}{a} + d + \dbmij{m}{b}{\barb} + d \geq 0$. To
    show consistency we need to show that
    $\forall i. \dbmij{m'}{i}{i} \geq 0$. Pick an arbitrary $i$, then:
    \[
      \dbmij{m'}{i}{i} =  \min \left(
        \begin{array}{l}
          \dbmij{m}{i}{i}, \\
          \dbmij{m}{i}{a}+ d + \dbmij{m}{b}{i},\\
          \dbmij{m}{i}{\bar{b}} + d + \dbmij{m}{\bar{a}}{i},\\
          \dbmij{m}{i}{\bar{b}} + d + \dbmij{m}{\bar{a}}{a} + d + \dbmij{m}{b}{i}, \\
          \dbmij{m}{i}{a} + d + \dbmij{m}{b}{\bar{b}} + d + \dbmij{m}{\bar{a}}{i} \\
        \end{array} \right )
    \]
    We will show that $\dbmij{m'}{i}{i} \geq 0$. Recall that $\dbm{m}$
    is closed, and thus the second line above simplifies to:
    $\dbmij{m}{i}{a} + d + \dbmij{m}{b}{i} \geq \dbmij{m}{b}{a} + d
    \geq 0$. Similarly the third line:
    $\dbmij{m}{i}{\barb} + d + \dbmij{m}{\bara}{i} \geq
    \dbmij{m}{\bara}{\barb} + d \geq 0$, the fourth line
    :$\dbmij{m}{i}{\bar{b}} + d + \dbmij{m}{\bar{a}}{a} + d +
    \dbmij{m}{b}{i} \geq \dbmij{m}{b}{\barb} + d + \dbmij{m}{\bara}{a}
    + d$ and the fifth line:
    $\dbmij{m}{i}{a} + d + \dbmij{m}{b}{\bar{b}} + d +
    \dbmij{m}{\bar{a}}{i} \geq \dbmij{m}{\bara}{a} + d +
    \dbmij{m}{b}{\barb} + d \geq 0$. Thus every entry in the min
    expression is greater than 0 and thus
    $\forall i.\dbmij{m}{i}{i} \geq 0$ as required. \qed
\end{proof}

\begin{proof}[for theorem~\ref{thm:incrclosureclosed}]
  Suppose $\dbm{m'}$ is consistent. Because $\dbm{m}$ is closed
  $0 = \dbmij{m}{i}{i} \geq \dbmij{m'}{i}{i} \geq 0$ hence
  $\dbmij{m'}{i}{i} = 0$. It therefore remains to show
  $\forall i,j, k. \dbmij{m'}{i}{k} + \dbmij{m'}{k}{j} \geq A$ where
  \[
    A = \min\left(
      \begin{array}{@{}l@{}}
        \dbmij{m}{i}{j},\\
        \dbmij{m}{i}{a}+ d + \dbmij{m}{b}{j}, \\
        \dbmij{m}{i}{\bar{b}} + d + \dbmij{m}{\bar{a}}{j},\\
        \dbmij{m}{i}{\bar{b}} + d + \dbmij{m}{\bar{a}}{a} + d + \dbmij{m}{b}{j}, \\
        \dbmij{m}{i}{a} + d + \dbmij{m}{b}{\bar{b}} + d + \dbmij{m}{\bar{a}}{j} \\
      \end{array}
    \right )
    \]
    There are 5 cases for $\dbmij{m'}{i}{k}$ and 5 for
    $\dbmij{m'}{k}{j}$ giving 25 in total:
     \begin{enumerate}
     \item[1-1.] Suppose $\dbmij{m'}{i}{k} = \dbmij{m}{i}{k}$ and $\dbmij{m'}{k}{j} = \dbmij{m}{k}{j}$.  Because $\dbm{m}$ is closed:
       \begin{align*}
         \dbmij{m'}{i}{k} + \dbmij{m'}{k}{j} & = \dbmij{m}{i}{k} + \dbmij{m}{k}{j} \geq \dbmij{m}{i}{j} \geq A 
       \end{align*}
       
       \item[1-2.] Suppose  $\dbmij{m'}{i}{k} = \dbmij{m}{i}{k}$ and $\dbmij{m'}{k}{j} = \dbmij{m}{k}{a} + d + \dbmij{m}{b}{j}$.  Because $\dbm{m}$ is closed:
         \begin{align*}
           \dbmij{m'}{i}{k} + \dbmij{m'}{k}{j} & = \dbmij{m}{i}{k} + \dbmij{m}{k}{a} + d + \dbmij{m}{b}{j} \geq \dbmij{m}{i}{a} + d +  \dbmij{m}{b}{j} \geq A 
         \end{align*}

       \item[1-3.] Suppose $\dbmij{m'}{i}{k} = \dbmij{m}{i}{k}$ and $\dbmij{m'}{k}{j} = \dbmij{m}{k}{\bar{b}} + d + \dbmij{m}{\bar{a}}{j}$.  Because $\dbm{m}$ is closed: 
         \begin{align*}
           \dbmij{m'}{i}{k} + \dbmij{m'}{k}{j} & = \dbmij{m}{i}{k} + \dbmij{m}{k}{\bar{b}} + d + \dbmij{m}{\bar{a}}{j} 
            \geq \dbmij{m}{i}{\bar{b}} + d +  \dbmij{m}{\bar{a}}{j} \geq A 
         \end{align*}

       \item[1-4.] Suppose $\dbmij{m'}{i}{k} = \dbmij{m}{i}{k}$ and $\dbmij{m'}{k}{j} = \dbmij{m}{k}{\bar{b}} + d + \dbmij{m}{\bar{a}}{a} + d + \dbmij{m}{b}{j}$. Because $\dbm{m}$ is closed:
         \begin{align*}
           \dbmij{m'}{i}{k} + \dbmij{m'}{k}{j} & = \dbmij{m}{i}{k} + \dbmij{m}{k}{\bar{b}} + d + \dbmij{m}{\bar{a}}{a} + d + \dbmij{m}{b}{j} \\
                                               & \geq \dbmij{m}{i}{\bar{b}} + d +  \dbmij{m}{\bar{a}}{a} + d + \dbmij{m}{b}{j} \geq A 
         \end{align*}

       \item[1-5.] Suppose $\dbmij{m'}{i}{k} = \dbmij{m}{i}{k}$ and $\dbmij{m'}{k}{j} = \dbmij{m}{k}{a} + d + \dbmij{m}{b}{\bar{b}} + d + \dbmij{m}{\bar{a}}{j}$.  Because $\dbm{m}$ is closed:
         \begin{align*}
           \dbmij{m'}{i}{k} + \dbmij{m'}{k}{j} & = \dbmij{m}{i}{k} + \dbmij{m}{k}{a} + d + \dbmij{m}{b}{\bar{b}} + d + \dbmij{m}{\bar{a}}{j} \\
                                               & \geq \dbmij{m}{i}{a} + d +  \dbmij{m}{b}{\bar{b}} + d + \dbmij{m}{\bar{a}}{j} \geq A 
         \end{align*}

       \item[2-1.] Suppose $\dbmij{m'}{i}{k} = \dbmij{m}{i}{a} + d + \dbmij{m}{b}{k}$ and $\dbmij{m'}{k}{j} = \dbmij{m}{k}{j}$. 
       Symmetric to case~1-2.
 
         \item[2-2.] Suppose $\dbmij{m'}{i}{k} = \dbmij{m}{i}{a} + d + \dbmij{m}{b}{k}$ and $\dbmij{m'}{k}{j} = \dbmij{m}{k}{a} + d + \dbmij{m}{b}{j}$.  
         Because $\dbm{m}$ is closed and by Lemma~\ref{lemma:nonzero}:
           \begin{align*}
             \dbmij{m'}{i}{k} + \dbmij{m'}{k}{j} & = \dbmij{m}{i}{a} + d + \dbmij{m}{b}{k} + \dbmij{m}{k}{a} + d + \dbmij{m}{b}{j} \\
                                                 & \geq \dbmij{m}{i}{a} + d + \dbmij{m}{b}{a} + d + \dbmij{m}{b}{j} 
                                                 \geq \dbmij{m}{i}{a} + d + \dbmij{m}{b}{j} \geq A 
           \end{align*}

         \item[2-3.] Suppose $\dbmij{m'}{i}{k} = \dbmij{m}{i}{a} + d + \dbmij{m}{b}{k}$ and $\dbmij{m'}{k}{j} = \dbmij{m}{k}{\bar{b}} + d + \dbmij{m}{\bar{a}}{j}$.  Because $\dbm{m}$ is closed:
           \begin{align*}
             \dbmij{m'}{i}{k} + \dbmij{m'}{k}{j} & = \dbmij{m}{i}{a} + d + \dbmij{m}{b}{k} + \dbmij{m}{k}{\bar{b}} + d + \dbmij{m}{\bar{a}}{j} \\
                                                 & \geq \dbmij{m}{i}{a} + d + \dbmij{m}{b}{\bar{b}} + d + \dbmij{m}{\bar{a}}{j} \geq A 
           \end{align*}
           
         \item[2-4.] Suppose $\dbmij{m'}{i}{k} = \dbmij{m}{i}{a} + d + \dbmij{m}{b}{k}$ and $\dbmij{m'}{k}{j} = \dbmij{m}{k}{\bar{b}} + d + \dbmij{m}{\bar{a}}{a} + d + \dbmij{m}{b}{j}$.  
         Because $\dbm{m}$ is closed and by Lemma~\ref{lemma:nonzero}:
           \begin{align*}
             \dbmij{m'}{i}{k} + \dbmij{m'}{k}{j} & = \dbmij{m}{i}{a} + d + \dbmij{m}{b}{k} + \dbmij{m}{k}{\bar{b}} + d + \dbmij{m}{\bar{a}}{a} + d + \dbmij{m}{b}{j} \\
                                                 & \geq \dbmij{m}{i}{a} + d + \dbmij{m}{b}{\bar{b}} + d + \dbmij{m}{\bar{a}}{a} + d + \dbmij{m}{b}{j} \\
                                                 & \geq \dbmij{m}{i}{a} + d + \dbmij{m}{b}{j} \geq A 
           \end{align*}
           
         \item[2-5.] Suppose $\dbmij{m'}{i}{k} = \dbmij{m}{i}{a} + d + \dbmij{m}{b}{k}$ and $\dbmij{m'}{k}{j} = \dbmij{m}{k}{a} + d + \dbmij{m}{b}{\bar{b}} + d + \dbmij{m}{\bar{a}}{j}$.  
         Because $\dbm{m}$ is closed and by Lemma~\ref{lemma:nonzero}:
           \begin{align*}
             \dbmij{m'}{i}{k} + \dbmij{m'}{k}{j} & = \dbmij{m}{i}{a} + d + \dbmij{m}{b}{k} + \dbmij{m}{k}{a} + d + \dbmij{m}{b}{\bar{b}} + d + \dbmij{m}{\bar{a}}{j} \\
                                                 & \geq \dbmij{m}{i}{a} + d + \dbmij{m}{b}{a} + d + \dbmij{m}{b}{\bar{b}} + d + \dbmij{m}{\bar{a}}{j} \\
                                                 & \geq \dbmij{m}{i}{a} + d + \dbmij{m}{b}{\bar{b}} + d + \dbmij{m}{\bar{a}}{j}  \geq A 
           \end{align*}

      \item[3-1.] Suppose $\dbmij{m'}{i}{k} = \dbmij{m}{i}{\bar{b}} + d + \dbmij{m}{\bar{a}}{k}$ and $\dbmij{m'}{k}{j} = \dbmij{m}{k}{j}$.  
      Symmetric to case~1-3. 

         \item[3-2.] Suppose $\dbmij{m'}{i}{k} = \dbmij{m}{i}{\bar{b}} + d + \dbmij{m}{\bar{a}}{k}$ and $\dbmij{m'}{k}{j} = \dbmij{m}{k}{a} + d + \dbmij{m}{b}{j}$.  Symmetric to case~2-3.

         \item[3-3.] Suppose $\dbmij{m'}{i}{k} = \dbmij{m}{i}{\bar{b}} + d + \dbmij{m}{\bar{a}}{k}$ and $\dbmij{m'}{k}{j} = \dbmij{m}{k}{\bar{b}} + d + \dbmij{m}{\bar{a}}{j}$.  Because $\dbm{m}$ is closed and by 
         Lemma~\ref{lemma:nonzero}: 
           \begin{align*}
             \dbmij{m'}{i}{k} + \dbmij{m'}{k}{j} &= \dbmij{m}{i}{\bar{b}} + d + \dbmij{m}{\bar{a}}{k} + \dbmij{m}{k}{\bar{b}} + d + \dbmij{m}{\bar{a}}{j} \\
                                                 &\geq \dbmij{m}{i}{\bar{b}} + d + \dbmij{m}{\bar{a}}{\bar{b}} + d + \dbmij{m}{\bar{a}}{j} \\
                                                 &\geq \dbmij{m}{i}{\bar{b}} + d + \dbmij{m}{\bar{a}}{j} \geq A 
           \end{align*}

         \item[3-4.] Suppose $\dbmij{m'}{i}{k} = \dbmij{m}{i}{\bar{b}} + d + \dbmij{m}{\bar{a}}{k}$ and $\dbmij{m'}{k}{j} = \dbmij{m}{k}{\bar{b}} + d + \dbmij{m}{\bar{a}}{a} + d + \dbmij{m}{b}{j}$.  Because $\dbm{m}$ is closed and by Lemma~\ref{lemma:nonzero}:
           \begin{align*}
             \dbmij{m'}{i}{k} + \dbmij{m'}{k}{j} &= \dbmij{m}{i}{\bar{b}} + d + \dbmij{m}{\bar{a}}{k} + \dbmij{m}{k}{\bar{b}} + d + \dbmij{m}{\bar{a}}{a} + d + \dbmij{m}{b}{j} \\
                                                 &\geq \dbmij{m}{i}{\bar{b}} + d + \dbmij{m}{\bar{a}}{\bar{b}} + d + \dbmij{m}{\bar{a}}{a} + d + \dbmij{m}{b}{j} \\
                                                 &\geq \dbmij{m}{i}{\bar{b}} + d + \dbmij{m}{\bar{a}}{a} + d + \dbmij{m}{b}{j} \geq A  
           \end{align*}
           
         \item[3-5.] Suppose $\dbmij{m'}{i}{k} = \dbmij{m}{i}{\bar{b}} + d + \dbmij{m}{\bar{a}}{k}$ and $\dbmij{m'}{k}{j} = \dbmij{m}{k}{a} + d + \dbmij{m}{b}{\bar{b}} + d + \dbmij{m}{\bar{a}}{j}$.  Because $\dbm{m}$ is closed and by Lemma~\ref{lemma:nonzero}:
           \begin{align*}
             \dbmij{m'}{i}{k} + \dbmij{m'}{k}{j} &= \dbmij{m}{i}{\bar{b}} + d + \dbmij{m}{\bar{a}}{k} + \dbmij{m}{k}{a} + d + \dbmij{m}{b}{\bar{b}} + d + \dbmij{m}{\bar{a}}{j} \\
                                                 &= \dbmij{m}{i}{\bar{b}} + d + \dbmij{m}{\bar{a}}{a} + d + \dbmij{m}{b}{\bar{b}} + d + \dbmij{m}{\bar{a}}{j} \\
                                                 &= \dbmij{m}{i}{\bar{b}} + d + \dbmij{m}{\bar{a}}{j} \geq A 
           \end{align*}

\item[4-1.] Suppose $\dbmij{m'}{i}{k} = \dbmij{m}{i}{\bar{b}} + d + \dbmij{m}{\bar{a}}{a} + d + \dbmij{m}{b}{k}$ and $\dbmij{m'}{k}{j} = \dbmij{m}{k}{j}$.  
Symmetric to case~1-4. 
             
         \item[4-2.] Suppose $\dbmij{m'}{i}{k} = \dbmij{m}{i}{\bar{b}} + d + \dbmij{m}{\bar{a}}{a} + d + \dbmij{m}{b}{k}$ and $\dbmij{m'}{k}{j} = \dbmij{m}{k}{a} + d + \dbmij{m}{b}{j}$.  
Symmetric to case~2-4. 

         \item[4-3.] Suppose $\dbmij{m'}{i}{k} = \dbmij{m}{i}{\bar{b}} + d + \dbmij{m}{\bar{a}}{a} + d + \dbmij{m}{b}{k}$ and $\dbmij{m'}{k}{j} = \dbmij{m}{k}{\bar{b}} + d + \dbmij{m}{\bar{a}}{j}$.  
Symmetric to case~3-4.        
         
         \item[4-4.] Suppose $\dbmij{m'}{i}{k} = \dbmij{m}{i}{\bar{b}} + d + \dbmij{m}{\bar{a}}{a} + d + \dbmij{m}{b}{k}$ and $\dbmij{m'}{k}{j} = \dbmij{m}{k}{\bar{b}} + d + \dbmij{m}{\bar{a}}{a} + d + \dbmij{m}{b}{j}$.  Because $\dbm{m}$ is closed and by Lemma~\ref{lemma:nonzero}:
           \begin{align*}
             \dbmij{m'}{i}{k} + \dbmij{m'}{k}{j} &= \dbmij{m}{i}{\bar{b}} + d + \dbmij{m}{\bar{a}}{a} + d + \dbmij{m}{b}{k} + \dbmij{m}{k}{\bar{b}} + d + \dbmij{m}{\bar{a}}{a} + d + \dbmij{m}{b}{j} \\
                                                 &\geq \dbmij{m}{i}{\bar{b}} + d + \dbmij{m}{\bar{a}}{a} + d + \dbmij{m}{b}{\bar{b}} + d + \dbmij{m}{\bar{a}}{a} + d + \dbmij{m}{b}{j} \\
                                                 &\geq \dbmij{m}{i}{\bar{b}} + d + \dbmij{m}{\bar{a}}{a} + d + \dbmij{m}{b}{j} \geq A 
           \end{align*}

         \item[4-5.] Suppose $\dbmij{m'}{i}{k} = \dbmij{m}{i}{\bar{b}} + d + \dbmij{m}{\bar{a}}{a} + d + \dbmij{m}{b}{k}$ and $\dbmij{m'}{k}{j} = \dbmij{m}{k}{a} + d + \dbmij{m}{b}{\bar{b}} + d + \dbmij{m}{\bar{a}}{j}$.  Because $\dbm{m}$ is closed and by Lemma~\ref{lemma:nonzero}:
           \begin{align*}
             \dbmij{m'}{i}{k} + \dbmij{m'}{k}{j} &= \dbmij{m}{i}{\bar{b}} + d + \dbmij{m}{\bar{a}}{a} + d + \dbmij{m}{b}{k} + \dbmij{m}{k}{a} + d + \dbmij{m}{b}{\bar{b}} + d + \dbmij{m}{\bar{a}}{j} \\
                                                 &\geq \dbmij{m}{i}{\bar{b}} + d + \dbmij{m}{\bar{a}}{a} + d + \dbmij{m}{b}{a} + d + \dbmij{m}{b}{\bar{b}} + d + \dbmij{m}{\bar{a}}{j} \\
                                                 &\geq \dbmij{m}{i}{\bar{b}} + d + \dbmij{m}{\bar{a}}{a} + d + \dbmij{m}{b}{\bar{b}} + d + \dbmij{m}{\bar{a}}{j} \\
                                                 &\geq \dbmij{m}{i}{\bar{b}} + d + \dbmij{m}{\bar{a}}{j} \geq A 
           \end{align*}

       \item[5-1.] Suppose $\dbmij{m'}{i}{k} = \dbmij{m}{i}{a} + d + \dbmij{m}{b}{\bar{b}} + d + \dbmij{m}{\bar{a}}{k}$ and $\dbmij{m'}{k}{j} = \dbmij{m}{k}{j}$.  
       Symmetric to case~1-5.

           \item[5-2.] Suppose
               $\dbmij{m'}{i}{k} = \dbmij{m}{i}{a} + d + \dbmij{m}{b}{\bar{b}} + d + \dbmij{m}{\bar{a}}{k}$ and
               $\dbmij{m'}{k}{j} = \dbmij{m}{k}{a} + d + \dbmij{m}{b}{j}$.  
Symmetric to case~2-5.
             
           \item[5-3.] Suppose
               $\dbmij{m'}{i}{k} = \dbmij{m}{i}{a} + d + \dbmij{m}{b}{\bar{b}} + d + \dbmij{m}{\bar{a}}{k}$ and
               $\dbmij{m'}{k}{j} = \dbmij{m}{k}{\bar{b}} + d + \dbmij{m}{\bar{a}}{j}$.  
Symmetric to case~3-5.
             
           \item[5-4.] Suppose
               $\dbmij{m'}{i}{k} = \dbmij{m}{i}{a} + d + \dbmij{m}{b}{\bar{b}} + d + \dbmij{m}{\bar{a}}{k}$ and
               $\dbmij{m'}{k}{j} = \dbmij{m}{k}{\bar{b}} + d + \dbmij{m}{\bar{a}}{a} + d + \dbmij{m}{b}{j}$.    
 Symmetric to case~4-5.              
             
           \item[5-5.] Suppose
               $\dbmij{m'}{i}{k} = \dbmij{m}{i}{a} + d + \dbmij{m}{b}{\bar{b}} + d + \dbmij{m}{\bar{a}}{k}$ and
               $\dbmij{m'}{k}{j} = \dbmij{m}{k}{a} + d + \dbmij{m}{b}{\bar{b}} + d + \dbmij{m}{\bar{a}}{j}$.  Because $\dbm{m}$ is closed and by Lemma~\ref{lemma:nonzero}:
             \begin{align*}
               \dbmij{m'}{i}{k} + \dbmij{m'}{k}{j} &= \dbmij{m}{i}{a} + d + \dbmij{m}{b}{\bar{b}} + d + \dbmij{m}{\bar{a}}{k} + \dbmij{m}{k}{a} + d + \dbmij{m}{b}{\bar{b}} + d + \dbmij{m}{\bar{a}}{j} \\
                                                   &\geq \dbmij{m}{i}{a} + d + \dbmij{m}{b}{\bar{b}} + d + \dbmij{m}{\bar{a}}{a} + d + \dbmij{m}{b}{\bar{b}} + d + \dbmij{m}{\bar{a}}{j} \\
                                                   &\geq \dbmij{m}{i}{a} + d + \dbmij{m}{b}{\bar{b}} + d + \dbmij{m}{\bar{a}}{j} \geq A 
             \end{align*}
           \end{enumerate} \qed

\end{proof}


%% file: proofs-coherence.tex

\subsection{Proofs for Properties of Incremental Closure}

\begin{proof}[for proposition~\ref{lemma:coherence}] \mbox{}
\begin{itemize}

\item
Suppose $\dbmij{m'}{i}{j} = \dbmij{m}{i}{j}$.
Because $\dbm{m}$ is coherent $\dbmij{m'}{i}{j} = \dbmij{m}{\barj}{\bari} \geq \dbmij{m'}{\barj}{\bari}$.

\item
Suppose $\dbmij{m'}{i}{j} = \dbmij{m}{i}{a}+ d + \dbmij{m}{b}{j}$.
Because $\dbm{m}$ is coherent $\dbmij{m'}{i}{j} = \dbmij{m}{\barj}{\bar{b}} + d + \dbmij{m}{\bar{a}}{\bari} \geq  \dbmij{m'}{\barj}{\bari}$.

\item
Suppose $\dbmij{m'}{i}{j} = \dbmij{m}{i}{\bar{b}} + d + \dbmij{m}{\bar{a}}{j}$.
Similar to the previous case.

\item
Suppose $\dbmij{m'}{i}{j} = \dbmij{m}{i}{\bar{b}} + d + \dbmij{m}{\bar{a}}{a} + d + \dbmij{m}{b}{j}$.
Because $\dbm{m}$ is coherent $\dbmij{m'}{i}{j} = \dbmij{m}{\barj}{\bar{b}} + d + \dbmij{m}{\bar{a}}{a}  + d + \dbmij{m}{b}{\bari} \geq  \dbmij{m'}{\barj}{\bari}$.

\item
Suppose $\dbmij{m'}{i}{j} = \dbmij{m}{i}{a} + d + \dbmij{m}{b}{\bar{b}} + d + \dbmij{m}{\bar{a}}{j}$.
Similar to the previous case.

\end{itemize}
Since $\dbmij{m'}{i}{j} \geq  \dbmij{m'}{\barj}{\bari}$ it follows
$\dbmij{m'}{\barj}{\bari} \geq \dbmij{m'}{i}{j}$ hence
$\dbmij{m'}{i}{j} =  \dbmij{m'}{\barj}{\bari}$ as required. \qed
\end{proof}

\begin{proof}[for proposition~\ref{lemma-idempotence}] Suppose $\dbm{m'}$ is consistent.
By Lemma~\ref{lemma:nonzero} it follows that
$\dbmij{m}{b}{a} + d \geq 0$,
$\dbmij{m}{\bar{a}}{\bar{b}} + d \geq 0$,
$\dbmij{m}{\bar{a}}{a} + d + \dbmij{m}{b}{\bar{b}} + d \geq 0$
and
$\dbmij{m}{b}{\bar{b}} + d + \dbmij{m}{\bar{a}}{a} + d \geq 0$.
Therefore
\[
    \dbmij{m'}{\bara}{a} = \min\left ( \begin{array}{l}
                                          \dbmij{m}{\bara}{a}, \\
                                          \dbmij{m}{\bara}{a} + d + \dbmij{m}{b}{a}, \\
                                          \dbmij{m}{\bara}{\barb} + d + \dbmij{m}{\bara}{a} \\
                                          \dbmij{m}{\bara}{a} + d + \dbmij{m}{b}{\barb} + d + \dbmij{m}{\bara}{a} \\
                                          \dbmij{m}{\bara}{\barb} + d + \dbmij{m}{\bara}{a} + d + \dbmij{m}{b}{a}  \\
                                        \end{array} \right ) = \dbmij{m}{\bara}{a}
\]
Likewise $\dbmij{m'}{b}{\barb} = \dbmij{m}{b}{\barb}$.  Using the same inequalities it follows
\[
\begin{array}{rcl@{\qquad}rcl}
\dbmij{m'}{i}{a}  
& 
= 
&
\min \left (\begin{array}{l}
                                                   \dbmij{m}{i}{a}, 
                                                   \dbmij{m}{i}{\barb} + d + \dbmij{m}{\bara}{a} \\
                                                 \end{array} \right )                                                
&                                                                                            
\dbmij{m'}{b}{j}  
& 
= 
&                                      \min \left (\begin{array}{l}
                                                   \dbmij{m}{b}{j}, 
                                                   \dbmij{m}{b}{\barb} + d + \dbmij{m}{\bara}{j} \\
                                                 \end{array} \right ) 
\\[0.5ex]
\dbmij{m'}{i}{\barb}  
& = 
&\min \left (\begin{array}{l}
                                                   \dbmij{m}{i}{\barb}, 
                                                   \dbmij{m}{i}{a} + d + \dbmij{m}{b}{\barb} \\
                                                 \end{array} \right ) 
&
\dbmij{m'}{\bara}{j}  & =  &\min \left (\begin{array}{l}
                                                   \dbmij{m}{\bara}{j}, 
                                                   \dbmij{m}{\bara}{a} + d + \dbmij{m}{b}{j} \\
                                                 \end{array} \right )                                                  
\end{array}
\]
Therefore
\begin{align*}
\dbmij{m'}{i}{a} + d + \dbmij{m'}{b}{j}                                            
& =
                                      \min \left( \begin{array}{l}
                                                    \dbmij{m}{i}{a} + d + \dbmij{m}{b}{j} \\
                                                    \dbmij{m}{i}{a} + d + \dbmij{m}{b}{\barb} + d + \dbmij{m}{\bara}{j} \\
                                                    \dbmij{m}{i}{\barb} + d + \dbmij{m}{\bara}{a} + d + \dbmij{m}{b}{j} \\
                                                   \dbmij{m}{i}{\barb} + d + \dbmij{m}{\bara}{a} + d + \dbmij{m}{b}{\barb} + d + \dbmij{m}{\bara}{j} \\                                                    
                                                  \end{array} \right) \\
& \geq
                                      \min \left( \begin{array}{l}
                                                    \dbmij{m}{i}{a} + d + \dbmij{m}{b}{j} \\
                                                    \dbmij{m}{i}{a} + d + \dbmij{m}{b}{\barb} + d + \dbmij{m}{\bara}{j} \\
                                                    \dbmij{m}{i}{\barb} + d + \dbmij{m}{\bara}{a} + d + \dbmij{m}{b}{j} \\
                                                   \dbmij{m}{i}{\barb} + d + \dbmij{m}{\bara}{j} \\                                                    
                                                  \end{array} \right)
\geq
\dbmij{m'}{i}{j}
\end{align*}
Likewise $\dbmij{m'}{i}{\barb} + d + \dbmij{m'}{\bara}{j} \geq \dbmij{m'}{i}{j}$.
Now consider
\begin{eqnarray*}
\lefteqn{\dbmij{m'}{i}{a} + d +  \dbmij{m'}{b}{\barb} + d + \dbmij{m'}{\bara}{j}} \\
& \qquad = & \min \left (\begin{array}{l}
                                                     \dbmij{m}{i}{a} + d + \dbmij{m}{b}{\barb} + d + \dbmij{m}{\bara}{j} \\
                                                     \dbmij{m}{i}{a} + d + \dbmij{m}{b}{\barb} + d + \dbmij{m}{\bara}{a} + d + \dbmij{m}{b}{j} \\
                                                     \dbmij{m}{i}{\barb} + d+ \dbmij{m}{\bara}{a} + d + \dbmij{m}{b}{\barb} + d + \dbmij{m}{\bara}{j} \\
                                                     \dbmij{m}{i}{\barb} + d + \dbmij{m}{\bara}{a} + d + \dbmij{m}{b}{\barb} + d + \dbmij{m}{\bara}{a} + d + \dbmij{m}{b}{j} \\
                                      \end{array} \right ) \\
                      & \qquad \geq & \min \left ( \begin{array}{l}
                                                   \dbmij{m}{i}{a} + d + \dbmij{m}{b}{\barb} + d + \dbmij{m}{\bara}{j} \\
                                                    \dbmij{m}{i}{a} + d + \dbmij{m}{b}{j} \\
                                                    \dbmij{m}{i}{\barb} + d + \dbmij{m}{\bara}{j} \\
                                                    \dbmij{m}{i}{\barb} + d + \dbmij{m}{\bara}{a} + d + \dbmij{m}{b}{j} \\
                                                  \end{array} \right ) \geq \dbmij{m'}{i}{j}
\end{eqnarray*}
Likewise $\dbmij{m'}{i}{\barb} + d + \dbmij{m'}{\bara}{a} + d + \dbmij{m'}{b}{j}  \geq \dbmij{m'}{i}{j}$.
Thus $\dbmij{m''}{i}{j} = \dbmij{m'}{i}{j}$.  
Now suppose $\dbm{m'}$ is not consistent.  Since $\dbmij{m''}{i}{i} \leq \dbmij{m'}{i}{i}$ then $\dbm{m''}$ is not consistent. 
\qed
\end{proof}


%% file: proofs-strongclosure.tex

\subsection{Proofs for Incremental Strong Closure}

  \begin{proof}[for theorem~\ref{thm:strongclosurestrengthen}]
    Observe that
    $\dbmij{m'}{i}{\bari} = \min(\dbmij{m}{i}{\bari},
    (\dbmij{m}{i}{\bari} + \dbmij{m}{i}{\bari})/2) =
    \dbmij{m}{i}{\bari}$ and likewise
    $\dbmij{m'}{j}{\barj} = \dbmij{m}{j}{\barj}$.
    Therefore
    \begin{align*}
      \frac{\dbmij{m'}{i}{\bari} + \dbmij{m'}{\barj}{j}}{2} &=  \frac{\dbmij{m}{i}{\bari} + \dbmij{m}{\barj}{j}}{2} 
                                                             \geq \min \left ( \begin{array}{c} \dbmij{m}{i}{j} \\ \frac{\dbmij{m}{i}{\bari} + \dbmij{m}{\barj}{j}}{2} \end{array}\right ) 
                                                            =\dbmij{m'}{i}{j}
    \end{align*}

Because $\dbm{m}$ is closed
$0 = \dbmij{m}{i}{i} \leq \dbmij{m}{i}{\bari} + \dbmij{m}{\bari}{i}$
and thus
\[
\dbmij{m'}{i}{i} = 
\min(\dbmij{m}{i}{i},   (\dbmij{m}{i}{\bari} + \dbmij{m}{\bari}{i})/2) =
\min(0,   (\dbmij{m}{i}{\bari} + \dbmij{m}{\bari}{i})/2) = 0
\]
  
To show
  $\dbmij{m'}{i}{j} \leq \dbmij{m'}{i}{k} +
  \dbmij{m'}{k}{j}$ we proceed by case analysis:
    \begin{itemize}

    \item Suppose $\dbmij{m'}{i}{k} = \dbmij{m}{i}{k}$ and $\dbmij{m'}{k}{j} = \dbmij{m}{k}{j}$. 
    Because $\dbm{m}$ is closed:
\[
\dbmij{m'}{i}{j} \leq \dbmij{m}{i}{j} \leq \dbmij{m}{i}{k} + \dbmij{m}{k}{j} = 
      \dbmij{m'}{i}{k} + \dbmij{m'}{k}{j} \]
    
    \item Suppose $\dbmij{m'}{i}{k} \neq \dbmij{m}{i}{k}$ and $\dbmij{m'}{k}{j} = \dbmij{m}{k}{j}$.
    Because $\dbm{m}$ is closed and coherent:
      \begin{align*}
       2\dbmij{m'}{i}{k} + 2\dbmij{m'}{k}{j} & = \dbmij{m}{i}{\bari} + \dbmij{m}{\bark}{k} + 2\dbmij{m}{k}{j} \geq \dbmij{m}{i}{\bari} + \dbmij{m}{\bark}{j} + \dbmij{m}{k}{j}  \\
       & = \dbmij{m}{i}{\bari} + \dbmij{m}{\barj}{k} + \dbmij{m}{k}{j}  \geq \dbmij{m}{i}{\bari} + \dbmij{m}{\barj}{j} \geq 2\dbmij{m'}{i}{j} 
      \end{align*}
     
    \item Suppose $\dbmij{m'}{i}{k} = \dbmij{m}{i}{k}$ and $\dbmij{m'}{k}{j} \neq \dbmij{m}{k}{j}$.     
 Symmetric to the previous case.    
      
    \item Suppose $\dbmij{m'}{i}{k} \neq \dbmij{m}{i}{k}$ and $\dbmij{m'}{k}{j} \neq \dbmij{m}{k}{j}$.
Because $\dbm{m}$ is closed:
      \begin{align*}
        2\dbmij{m'}{i}{k} + 2\dbmij{m'}{k}{j} & = \dbmij{m}{i}{\bari} + \dbmij{m}{\bark}{k} + \dbmij{m}{k}{\bark} + \dbmij{m}{\barj}{j} &\\
        & \geq \dbmij{m}{i}{\bari} + \dbmij{m}{\bark}{\bark} + \dbmij{m}{\barj}{j} = \dbmij{m}{i}{\bari} + 0 + \dbmij{m}{\barj}{j} \geq 2\dbmij{m'}{i}{j} 
      \end{align*}

    \end{itemize}
\qed

  \end{proof}

\begin{proof}[for proposition~\ref{lemma-stengthen-idempotent}]
Let $\dbm{m''} = \Strengthen(\dbm{m'})$.
Observe $\dbmij{m'}{i}{\bari} = \min(\dbmij{m}{i}{\bari}, (\dbmij{m}{i}{\bari} + \dbmij{m}{i}{\bari})/2) = \dbmij{m}{i}{\bari}$ and
likewise $\dbmij{m'}{\barj}{j}  = \dbmij{m}{\barj}{j}$.
Therefore
\[
\begin{array}{rcl}
\dbmij{m''}{i}{j} & = & \min(\dbmij{m'}{i}{j}, (\dbmij{m'}{i}{\bari} + \dbmij{m'}{\barj}{j})/2) \\
& = & \min(\min(\dbmij{m}{i}{j}, (\dbmij{m}{i}{\bari} + \dbmij{m}{\barj}{j})/2), (\dbmij{m}{i}{\bari} + \dbmij{m}{\barj}{j})/2) \\
& = & \min(\dbmij{m}{i}{j}, (\dbmij{m}{i}{\bari} + \dbmij{m}{\barj}{j})/2) = \dbmij{m'}{i}{j} 
\end{array}\]
\qed
\end{proof}

\begin{proof}[for proposition~\ref{prop-str-monotonicity}]
\[
\begin{array}{rcl}
\Strengthen(\dbmij{m^2}{i}{j}) & = & \min(\dbmij{m^2}{i}{j}, \frac{\dbmij{m^2}{i}{\bari} + \dbmij{m^2}{\barj}{j}}{2}) \\ 
& \geq & \min(\dbmij{m^1}{i}{j}, \frac{\dbmij{m^1}{i}{\bari} + \dbmij{m^1}{\barj}{j}}{2}) = \Strengthen(\dbmij{m^1}{i}{j})
\end{array}
\]
\qed
\end{proof}

\begin{proof}[for proposition~\ref{prop:strongreductive}]

  $\dbmij{m'}{i}{j} = \min(\dbmij{m}{i}{j}, \frac{\dbmij{m}{i}{\bari} + \dbmij{m}{\barj}{j}}{2}) \leq \dbmij{m}{i}{j}$
 \qed 
\end{proof}

\begin{proof}[for proposition~\ref{prop:strongcoherence}]
  \begin{align*}
    \dbmij{m'}{i}{j} &= \min(\dbmij{m}{i}{j}, \frac{\dbmij{m}{i}{\bari} + \dbmij{m}{\barj}{j}}{2}) 
                     = \min(\dbmij{m}{\barj}{\bari}, \frac{\dbmij{m}{\barj}{j} + \dbmij{m}{i}{\bari}}{2}) = \dbmij{m'}{\barj}{\bari}
  \end{align*}
  \qed
\end{proof}

\begin{proof}[for theorem~\ref{thm:incrstrongclosure}]
   We prove that $\forall i,j. \dbmij{m'}{i}{j} = \dbmij{m^*}{i}{j}$. Pick some $i,j$.
   \begin{itemize}
   \item Suppose $j = \bari$.  Then
     \begin{align*}
       \dbmij{m^*}{i}{\bari} & = \min(\dbmij{m^\dagger}{i}{\bari}, \dbmij{m^\dagger}{i}{\bari}/2 +  \dbmij{m^\dagger}{i}{\bari}/2) 
                              = \dbmij{m^\dagger}{i}{\bari} \\
                             & = \min \left ( \begin{array}{@{}l@{}}
                                   \dbmij{m}{i}{\bari},\\
                                   \dbmij{m}{i}{a}+ d + \dbmij{m}{b}{\bari}, \\
                                   \dbmij{m}{i}{\bar{b}} + d + \dbmij{m}{\bar{a}}{\bari},\\
                                   \dbmij{m}{i}{\bar{b}} + d + \dbmij{m}{\bar{a}}{a} + d + \dbmij{m}{b}{\bari}, \\
                                   \dbmij{m}{i}{a} + d + \dbmij{m}{b}{\bar{b}} + d + \dbmij{m}{\bar{a}}{\bari} \\
                                 \end{array} \right ) 
                              = \dbmij{m'}{i}{\bari} 
     \end{align*}

     \item Suppose $j \neq \bari$.  Then
       \begin{align*}
       \dbmij{m^*}{i}{j} & = \min(\dbmij{m^\dagger}{i}{j}, \dbmij{m^\dagger}{i}{\bari}/2 +  \dbmij{m^\dagger}{j}{\barj}/2) \\
                         & = \min(\dbmij{m^\dagger}{i}{j}, \dbmij{m'}{i}{\bari}/2 +  \dbmij{m'}{j}{\barj}/2) \\
                         & = \min \left ( \begin{array}{l}
                                            \dbmij{m}{i}{j}, \\
                                            \dbmij{m}{i}{a}+ d + \dbmij{m}{b}{j},\\
                                            \dbmij{m}{i}{\bar{b}} + d + \dbmij{m}{\bar{a}}{j},\\
                                            \dbmij{m}{i}{\bar{b}} + d + \dbmij{m}{\bar{a}}{a} + d + \dbmij{m}{b}{j}, \\
                                            \dbmij{m}{i}{a} + d + \dbmij{m}{b}{\bar{b}} + d + \dbmij{m}{\bar{a}}{j}, \\
                                            (\dbmij{m'}{i}{\bari} + \dbmij{m'}{\barj}{j})/2
                                          \end{array} \right ) 
                          = \dbmij{m'}{i}{j} 
       \end{align*}
   \end{itemize}
   \qed
  \end{proof}

  \begin{proof}[Proof for Theorem~\ref{lemma:incrstrongclosurereduce}]
Suppose $\dbmij{m}{a}{b} + d \geq 0$. Then it is sufficient to show that:
    \[
\min\left (
          \begin{array}{l}
            \dbmij{m}{i}{j}, \\
            \dbmij{m}{i}{a}+ d + \dbmij{m}{b}{j},\\
            \dbmij{m}{i}{\bar{b}} + d + \dbmij{m}{\bar{a}}{j},\\
            \dbmij{m}{i}{\bar{b}} + d + \dbmij{m}{\bar{a}}{a} + d + \dbmij{m}{b}{j}, \\
            \dbmij{m}{i}{a} + d + \dbmij{m}{b}{\bar{b}} + d + \dbmij{m}{\bar{a}}{j}, \\
            (\dbmij{m'}{i}{\bari} + \dbmij{m'}{\barj}{j})/2
          \end{array} \right )
        =
        \min\left (
          \begin{array}{l}
            \dbmij{m}{i}{j}, \\
            \dbmij{m}{i}{a}+ d + \dbmij{m}{b}{j},\\
            \dbmij{m}{i}{\bar{b}} + d + \dbmij{m}{\bar{a}}{j},\\
            \dbmij{m}{i}{\bar{b}} + d + \dbmij{m}{\bar{a}}{a} + d + \dbmij{m}{b}{j}, \\
            \dbmij{m}{i}{a} + d + \dbmij{m}{b}{\bar{b}} + d + \dbmij{m}{\bar{a}}{j}, \\
            (\dbmij{m}{i}{a} + d + \dbmij{m}{b}{\bari} + \dbmij{m}{\barj}{j})/2, \\
            (\dbmij{m}{i}{\bari} + \dbmij{m}{\barj}{a} + d + \dbmij{m}{b}{j})/2
          \end{array} \right )
    \]
where
    \[
      \dbmij{m'}{i}{\bari} =
      \min \left (
        \begin{array}{@{}l@{}}
          \dbmij{m}{i}{\bari}, \\
          \dbmij{m}{i}{a}+ d + \dbmij{m}{b}{\bari},\\
          \dbmij{m}{i}{\bar{b}} + d + \dbmij{m}{\bar{a}}{a} + d + \dbmij{m}{b}{\bari}, \\
          \dbmij{m}{i}{a} + d + \dbmij{m}{b}{\bar{b}} + d + \dbmij{m}{\bar{a}}{\bari}, \\
        \end{array}
      \right )
\quad
      \dbmij{m'}{\barj}{j} =
      \min \left (
        \begin{array}{@{}l@{}}
          \dbmij{m}{\barj}{j}, \\
          \dbmij{m}{\barj}{a}+ d + \dbmij{m}{b}{j},\\
          \dbmij{m}{\barj}{\bar{b}} + d + \dbmij{m}{\bar{a}}{a} + d + \dbmij{m}{b}{j}, \\
          \dbmij{m}{\barj}{a} + d + \dbmij{m}{b}{\bar{b}} + d + \dbmij{m}{\bar{a}}{j}, \\
        \end{array}
      \right )
    \]
    Using the above, 
    $(\dbmij{m'}{i}{\bari} + \dbmij{m'}{\barj}{j})/2$ expands into one of the following cases:
    \begin{itemize}
    \item[1-1] Suppose $\dbmij{m'}{i}{\bari} = \dbmij{m}{i}{\bari}$ and $\dbmij{m'}{\barj}{j} = \dbmij{m}{\barj}{j}$.
By strong closure  $\frac{\dbmij{m}{i}{\bari} + \dbmij{m}{\barj}{j}}{2} \geq \dbmij{m}{i}{j}$.
Thus this case is redundant.

    \item[1-2] Suppose $\dbmij{m'}{i}{\bari} = \dbmij{m}{i}{\bari}$ and $\dbmij{m'}{\barj}{j} = \dbmij{m}{\barj}{a} + d + \dbmij{m}{b}{j}$.
      This case is not redundant.

    \item[1-3]  Suppose $\dbmij{m'}{i}{\bari} = \dbmij{m}{i}{\bari}$ and $\dbmij{m'}{\barj}{j} = \dbmij{m}{\barj}{\barb} + d + \dbmij{m}{\bara}{a} + d + \dbmij{m}{b}{j}$.  By strong closure and coherence:
      \begin{align*}
              \frac{\dbmij{m}{i}{\bari} + (\dbmij{m}{\barj}{\barb} + d + \dbmij{m}{\bara}{a} + d + \dbmij{m}{b}{j})}{2} & = \\
        \frac{\dbmij{m}{i}{\bari} + \dbmij{m}{\bara}{a}}{2} + \frac{2d + \dbmij{m}{\barj}{\barb} + \dbmij{m}{b}{j}}{2} & \geq  \dbmij{m}{i}{a} + \frac{2d + 2\dbmij{m}{b}{j}}{2} =  \dbmij{m}{i}{a} + d + \dbmij{m}{b}{j}
      \end{align*}

    \item[1-4]  Suppose $\dbmij{m'}{i}{\bari} = \dbmij{m}{i}{\bari}$ and $\dbmij{m'}{\barj}{j} = \dbmij{m}{\barj}{a} + d + \dbmij{m}{b}{\barb} + d + \dbmij{m}{\bara}{j}$. By strong closure and coherence:
      \begin{align*}
      \frac{\dbmij{m}{i}{\bari} + (\dbmij{m}{\barj}{a} + d + \dbmij{m}{b}{\barb} + d + \dbmij{m}{\bara}{j})}{2} & = \\
              \frac{\dbmij{m}{i}{\bari} + \dbmij{m}{b}{\barb}}{2} + \frac{2d + \dbmij{m}{\barj}{a} + \dbmij{m}{\bara}{j}}{2} & 
              \geq  \dbmij{m}{i}{\barb} + \frac{2d + 2\dbmij{m}{\bara}{j}}{2} =  \dbmij{m}{i}{\barb} + d + \dbmij{m}{\bara}{j}
      \end{align*}

    \item[2-1]  Suppose $\dbmij{m'}{i}{\bari} = \dbmij{m}{i}{a} + d + \dbmij{m}{b}{\bari}$ and $\dbmij{m'}{\barj}{j} = \dbmij{m}{\barj}{j}$.
      This case is not redundant.

    \item[2-2]  Suppose $\dbmij{m'}{i}{\bari} = \dbmij{m}{i}{a} + d + \dbmij{m}{b}{\bari}$ and 
    $\dbmij{m'}{\barj}{j} = \dbmij{m}{\barj}{a} + d + \dbmij{m}{b}{j}$.
Observe that if $x \leq y$ then $x \leq (x+y)/2 \leq y$ and if
 $y \leq x$ then $y \leq (x+y)/2 \leq x$. Thus $(x+y)/2 \geq \min(x, y)$ hence
\[    
\frac{(\dbmij{m}{i}{a} + d + \dbmij{m}{b}{\bari}) + (\dbmij{m}{\barj}{a} + d + \dbmij{m}{b}{j})}{2}
\geq
\min(\dbmij{m}{i}{a} + d + \dbmij{m}{b}{\bari}, \dbmij{m}{\barj}{a} + d + \dbmij{m}{b}{j})
\]
Thus this case is redundant.

    \item[2-3]  Suppose $\dbmij{m'}{i}{\bari} = \dbmij{m}{i}{a} + d + \dbmij{m}{b}{\bari} $ and $\dbmij{m'}{\barj}{j} = \dbmij{m}{\barj}{\barb} + d + \dbmij{m}{\bara}{a} + d + \dbmij{m}{b}{j}$. By coherence and using $(x+y)/2 \geq \min(x, y)$:
      \begin{align*}
        &\phantom{=} \frac{ (\dbmij{m}{i}{a} + d + \dbmij{m}{b}{\bari}) + (\dbmij{m}{\barj}{\barb} + d + \dbmij{m}{\bara}{a} + d + \dbmij{m}{b}{j})}{2}  \\
        &= \frac{ (\dbmij{m}{i}{a} + d + \dbmij{m}{b}{j}) + (\dbmij{m}{i}{\barb} +  d + \dbmij{m}{\bara}{a} + d + \dbmij{m}{b}{j}) }{2} \\
        &\geq \min(\dbmij{m}{i}{a} + d + \dbmij{m}{b}{j}, \dbmij{m}{i}{\barb} +  d + \dbmij{m}{\bara}{a} + d + \dbmij{m}{b}{j})
      \end{align*}
      Thus this case is redundant.
      
          \item[2-4]  Suppose $\dbmij{m'}{i}{\bari} = \dbmij{m}{i}{a} + d + \dbmij{m}{b}{\bari} $ and $\dbmij{m'}{\barj}{j} = \dbmij{m}{\barj}{a} + d + \dbmij{m}{b}{\barb} + d + \dbmij{m}{\bara}{j}$. By coherence and using $(x+y)/2 \geq \min(x, y)$:
      \begin{align*}
        &\phantom{=} \frac{ (\dbmij{m}{i}{a} + d + \dbmij{m}{b}{\bari}) + (\dbmij{m}{\barj}{a} + d + \dbmij{m}{b}{\barb} + d + \dbmij{m}{\bara}{j})}{2} \\
        &= \frac{(\dbmij{m}{i}{a} + d + \dbmij{m}{b}{\barb} + d + \dbmij{m}{\bara}{j}) + (\dbmij{m}{i}{\barb} + d + \dbmij{m}{\bara}{j})}{2} \\
        &\geq \min(\dbmij{m}{i}{a} + d + \dbmij{m}{b}{\barb} + d + \dbmij{m}{\bara}{j}, \dbmij{m}{i}{\barb} + d + \dbmij{m}{\bara}{j})
      \end{align*}
      Thus this case is redundant.
      
    \item[3-1]  Suppose $\dbmij{m'}{i}{\bari} = \dbmij{m}{i}{\barb} + d + \dbmij{m}{\bara}{a} + d + \dbmij{m}{b}{\bari}$ and $\dbmij{m'}{\barj}{j} = \dbmij{m}{\barj}{j}$.
Symmetric to 1-3.

    \item[3-2]  Suppose $\dbmij{m'}{i}{\bari} = \dbmij{m}{i}{\barb} + d + \dbmij{m}{\bara}{a} + d + \dbmij{m}{b}{\bari}$ and $\dbmij{m'}{\barj}{j} = \dbmij{m}{\barj}{a} + d + \dbmij{m}{b}{j}$.
Symmetric to case 2-3.
      
    \item[3-3]  Suppose $\dbmij{m'}{i}{\bari} = \dbmij{m}{i}{\barb} + d + \dbmij{m}{\bara}{a} + d + \dbmij{m}{b}{\bari}$ and $\dbmij{m'}{\barj}{j} = \dbmij{m}{\barj}{\barb} + d + \dbmij{m}{\bara}{a} + d + \dbmij{m}{b}{j}$. Then
      \begin{align*}
        &\phantom{=} \frac{ (\dbmij{m}{i}{\barb} + d + \dbmij{m}{\bara}{a} + d + \dbmij{m}{b}{\bari}) + (\dbmij{m}{\barj}{\barb} + d + \dbmij{m}{\bara}{a} + d + \dbmij{m}{b}{j})}{2} \\
        &=\frac{(\dbmij{m}{i}{\barb} + d + \dbmij{m}{\bara}{a} + d + \dbmij{m}{b}{j}) + (\dbmij{m}{i}{\barb} + d + \dbmij{m}{\bara}{a} + d + \dbmij{m}{b}{j})}{2} \\
        &=\frac{(\dbmij{m}{i}{\barb} + d + \dbmij{m}{\bara}{a} + d + \dbmij{m}{b}{j}) + (\dbmij{m}{i}{\barb} + d + \dbmij{m}{\bara}{a} + d + \dbmij{m}{b}{j})}{2} \\
        &= \dbmij{m}{i}{\barb} + d + \dbmij{m}{\bara}{a} + d + \dbmij{m}{b}{j}
      \end{align*}
      Thus this case is redundant.
      
    \item[3-4]  Suppose $\dbmij{m'}{i}{\bari} = \dbmij{m}{i}{\barb} + d + \dbmij{m}{\bara}{a} + d + \dbmij{m}{b}{\bari}$
      and $\dbmij{m'}{\barj}{j} = \dbmij{m}{\barj}{a} + d + \dbmij{m}{b}{\barb} + d + \dbmij{m}{\bara}{j}$. By coherence, strong closure and because $\dbmij{m}{b}{a} + d \geq 0$:
      \begin{align*}
       &\phantom{=} \frac{(\dbmij{m}{i}{\barb} + d + \dbmij{m}{\bara}{a} + d + \dbmij{m}{b}{\bari}) + (\dbmij{m}{\barj}{a} + d + \dbmij{m}{b}{\barb} + d + \dbmij{m}{\bara}{j})}{2} \\
        &=\frac{\dbmij{m}{\bara}{a} + \dbmij{m}{b}{\barb}}{2} + \frac{4d + 2\dbmij{m}{i}{\barb} + 2\dbmij{m}{\bara}{j}}{2} \geq \dbmij{m}{\bara}{\barb} + 2d + \dbmij{m}{i}{\barb} + \dbmij{m}{\bara}{j} \\
        &=  (\dbmij{m}{b}{a} + d) + \dbmij{m}{i}{\barb} + d + \dbmij{m}{\bara}{j} \geq \dbmij{m}{i}{\barb} + d + \dbmij{m}{\bara}{j}
      \end{align*}
      Thus this case is redundant.
      
    \item[4-1]  Suppose $\dbmij{m'}{i}{\bari} = \dbmij{m}{i}{a} + d + \dbmij{m}{b}{\barb} + d + \dbmij{m}{\bara}{\bari}$ and $\dbmij{m'}{\barj}{j} = \dbmij{m}{\barj}{j}$.
Symmetric to case 1-4.
      
          \item[4-2]  Suppose $\dbmij{m'}{i}{\bari} = \dbmij{m}{i}{a} + d + \dbmij{m}{b}{\barb} + d + \dbmij{m}{\bara}{\bari}$ and $\dbmij{m'}{\barj}{j} = \dbmij{m}{\barj}{a} + d + \dbmij{m}{b}{j}$.  Symmetric to case 2-4.
      
    \item[4-3]  Suppose $\dbmij{m'}{i}{\bari} = \dbmij{m}{i}{a} + d + \dbmij{m}{b}{\barb} + d + \dbmij{m}{\bara}{\bari}$ and $\dbmij{m'}{\barj}{j} = \dbmij{m}{\barj}{\barb} + d + \dbmij{m}{\bara}{a} + d + \dbmij{m}{b}{j}$. 
By coherence, strong closure and because $\dbmij{m}{\bara}{\barb} + d \geq 0$:
      \begin{align*}
       &\phantom{=} \frac{(\dbmij{m}{i}{a} + d + \dbmij{m}{b}{\barb} + d + \dbmij{m}{\bara}{\bari})  + (\dbmij{m}{\barj}{\barb} + d + \dbmij{m}{\bara}{a} + d + \dbmij{m}{b}{j})}{2} \\
        &=\frac{\dbmij{m}{\bara}{a} + \dbmij{m}{b}{\barb}}{2} + \frac{4d + 2\dbmij{m}{i}{a} + 2\dbmij{m}{b}{j}}{2} \geq \dbmij{m}{\bara}{\barb} + 2d + \dbmij{m}{i}{a} + \dbmij{m}{b}{j} \\
        &=  (\dbmij{m}{b}{a} + d) + \dbmij{m}{i}{a} + d + \dbmij{m}{b}{j} \geq \dbmij{m}{i}{a} + d + \dbmij{m}{b}{j}
      \end{align*}
      Thus this case is redundant.
      
          \item[4-4]  Suppose $\dbmij{m'}{i}{\bari} = \dbmij{m}{i}{a} + d + \dbmij{m}{b}{\barb} + d + \dbmij{m}{\bara}{\bari}$ and $\dbmij{m'}{\barj}{j} = \dbmij{m}{\barj}{a} + d + \dbmij{m}{b}{\barb} + d + \dbmij{m}{\bara}{j}$. By coherence:
      \begin{align*}
        &\phantom{=} \frac{ (\dbmij{m}{i}{a} + d + \dbmij{m}{b}{\barb} + d + \dbmij{m}{\bara}{\bari}) + (\dbmij{m}{\barj}{a} + d + \dbmij{m}{b}{\barb} + d + \dbmij{m}{\bara}{j})}{2} \\
        &= \frac{ (\dbmij{m}{i}{a} + d + \dbmij{m}{b}{\barb} + d + \dbmij{m}{\bara}{j}) + (\dbmij{m}{i}{a} + d + \dbmij{m}{b}{\barb} + d + \dbmij{m}{\bara}{j})}{2} \\
        & = (\dbmij{m}{i}{a} + d + \dbmij{m}{b}{\barb} + d + \dbmij{m}{\bara}{j}) 
      \end{align*}
      
    \end{itemize}
   
Now suppose that $\dbmij{m}{a}{b} + d < 0$. 
By corollary~\ref{cor:nonzero1} $\IncrementalClosure(\dbm{m},o)$ is not consistent and since $\dbm{m^{*}}  \leq \IncrementalClosure(\dbm{m},o)$ 
and $\dbm{m''}  \leq \IncrementalClosure(\dbm{m},o)$ 
it follows that both $\dbm{m^{*}} $ and $\dbm{m''} $ are not consistent as required.
    \qed
  \end{proof}

  


%% file: proofs-integerclosure.tex

\subsection{Proofs for Incremental Tight Closure}

\begin{proof}[for lemma~\ref{lemma:tightlyclosed}]
Suppose $\dbm{m'}$ is consistent.
By lemma~\ref{lemma:closed} it follows that $\dbm{m'}$ is closed.

We will now show that $\dbm{m'}$ is strongly closed i.e
\mbox{$\forall i,j. \dbmij{m'}{i}{j} \leq \dbmij{m'}{i}{\bari}/2 + \dbmij{m'}{\barj}{j} / 2$}.
 \begin{itemize}
\item Suppose $\dbmij{m'}{i}{\bari} = \dbmij{m}{i}{\bari}$ and $\dbmij{m'}{\barj}{j} = \dbmij{m}{\barj}{j}$. Then
  \begin{align*}
    \frac{\dbmij{m'}{i}{\bari}}{2} + \frac{\dbmij{m'}{\barj}{j}}{2} & = \frac{\dbmij{m}{i}{\bari}}{2} + \frac{\dbmij{m}{\barj}{j}}{2} \geq  \floorfrac{\dbmij{m}{i}{\bari}}{2} + \floorfrac{\dbmij{m}{\barj}{j}}{2}  \geq \dbmij{m'}{i}{j} 
  \end{align*}

\item Suppose $\dbmij{m'}{i}{\bari} \neq \dbmij{m}{i}{\bari}$ and $\dbmij{m'}{\barj}{j} = \dbmij{m}{\barj}{j}$.  Then
  \begin{align*}
    \frac{\dbmij{m'}{i}{\bari}}{2} + \frac{\dbmij{m'}{\barj}{j}}{2} & = \frac{\floor{\frac{\dbmij{m}{i}{\bari}}{2}} + \floor{\frac{\dbmij{m}{i}{\bari}}{2}}}{2} + \frac{\dbmij{m}{\barj}{j}}{2} & \\
    & = \floorfrac{\dbmij{m}{i}{\bari}}{2} + \frac{\dbmij{m}{\barj}{j}}{2} 
    \geq \floorfrac{\dbmij{m}{i}{\bari}}{2} + \floorfrac{\dbmij{m}{\barj}{j}}{2} 
    \geq  \dbmij{m}{i}{j} = \dbmij{m'}{i}{j} 
  \end{align*}

\item Suppose $\dbmij{m'}{i}{\bari} = \dbmij{m}{i}{\bari}$ and $\dbmij{m'}{\barj}{j} \neq \dbmij{m}{\barj}{j}$. Symmetric to the previous case.

\item Suppose $\dbmij{m'}{i}{\bari} \neq \dbmij{m}{i}{\bari}$ and $\dbmij{m'}{\barj}{j} \neq \dbmij{m}{\barj}{j}$.  Then

  \begin{align*}
    \frac{\dbmij{m'}{i}{\bari}}{2} + \frac{\dbmij{m'}{\barj}{j}}{2} & = \frac{\floorfrac{\dbmij{m}{i}{\bari}}{2} + \floorfrac{\dbmij{m}{i}{\bari}}{2}}{2} + \frac{\floorfrac{\dbmij{m}{\barj}{j}}{2} + \floorfrac{\dbmij{m}{\barj}{j}}{2}}{2} \\ & = \floorfrac{\dbmij{m}{i}{\bari}}{2} + \floorfrac{\dbmij{m}{\barj}{j}}{2} \geq  \dbmij{m'}{i}{j} 
  \end{align*}

\end{itemize}
Thus, if $\dbm{m'}$ is consistent, it is strongly closed. It remains
to show that $\forall i. \dbmij{m'}{i}{\bari}$ is even. Observe that:
\begin{align*}
   \dbmij{m'}{i}{\bari} &= \min(\dbmij{m}{i}{\bari}, \floorfrac{\dbmij{m}{i}{\bari}}{2}  +  \floorfrac{\dbmij{m}{i}{\bari}}{2} ) =  \min(\dbmij{m}{i}{\bari}, 2\floorfrac{\dbmij{m}{i}{\bari}}{2})
 \end{align*}
 \begin{itemize}
 \item Suppose $\dbmij{m}{i}{\bari}$ is even.
Then $2 \floorfrac{\dbmij{m}{i}{\bari}}{2} = \dbmij{m}{i}{\bari} = \dbmij{m'}{i}{\bari}$ which is even.
   
 \item Suppose $\dbmij{m}{i}{\bari}$ is odd.
Then $2 \floorfrac{\dbmij{m}{i}{\bari}}{2} = \dbmij{m}{i}{\bari} - 1 = \dbmij{m'}{i}{\bari}$ which is even. \qed
 \end{itemize}
 \end{proof}

\begin{proof}[for theorem~\ref{theorem:correct-inc-tight-closure}]
    We prove that $\forall i,j. \dbmij{m}{i}{j} = \dbmij{m'}{i}{j}$.
    Pick some $i,j$.
    \begin{itemize}
    \item Suppose $j = \bari$. Then
      \begin{align*}
        \dbmij{m^*}{i}{\bari} &= \min(\dbmij{m^{\ddag}}{i}{\bari}, \dbmij{m^{\ddag}}{i}{\bari} / 2 + \dbmij{m^{\ddag}}{i}{\bari}/2) = \dbmij{m^{\ddag}}{i}{\bari} = 2 \floor{\dbmij{m^{\dagger}}{i}{\bari}/2} \\
                              &= 2 \left \lfloor \min \left ( \begin{array}{@{}l@{}}
                                   \dbmij{m}{i}{\bari},\\
                                   \dbmij{m}{i}{a}+ d + \dbmij{m}{b}{\bari}, \\
                                   \dbmij{m}{i}{\bar{b}} + d + \dbmij{m}{\bar{a}}{\bari},\\
                                   \dbmij{m}{i}{\bar{b}} + d + \dbmij{m}{\bar{a}}{a} + d + \dbmij{m}{b}{\bari}, \\
                                   \dbmij{m}{i}{a} + d + \dbmij{m}{b}{\bar{b}} + d + \dbmij{m}{\bar{a}}{\bari} \\
                                 \end{array} \right ) / 2 \right \rfloor  = \dbmij{m'}{i}{\bari}
      \end{align*}

    \item Suppose $j \neq \bari$. Then
      \begin{align*}
        \dbmij{m^*}{i}{j} &= \min(\dbmij{m^{\ddag}}{i}{j}, \dbmij{m^{\ddag}}{i}{\bari} / 2 + \dbmij{m^{\ddag}}{\barj}{j}/2)  
                          = \min(\dbmij{m^{\dagger}}{i}{j}, \dbmij{m'}{i}{\bari} / 2 + \dbmij{m'}{\barj}{j}/2)  \\
                          &= \min \left ( \begin{array}{l}
                                            \dbmij{m}{i}{j}, \\
                                            \dbmij{m}{i}{a}+ d + \dbmij{m}{b}{j},\\
                                            \dbmij{m}{i}{\bar{b}} + d + \dbmij{m}{\bar{a}}{j},\\
                                            \dbmij{m}{i}{\bar{b}} + d + \dbmij{m}{\bar{a}}{a} + d + \dbmij{m}{b}{j}, \\
                                            \dbmij{m}{i}{a} + d + \dbmij{m}{b}{\bar{b}} + d + \dbmij{m}{\bar{a}}{j}, \\
                                            (\dbmij{m'}{i}{\bari} + \dbmij{m'}{\barj}{j})/2
                                          \end{array} \right ) = \dbmij{m'}{i}{j}
      \end{align*}
    \end{itemize}
    \qed
  \end{proof}

\begin{proof}[for proposition~\ref{prop:tighten-idempotence}]
  Let $\dbm{m''} = \Tighten(\dbm{m'})$. 
  \begin{itemize}
  \item Suppose $j \neq \bari$.  Then
      $\dbmij{m''}{i}{j} = \dbmij{m'}{i}{j}$.
    
  \item Suppose $j = \bari$.  Then
      $\dbmij{m''}{i}{\bari} = 2 \floorfrac{\dbmij{m'}{i}{\bari}}{2} = 2 \floorfrac{ 2 \floorfrac{\dbmij{m}{i}{\bari}}{2}}{2} =  2 \floorfrac{\dbmij{m}{i}{\bari}}{2}  = \dbmij{m'}{i}{\bari}$
  \end{itemize}
  \qed
\end{proof}

\begin{proof}[for proposition~\ref{prop:intmonotonicity}] \mbox{}
  \begin{itemize}
  \item Suppose $j \neq \bari$.  Then
      $\Tighten(\dbmij{m^2}{i}{j}) = \dbmij{m^2}{i}{j} \geq \dbmij{m^1}{i}{j} = \Tighten(\dbmij{m^1}{i}{j})$.
    
  \item Suppose $j = \bari$.  Then
      $\Tighten(\dbmij{m^2}{i}{\bari}) = 2 \floorfrac{\dbmij{m^2}{i}{\bari}}{2} \geq 2 \floorfrac{\dbmij{m^1}{i}{\bari}}{2} =  \Tighten(\dbmij{m^1}{i}{\bari})$
  \end{itemize}
  \qed
  \end{proof}

\begin{proof}[for proposition~\ref{prop:tightenreductiveness}] \mbox{}
  \begin{itemize}

  \item Suppose $j = \bari$. Then
    $\dbmij{m'}{i}{j} = \dbmij{m'}{i}{\bari} \leq 2
    \floorfrac{\dbmij{m}{i}{\bari}}{2} = \dbmij{m}{i}{\bari} =
    \dbmij{m}{i}{j}$.

  \item Suppose $j \neq \bari$. Then
    $\dbmij{m'}{i}{j} = \dbmij{m}{i}{j}$.
  \end{itemize}
  \qed
\end{proof}

\begin{proof}[for proposition~\ref{prop:tightencoherence}] \mbox{}
  \begin{itemize}
  \item Suppose $j = \bari$. Then
    $\dbmij{m'}{\barj}{\bari} = 2
    \floorfrac{\dbmij{m}{\barj}{\bari}}{2} = 2
    \floorfrac{\dbmij{m}{i}{j}}{2} = \dbmij{m'}{i}{j}$.
    
  \item Suppose $j \neq \bari$. Then
    $\dbmij{m'}{\barj}{\bari} = \dbmij{m}{\barj}{\bari} =
    \dbmij{m}{i}{j} = \dbmij{m'}{i}{j}$.
  \end{itemize}
  \qed
\end{proof}


%% file: proofs-inplace.tex

\subsection{Proofs for In-place Update}

\begin{proof}[for corollary~\ref{corollary:idempotence-facts}]
By Proposition~\ref{lemma-idempotence} it follows
$\dbm{m}' = \IncrementalClosure(\dbm{m}', o)$.  The result then follows
from Theorem~\ref{thm:incrclosureclosed}.
\end{proof}

\begin{proof}[for theorem~\ref{lemma:insitu}] Suppose $\dbm{m'}$ is consistent.

Let $k = 0$.  
It vacuously follows that $\forall 0 \leq \ell < k. \dbm{m^{k}_{\rho^{-1}(\ell)}} = \dbm{m'_{\rho^{-1}(\ell)}}$.
Moreover $\forall k \leq \ell < 4n^2.  \dbm{m^{k}_{\rho^{-1}(\ell)}} = \dbm{m_{\rho^{-1}(\ell)}}$  since $\dbm{m}^{0} = \dbm{m}$.

Now let $k > 0$ and suppose $\rho(i,j) = k$ and consider
\begin{align*}
  \dbmij{m^{k+1}}{i}{j} &= \min \left ( 
                         \begin{array}{l}
                           \dbmij{m^k}{i}{j} \\
                           \dbmij{m^k}{i}{a} + d + \dbmij{m^k}{b}{j} \\
                           \dbmij{m^k}{i}{\barb} + d + \dbmij{m^k}{\bara}{j} \\
                           \dbmij{m^k}{i}{a} + d + \dbmij{m^k}{b}{\barb} + d + \dbmij{m^k}{\bara}{j} \\
                           \dbmij{m^k}{i}{\barb} + d + \dbmij{m^k}{\bara}{a} + d + \dbmij{m^k}{b}{j} 
                         \end{array} \right ) 
\end{align*}
If $\rho^{-1}(i, a) < k$ then $\dbmij{m^{k}}{i}{a} = \dbmij{m'}{i}{a}$ whereas
if $\rho^{-1}(i, a) \geq k$ then $\dbmij{m^{k}}{i}{a} = \dbmij{m}{i}{a} \geq \dbmij{m'}{i}{a}$.
Thus $\dbmij{m^{k}}{i}{a} \geq \dbmij{m'}{i}{a}$ and likewise
$\dbmij{m^{k}}{b}{j} \geq \dbmij{m'}{b}{j}$.
By Corollary~\ref{corollary:idempotence-facts} it follows
$\dbmij{m^k}{i}{a} + d + \dbmij{m^k}{b}{j} \geq \dbmij{m'}{i}{a} + d + \dbmij{m'}{b}{j} \geq  \dbmij{m'}{i}{j}$.
By a similar argument
$\dbmij{m^k}{i}{\barb} + d + \dbmij{m^k}{\bara}{j} \geq  \dbmij{m'}{i}{j}$,
$\dbmij{m^k}{i}{a} + d + \dbmij{m^k}{b}{\barb} + d + \dbmij{m^k}{\bara}{j} \geq  \dbmij{m'}{i}{j}$
and likewise
$\dbmij{m^k}{i}{\barb} + d + \dbmij{m^k}{\bara}{a} + d + \dbmij{m^k}{b}{j} \geq  \dbmij{m'}{i}{j}$.

Since $\dbmij{m^k}{i}{j} = \dbmij{m}{i}{j} \geq \dbmij{m'}{i}{j}$ it follows $\dbmij{m^{k+1}}{i}{j} \geq  \dbmij{m'}{i}{j}$.
But $\dbm{m^k} \leq \dbm{m}$ and by Proposition~\ref{lemma-monotonicity} 
$\dbmij{m^{k+1}}{i}{j} \leq \dbmij{m'}{i}{j}$ hence $\dbmij{m^{k+1}}{i}{j} = \dbmij{m'}{i}{j}$.
Hence it follows
$\forall 0 \leq \ell < k + 1. \dbm{m^{k+1}_{\rho^{-1}(\ell)}} = \dbm{m'_{\rho^{-1}(\ell)}}$.
Moreover $\forall k + 1 \leq \ell < 4n^2  . \dbm{m^{k+1}_{\rho^{-1}(\ell)}} = \dbm{m_{\rho^{-1}(\ell)}}$.

Suppose $\dbm{m'}$ is inconsistent hence
$\dbmij{m'}{i}{i} < 0$. Put $k = \rho(i,i)$. 
But $\dbm{m^k} \leq \dbm{m}$ and by Proposition~\ref{lemma-monotonicity} 
$\dbmij{m^{4n^2}}{i}{i} = \dbmij{m^{k+1}}{i}{i} \leq \dbmij{m'}{i}{i} < 0$ as required.
\qed
\end{proof}

\begin{proof}[for lemma~\ref{lemma-str-close-idempotent}]  Suppose $\dbm{m'}$ is consistent.  By Proposition~\ref{lemma:coherence} $\dbm{m'}$ is coherent.  
\begin{enumerate}

\item To show $\dbmij{m''}{i}{j} \leq \dbmij{m''}{i}{a} + d + \dbmij{m''}{b}{j}$.
\begin{itemize}

\item Suppose $\dbmij{m''}{i}{a} = \dbmij{m'}{i}{a}$ and
$\dbmij{m''}{b}{j} = \dbmij{m'}{b}{j}$. Because $\dbm{m'}$ is consistent
by Corollary~\ref{corollary:idempotence-facts} it follows:
\[
\dbmij{m''}{i}{a} + d + \dbmij{m''}{b}{j} = \dbmij{m'}{i}{a} + d + \dbmij{m'}{b}{j} \geq \dbmij{m'}{i}{j}  \geq \dbmij{m''}{i}{j}
\]

\item Suppose $\dbmij{m''}{i}{a} = (\dbmij{m'}{i}{\bari} + \dbmij{m'}{\bara}{a})/2$
and $\dbmij{m''}{b}{j} = \dbmij{m'}{b}{j}$. Because $\dbm{m'}$ is consistent
by Corollary~\ref{corollary:idempotence-facts} it follows
$\dbmij{m'}{\bara}{j} \leq \dbmij{m'}{\bara}{a} + d + \dbmij{m'}{b}{j}$
and
$\dbmij{m'}{\barj}{j} \leq \dbmij{m'}{\barj}{a} + d + \dbmij{m'}{b}{j}$.
Hence
\begin{align*}
\dbmij{m''}{i}{a} + d + \dbmij{m''}{b}{j}  &=  (\dbmij{m'}{i}{\bari} + \dbmij{m'}{\bara}{a}+ 2d + 2\dbmij{m'}{b}{j})/2 \\
  & \geq (\dbmij{m'}{i}{\bari} + \dbmij{m'}{\bara}{j} + d + \dbmij{m'}{b}{j} )/2 \\
  & \geq (\dbmij{m'}{i}{\bari} + \dbmij{m'}{\barj}{j})/2 \geq \dbmij{m''}{i}{j}
\end{align*}

\item Suppose $\dbmij{m''}{i}{a} = \dbmij{m'}{i}{a}$
and $\dbmij{m''}{b}{j} = (\dbmij{m'}{b}{\barb} + \dbmij{m'}{\barj}{j})/2$. Symmetric to the previous case.

\item Suppose $\dbmij{m''}{i}{a} = (\dbmij{m'}{i}{\bari} + \dbmij{m'}{\bara}{a})/2$
and $\dbmij{m''}{b}{j} = (\dbmij{m'}{b}{\barb} + \dbmij{m'}{\barj}{j})/2$.
Because $\dbm{m'}$ is consistent
by Corollary~\ref{corollary:idempotence-facts} it follows
$\dbmij{m'}{\bara}{\barb} \leq \dbmij{m'}{\bara}{a} + d + \dbmij{m'}{b}{\barb}$
and
$\dbmij{m'}{i}{a} \leq \dbmij{m'}{i}{a} + d + \dbmij{m'}{b}{a}$
thus
$0 \leq d + \dbmij{m'}{b}{a}$. Hence
\begin{align*}
\dbmij{m''}{i}{a} + d + \dbmij{m''}{b}{j} &=  (\dbmij{m'}{i}{\bari} + \dbmij{m'}{\bara}{a} + 2d + \dbmij{m'}{b}{\barb} + \dbmij{m'}{\barj}{j})/2 \\
&\geq (\dbmij{m'}{i}{\bari} + \dbmij{m'}{\bara}{\barb} + d + \dbmij{m'}{\barj}{j})/2 \\ 
&\geq (\dbmij{m'}{i}{\bari} + \dbmij{m'}{\barj}{j})/2 \geq \dbmij{m'}{i}{j}
\end{align*}
\end{itemize}
\item To show $\dbmij{m''}{i}{j} \leq \dbmij{m''}{i}{\bar{b}} + d + \dbmij{m''}{\bar{a}}{j}$. Analogous to the previous case.

\item To show $\dbmij{m''}{i}{j} \leq \dbmij{m''}{i}{\bar{b}} + d + \dbmij{m''}{\bar{a}}{a} + d + \dbmij{m''}{b}{j}$.
  \begin{itemize}

\item Suppose $\dbmij{m''}{i}{\barb} = \dbmij{m'}{i}{\barb}$ and $\dbmij{m''}{b}{j} = \dbmij{m'}{b}{j}$.
Since $\dbmij{m''}{\bar{a}}{a} = \dbmij{m'}{\bar{a}}{a}$ 
and because $\dbm{m'}$ is consistent
by Corollary~\ref{corollary:idempotence-facts} it follows
\[
\begin{array}{lll}
\multicolumn{3}{l}{\dbmij{m''}{i}{\bar{b}} + d + \dbmij{m''}{\bar{a}}{a} + d + \dbmij{m''}{b}{j}} \\
\qquad & = & 
\dbmij{m'}{i}{\bar{b}} + d + \dbmij{m'}{\bar{a}}{a} + d + \dbmij{m'}{b}{j} \geq \dbmij{m'}{i}{j} \geq \dbmij{m''}{i}{j}
\end{array}
\]

  \item Suppose $\dbmij{m''}{i}{\barb} = (\dbmij{m}{i}{\bari} + \dbmij{m}{b}{\barb})/2$ and $\dbmij{m''}{b}{j}=\dbmij{m'}{b}{j}$.
Because $\dbm{m'}$ is consistent
by Corollary~\ref{corollary:idempotence-facts} it follows
$\dbmij{m'}{\bara}{j} \leq \dbmij{m'}{\bara}{a} + d + \dbmij{m'}{b}{j}$,
\mbox{$\dbmij{m'}{\barj}{j} \leq \dbmij{m'}{\barj}{a} + d + \dbmij{m'}{b}{j}$},
\mbox{$\dbmij{m'}{b}{a} \leq \dbmij{m'}{b}{\barb} + d + \dbmij{m'}{\bara}{a}$}
and
\mbox{$0 \leq d + \dbmij{m'}{b}{a}$}.  Therefore
\[
\begin{array}{lll}
\multicolumn{3}{l}{\dbmij{m''}{i}{\bar{b}} + d + \dbmij{m''}{\bar{a}}{a} + d + \dbmij{m''}{b}{j}} \\
\qquad & = & (\dbmij{m'}{i}{\bari} + \dbmij{m'}{b}{\barb} + 2\dbmij{m'}{\bara}{a} + 4d + 2\dbmij{m'}{b}{j})/2 \\
& \geq & (\dbmij{m'}{i}{\bari} + \dbmij{m'}{b}{\barb} + \dbmij{m'}{\bara}{a} + 3d + \dbmij{m'}{\bara}{j} + \dbmij{m'}{b}{j})/2 \\
& \geq & (\dbmij{m'}{i}{\bari} + \dbmij{m'}{b}{\barb} + \dbmij{m'}{\bara}{a} + 2d + \dbmij{m'}{\barj}{j})/2 \\
& \geq & (\dbmij{m'}{i}{\bari} + \dbmij{m'}{b}{a} + d + \dbmij{m'}{\barj}{j})/2 \\
& \geq & (\dbmij{m'}{i}{\bari} + \dbmij{m'}{\barj}{j})/2 \geq \dbmij{m''}{i}{j}
\end{array}
\]

  \item Suppose $\dbmij{m''}{i}{\barb} = \dbmij{m'}{i}{\barb}$ and $\dbmij{m''}{b}{j}=(\dbmij{m'}{b}{\barb} + \dbmij{m'}{\barj}{j})/2$. Symmetric to the previous case.

  \item Suppose $\dbmij{m''}{i}{\barb} = (\dbmij{m}{i}{\bari} + \dbmij{m}{b}{\barb})/2$ and $\dbmij{m''}{b}{j}=(\dbmij{m'}{b}{\barb} + \dbmij{m'}{\barj}{j})/2$.
Because $\dbm{m'}$ is consistent
by Corollary~\ref{corollary:idempotence-facts} it follows
\mbox{$\dbmij{m'}{b}{a} \leq \dbmij{m'}{b}{\barb} + d + \dbmij{m'}{\bara}{a}$}
and
\mbox{$0 \leq d + \dbmij{m'}{b}{a}$}.  Therefore    
\[
\begin{array}{lll}
\multicolumn{3}{l}{\dbmij{m''}{i}{\bar{b}} + d + \dbmij{m''}{\bar{a}}{a} + d + \dbmij{m''}{b}{j}} \\
\qquad & =      & (\dbmij{m'}{i}{\bari} + \dbmij{m'}{b}{\barb} + 4d + 2\dbmij{m''}{\bara}{a} + 2\dbmij{m'}{b}{\barb} + \dbmij{m'}{\barj}{j})/2 \\
       &\geq & (\dbmij{m'}{i}{\bari} + 2\dbmij{m'}{b}{a} + 2d + \dbmij{m'}{\barj}{j})/2 \\
      &\geq & (\dbmij{m'}{i}{\bari} + \dbmij{m'}{\barj}{j})/2 \geq \dbmij{m''}{i}{j}
\end{array}
\]
  \end{itemize}

\item To show 
$\dbmij{m''}{i}{j} \leq \dbmij{m''}{i}{a} + d + \dbmij{m''}{b}{\bar{b}} + d + \dbmij{m''}{\bar{a}}{j}$. 
Analogous to the previous case.

\end{enumerate}
It therefore follows that $\dbm{m'''} = \dbm{m''}$.
Now suppose $\dbm{m'}$ is not consistent.  Hence $\dbm{m''}$ is not consistent thus $\dbm{m'''}$ is not consistent.
\qed
\end{proof}

\begin{proof}[for theorem~\ref{lemma:stronginsitu}] Suppose $\dbm{m'}$ is consistent.

  Let $k = 0$. It vacuously follows that
  $\forall 0 \leq \ell < k. \dbm{m^{k}_{\rho^{-1}(\ell)}} =
  \dbm{m''_{\rho^{-1}(\ell)}}$. Moreover
  $\forall k \leq \ell < 4n^2. \dbm{m^{k}_{\rho^{-1}(\ell)}} =
  \dbm{m_{\rho^{-1}(\ell)}}$ since $\dbm{m}^{0} = \dbm{m}$.

Suppose $0 < k$ and $\rho(i,j) = k$.  Now suppose $j = \bari$. Then
  \[
    \dbmij{m^{k+1}}{i}{\bari} =  \min \left (\begin{array}{l}
                                           \dbmij{m^{k}}{i}{\bari}, \\
                                           \dbmij{m^{k}}{i}{a} + d + \dbmij{m^{k}}{b}{\bari}, \\
                                           \dbmij{m^{k}}{i}{\bar{b}} + d + \dbmij{m^{k}}{\bar{a}}{\bari},\\
                                           \dbmij{m^{k}}{i}{\bar{b}} + d + \dbmij{m^{k}}{\bar{a}}{a} + d + \dbmij{m^{k}}{b}{\bari}, \\
                                           \dbmij{m^{k}}{i}{a} + d + \dbmij{m^{k}}{b}{\bar{b}} + d + \dbmij{m^{k}}{\bar{a}}{\bari} \\
                                         \end{array} \right ) 
  \]
If $\rho(i,a) < k$ then
  $\dbmij{m^{k}}{i}{j} = \dbmij{m''}{i}{a}$ otherwise
  $\rho(i,a) \geq k$ then
  $\dbmij{m^{k}}{i}{a} = \dbmij{m}{i}{a} \geq \dbmij{m''}{i}{a}$ which
  implies $\dbmij{m^{k}}{i}{a} \geq \dbmij{m''}{i}{a}$. 
  Likewise $\dbmij{m^{k}}{b}{\bari} \geq \dbmij{m''}{b}{\bari}$. 
  By
  Lemma~\ref{lemma-str-close-idempotent} and
  Corollary~\ref{corollary:idempotence-facts} it follows
  $\dbmij{m^k}{i}{a} + d + \dbmij{m^k}{b}{j} \geq \dbmij{m''}{i}{a} +
  d + \dbmij{m''}{b}{j} \geq \dbmij{m''}{i}{j}$. By a similar argument
  $\dbmij{m^k}{i}{\barb} + d + \dbmij{m^k}{\bara}{j} \geq
  \dbmij{m''}{i}{j}$,
  $\dbmij{m^k}{i}{a} + d + \dbmij{m^k}{b}{\barb} + d +
  \dbmij{m^k}{\bara}{j} \geq \dbmij{m''}{i}{j}$ and likewise
  $\dbmij{m^k}{i}{\barb} + d + \dbmij{m^k}{\bara}{a} + d +
  \dbmij{m^k}{b}{j} \geq \dbmij{m''}{i}{j}$. Thus
  $\dbmij{m^{k+1}}{i}{j} \geq \dbmij{m''}{i}{j}$. 
  Now to show $\dbmij{m''}{i}{j} \geq \dbmij{m^{k+1}}{i}{j}$. Observe
\begin{align*}
  \dbmij{m''}{i}{\bari}  = \dbmij{m'}{i}{\bari} 
                    &= \min \left (\begin{array}{l}
                                                  \dbmij{m}{i}{j}, \\
                                                  \dbmij{m}{i}{a} + d + \dbmij{m}{b}{j}, \\
                                                  \dbmij{m}{i}{\bar{b}} + d + \dbmij{m}{\bar{a}}{j},\\
                                                  \dbmij{m}{i}{\bar{b}} + d + \dbmij{m}{\bar{a}}{a} + d + \dbmij{m}{b}{j}, \\
                                                  \dbmij{m}{i}{a} + d + \dbmij{m}{b}{\bar{b}} + d + \dbmij{m}{\bar{a}}{j} \\
                                                \end{array} \right ) \\
                    &\geq \min \left (\begin{array}{l}
                                                  \dbmij{m^{k}}{i}{j}, \\
                                                  \dbmij{m^{k}}{i}{a} + d + \dbmij{m^{k}}{b}{j}, \\
                                                  \dbmij{m^{k}}{i}{\bar{b}} + d + \dbmij{m^{k}}{\bar{a}}{j},\\
                                                  \dbmij{m^{k}}{i}{\bar{b}} + d + \dbmij{m^{k}}{\bar{a}}{a} + d + \dbmij{m^{k}}{b}{j}, \\
                                                  \dbmij{m^{k}}{i}{a} + d + \dbmij{m^{k}}{b}{\bar{b}} + d + \dbmij{m^{k}}{\bar{a}}{j} \\
                                                \end{array} \right ) = \dbmij{m^{k+1}}{i}{\bari}  
\end{align*}
  Hence
  $\forall 0 \leq \ell < k. \dbm{m}^{k}_{\rho^{-1}(\ell)} =
  \dbm{m''}_{\rho^{-1}(\ell)}$. Moreover
  $\forall k + 1 \leq \ell < 4n^2 . \dbm{m^{k+1}_{\rho^{-1}(\ell)}} =
  \dbm{m_{\rho^{-1}(\ell)}}$ follows from the inductive hypothesis and
  the definition of $\dbmij{m^{k+1}}{i}{j}$.

  Now suppose that $j \neq \bari$.  Then  $2n < \rho(i,j)$  and
  consider
\begin{align*}
  \dbmij{m^{k+1}}{i}{j} &= \min \left ( 
                         \begin{array}{l}
                           \dbmij{m^k}{i}{j} \\
                           \dbmij{m^k}{i}{a} + d + \dbmij{m^k}{b}{j} \\
                           \dbmij{m^k}{i}{\barb} + d + \dbmij{m^k}{\bara}{j} \\
                           \dbmij{m^k}{i}{a} + d + \dbmij{m^k}{b}{\barb} + d + \dbmij{m^k}{\bara}{j} \\
                           \dbmij{m^k}{i}{\barb} + d + \dbmij{m^k}{\bara}{a} + d + \dbmij{m^k}{b}{j}, \\
                           (\dbmij{m^k}{i}{\bari} + \dbmij{m^k}{\barj}{j})/2 
                         \end{array} \right ) 
\end{align*}
Notice that
$\dbmij{m^{k}}{i}{\bari} + \dbmij{m^{k}}{\barj}{j}/2 =
\dbmij{m''}{i}{\bari} + \dbmij{m''}{\barj}{j}/2$, since
$\rho(i,\bari) < 2n \leq \rho(i,j) = k $ and
$\rho(\barj,j) < \rho(i,j) = k$. By

Lemma~\ref{lemma-stengthen-idempotent},
$\dbmij{m''}{i}{\bari} + \dbmij{m''}{\barj}{j}/2 \geq
\dbmij{m''}{i}{j}$.
Repeating the argument above it follows that $\dbmij{m^{k}}{i}{j} \geq \dbmij{m''}{i}{j}$
Hence
$\forall 0 \leq \ell < k. \dbm{m}^{k}_{\rho^{-1}(\ell)} =
\dbm{m''}_{\rho^{-1}(\ell)}$.
Now to show $\dbmij{m''}{i}{j} \geq \dbmij{m^{k+1}}{i}{j}$. Observe that:
\begin{align*}
  \dbmij{m''}{i}{j} &= \min(\dbmij{m'}{i}{j}, \frac{\dbmij{m}{i}{\bari} + \dbmij{m}{\barj}{j}}{2}) \\ 
                    &= \min \left ( \min \left (\begin{array}{l}
                                                  \dbmij{m}{i}{j}, \\
                                                  \dbmij{m}{i}{a} + d + \dbmij{m}{b}{j}, \\
                                                  \dbmij{m}{i}{\bar{b}} + d + \dbmij{m}{\bar{a}}{j},\\
                                                  \dbmij{m}{i}{\bar{b}} + d + \dbmij{m}{\bar{a}}{a} + d + \dbmij{m}{b}{j}, \\
                                                  \dbmij{m}{i}{a} + d + \dbmij{m}{b}{\bar{b}} + d + \dbmij{m}{\bar{a}}{j} \\
                                                \end{array} \right ), \frac{\dbmij{m}{i}{\bari} + \dbmij{m}{\barj}{j}}{2} \right) \\
                    &\geq \min \left ( \min \left (\begin{array}{l}
                                                  \dbmij{m^{k}}{i}{j}, \\
                                                  \dbmij{m^{k}}{i}{a} + d + \dbmij{m^{k}}{b}{j}, \\
                                                  \dbmij{m^{k}}{i}{\bar{b}} + d + \dbmij{m^{k}}{\bar{a}}{j},\\
                                                  \dbmij{m^{k}}{i}{\bar{b}} + d + \dbmij{m^{k}}{\bar{a}}{a} + d + \dbmij{m^{k}}{b}{j}, \\
                                                  \dbmij{m^{k}}{i}{a} + d + \dbmij{m^{k}}{b}{\bar{b}} + d + \dbmij{m^{k}}{\bar{a}}{j} \\
                                                \end{array} \right ), \frac{\dbmij{m'}{i}{\bari} + \dbmij{m'}{\barj}{j}}{2} \right) = \dbmij{m^{k+1}}{i}{j} \\ 
\end{align*}
\noindent Hence it follows
$\forall 0 \leq \ell < k + 1. \dbm{m^{k+1}_{\rho^{-1}(\ell)}} = \dbm{m''_{\rho^{-1}(\ell)}}$. Note
$\forall k + 1 \leq \ell < 4n^2 . \dbm{m^{k+1}_{\rho^{-1}(\ell)}} =
\dbm{m_{\rho^{-1}(\ell)}}$ follows from the inductive hypothesis and the definition of $\dbmij{m^{k+1}}{i}{j}$. 

Suppose $\dbm{m'}$ is inconsistent hence
$\dbmij{m'}{i}{i} < 0$. Put $k = \rho(i,i)$. 
But $\dbm{m^k} \leq \dbm{m}$ and by Proposition~\ref{lemma-monotonicity} 
$\dbmij{m^{4n^2}}{i}{i} = \dbmij{m^{k+1}}{i}{i} \leq \dbmij{m'}{i}{i} < 0$ as required.
\qed
\end{proof}

  \begin{proof}[for lemma~\ref{lemma:tightencloseidempotent}]
    Suppose $\dbm{m'''}$ is consistent.     
    By Proposition \ref{prop:strongreductive} $\dbm{m}''' \leq \dbm{m}''$
    and by Proposition~\ref{prop:tightenreductiveness}  $\dbm{m}'' \leq \dbm{m}'$ thus $\dbm{m}'$
    is consistent. By Theorem~\ref{thm:incrclosureclosed} $\dbm{m}'$ is closed
    hence $\dbmij{m'}{a}{a} = \dbmij{m'}{b}{b} = \dbmij{m'}{\bara}{\bara} = \dbmij{m'}{\barb}{\barb} = 0$.
    By Corollary~\ref{corollary:idempotence-facts} it follows that
    $\dbmij{m'}{a}{b} \leq \dbmij{m'}{a}{a} + d + \dbmij{m'}{b}{b} = d$
    and 
    $\dbmij{m'}{\barb}{\bara} \leq \dbmij{m'}{\barb}{\barb} + d + \dbmij{m'}{\bara}{\bara} = d$
    therefore
    $\dbmij{m'''}{a}{b} \leq d$
    and
    $\dbmij{m'''}{\barb}{\bara} \leq d$.
    By Proposition~\ref{lemma:coherence} $\dbm{m'}$ is coherent hence $\dbm{m'''}$ is closed by Lemma~\ref{thm:strongclosurestrengthen}.
    \begin{itemize}
    \item To show $\dbmij{m'''}{i}{a} + d + \dbmij{m'''}{b}{j} \geq \dbmij{m'''}{i}{j}$. Since $\dbm{m'''}$ is closed it follows
      \begin{align*}
        \dbmij{m'''}{i}{a} + d + \dbmij{m'''}{b}{j} & \geq \dbmij{m'''}{i}{a} + \dbmij{m'''}{a}{b} + \dbmij{m'''}{b}{j} 
                                                      \geq \dbmij{m'''}{i}{b} + \dbmij{m'''}{b}{j} 
                                                      \geq \dbmij{m'''}{i}{j} 
      \end{align*}
      
    \item To show $\dbmij{m'''}{i}{\barb} + d + \dbmij{m'''}{\bara}{j} \geq \dbmij{m'''}{i}{j}$.
      Since $\dbm{m'''}$ is closed it follows
      \begin{align*}
        \dbmij{m'''}{i}{\barb} + d + \dbmij{m'''}{\bara}{j} 
        & \geq \dbmij{m'''}{i}{\barb} + \dbmij{m'''}{\barb}{\bara} + \dbmij{m'''}{\bara}{j} 
          \geq \dbmij{m'''}{i}{\bara} + \dbmij{m'''}{\bara}{j}  \geq \dbmij{m'''}{i}{j} 
      \end{align*}
      
    \item To show $\dbmij{m'''}{i}{a} + d + \dbmij{m'''}{b}{\barb} + d + \dbmij{m'''}{\bara}{j} \geq \dbmij{m'''}{i}{j}$.
      Since $\dbm{m'''}$ is closed 
      \begin{align*}
        \dbmij{m'''}{i}{a} + d + \dbmij{m'''}{b}{\barb} + d + \dbmij{m'''}{\bara}{j} 
        & \geq \dbmij{m'''}{i}{a} + \dbmij{m'''}{a}{b} + \dbmij{m'''}{b}{\barb} + \dbmij{m'''}{\barb}{\bara} + \dbmij{m'''}{\bara}{j} \\
        & \geq \dbmij{m'''}{i}{b} + \dbmij{m'''}{b}{\barb} + \dbmij{m'''}{\barb}{\bara} + \dbmij{m'''}{\bara}{j} \\
        & \geq \dbmij{m'''}{i}{\barb} + \dbmij{m'''}{\barb}{\bara} + \dbmij{m'''}{\bara}{j} \\
        & \geq \dbmij{m'''}{i}{\bara} + \dbmij{m'''}{\bara}{j} \geq \dbmij{m'''}{i}{j} 
      \end{align*}
      
    \item To show $\dbmij{m'''}{i}{\barb} + d + \dbmij{m'''}{\bara}{a} + d + \dbmij{m'''}{b}{j} \geq \dbmij{m'''}{i}{j}$.
      Since $\dbm{m'''}$ is closed 
      \begin{align*}
        \dbmij{m'''}{i}{\barb} + d + \dbmij{m'''}{\bara}{a} + d + \dbmij{m'''}{b}{j} 
        & \geq \dbmij{m'''}{i}{\barb} + \dbmij{m'''}{\barb}{\bara} + \dbmij{m'''}{\bara}{a} + \dbmij{m'''}{a}{b} + \dbmij{m'''}{b}{j} \\
        & \geq \dbmij{m'''}{i}{\bara} + \dbmij{m'''}{\bara}{a} + \dbmij{m'''}{a}{b} + \dbmij{m'''}{b}{j} \\
        & \geq \dbmij{m'''}{i}{a} + \dbmij{m'''}{a}{b} + \dbmij{m'''}{b}{j} \\
        & \geq \dbmij{m'''}{i}{b} + \dbmij{m'''}{b}{j} \geq \dbmij{m'''}{i}{j}
      \end{align*}

    \end{itemize}
    By Proposition~\ref{lemma-idempotence} it follows that $\dbm{m^*} = \dbm{m'''}$. \qed
\end{proof}

\begin{proof}[for theorem~\ref{lemma:tightinsitu}] Suppose $\dbm{m'}$ is consistent.

Let $k = 0$.  
It vacuously follows that $\forall 0 \leq \ell < k. \dbm{m^{k}_{\rho^{-1}(\ell)}} = \dbm{m''_{\rho^{-1}(\ell)}}$.
Moreover $\forall k \leq \ell < 4n^2.  \dbm{m^{k}_{\rho^{-1}(\ell)}} = \dbm{m_{\rho^{-1}(\ell)}}$  since $\dbm{m}^{0} = \dbm{m}$.
Now let $k > 0$ and suppose $\rho(i,j) = k$.  Now suppose that $j = \bari$.  Then
\[
  \dbmij{m^{k+1}}{i}{j} = 2 \left\lfloor \min \left (
      \begin{array}{l}
        \dbmij{m^{k}}{i}{\bari}, \\
        \dbmij{m^{k}}{i}{a} + d + \dbmij{m^{k}}{b}{\bari}, \\
        \dbmij{m^{k}}{i}{\bar{b}} + d + \dbmij{m^{k}}{\bar{a}}{\bari},\\
        \dbmij{m^{k}}{i}{\bar{b}} + d + \dbmij{m^{k}}{\bar{a}}{a} + d + \dbmij{m^{k}}{b}{\bari}, \\
        \dbmij{m^{k}}{i}{a} + d + \dbmij{m^{k}}{b}{\bar{b}} + d + \dbmij{m^{k}}{\bar{a}}{\bari} \\
      \end{array}
    \right )/2 \right\rfloor  
\]
If $\rho^{-1}(i,a) < k$ then
$\dbmij{m^{k}}{i}{a} = \dbmij{m'''}{i}{a}$ whereas if
$\rho^{-1}(i,a) \geq k$ then
$\dbmij{m^{k}}{i}{a} = \dbmij{m}{i}{a} \geq \dbmij{m'''}{i}{a}$: this
implies that $\dbmij{m^{k}}{i}{a} \geq \dbmij{m'''}{i}{a}$ and
likewise $\dbmij{m^{k}}{b}{j} \geq \dbmij{m}{b}{j}$. By
Lemma~\ref{lemma:tightencloseidempotent} and
Corollary~\ref{corollary:idempotence-facts} it follows that
$\dbmij{m^{k}}{i}{a} + d + \dbmij{m^{k}}{b}{j} \geq \dbmij{m'''}{i}{a}
+ d + \dbmij{m'''}{b}{j}$. By a similar argument
$\dbmij{m^k}{i}{\barb} + d + \dbmij{m^k}{\bara}{j} \geq
\dbmij{m'''}{i}{j}$,
$\dbmij{m^k}{i}{a} + d + \dbmij{m^k}{b}{\barb} + d +
\dbmij{m^k}{\bara}{j} \geq \dbmij{m'''}{i}{j}$ and likewise
$\dbmij{m^k}{i}{\barb} + d + \dbmij{m^k}{\bara}{a} + d +
\dbmij{m^k}{b}{j} \geq \dbmij{m'''}{i}{j}$. Moreover
$(\dbmij{m''}{i}{\bari} + \dbmij{m''}{\barj}{j})/2 \geq
\min(\dbmij{m''}{i}{j}, (\dbmij{m''}{i}{\bari} +
\dbmij{m''}{\barj}{j})/2) = \dbmij{m'''}{i}{j}$. Thus
$\dbmij{m^{k}}{i}{j} \geq \dbmij{m'''}{i}{j}$.
Now to show $\dbmij{m'''}{i}{j} \geq \dbmij{m^{k+1}}{i}{j}$.
\begin{align*}
  \dbmij{m'''}{i}{\bari} = \dbmij{m''}{i}{\bari} = 2 \left \lfloor \dbmij{m'}{i}{\bari}/2 \right \rfloor &=
                                                                                                           2 \left \lfloor 
                                                                                                           \min \left ( \begin{array}{l}
                                                                                                                          \dbmij{m}{i}{\bari}, \\
                                                                                                                          \dbmij{m}{i}{a} + d + \dbmij{m}{b}{\bari}, \\
                                                                                                                          \dbmij{m}{i}{\bar{b}} + d + \dbmij{m}{\bar{a}}{\bari},\\
                                                                                                                          \dbmij{m}{i}{\bar{b}} + d + \dbmij{m}{\bar{a}}{a} + d + \dbmij{m}{b}{\bari}, \\
                                                                                                                          \dbmij{m}{i}{a} + d + \dbmij{m}{b}{\bar{b}} + d + \dbmij{m}{\bar{a}}{\bari} \\
                                                                                                                        \end{array} \right )/2 \right \rfloor  \\
                                                                                                         & \geq  2 \left \lfloor 
                                                                                                           \min \left ( \begin{array}{l}
                                                                                                                          \dbmij{m^{k}}{i}{\bari}, \\
                                                                                                                          \dbmij{m^{k}}{i}{a} + d + \dbmij{m^{k}}{b}{\bari}, \\
                                                                                                                          \dbmij{m^{k}}{i}{\bar{b}} + d + \dbmij{m^{k}}{\bar{a}}{\bari},\\
                                                                                                                          \dbmij{m^{k}}{i}{\bar{b}} + d + \dbmij{m^{k}}{\bar{a}}{a} + d + \dbmij{m^{k}}{b}{\bari}, \\
                                                                                                                          \dbmij{m^{k}}{i}{a} + d + \dbmij{m^{k}}{b}{\bar{b}} + d + \dbmij{m^{k}}{\bar{a}}{\bari} \\
                                                                                                                        \end{array} \right )/2 \right \rfloor  \\ &= \dbmij{m^{k+1}}{i}{\bari}
\end{align*}
Hence it follows
$\forall 0 \leq \ell < k + 1. \dbm{m^{k+1}_{\rho^{-1}(\ell)}} =
\dbm{m''_{\rho^{-1}(\ell)}}$. Moreover
$\forall k + 1 \leq \ell < 4n^2 . \dbm{m^{k+1}_{\rho^{-1}(\ell)}} =
\dbm{m_{\rho^{-1}(\ell)}}$ follows from the inductive hypothesis and
the definition of $\dbmij{m^{k+1}}{i}{j}$.

Now suppose that
$j \neq \bari$ and consider
\begin{align*}
  \dbmij{m^{k+1}}{i}{j} &= \min \left ( 
                         \begin{array}{l}
                           \dbmij{m^k}{i}{j} \\
                           \dbmij{m^k}{i}{a} + d + \dbmij{m^k}{b}{j} \\
                           \dbmij{m^k}{i}{\barb} + d + \dbmij{m^k}{\bara}{j} \\
                           \dbmij{m^k}{i}{a} + d + \dbmij{m^k}{b}{\barb} + d + \dbmij{m^k}{\bara}{j} \\
                           \dbmij{m^k}{i}{\barb} + d + \dbmij{m^k}{\bara}{a} + d + \dbmij{m^k}{b}{j}, \\
                           (\dbmij{m^k}{i}{\bari} + \dbmij{m^k}{\barj}{j})/2 
                         \end{array} \right ) 
\end{align*}
Notice that
$(\dbmij{m^{k}}{i}{\bari} + \dbmij{m^{k}}{\barj}{j})/2 \geq
(\dbmij{m'''}{i}{\bari} + \dbmij{m'''}{\barj}{j})/2$ since
$\rho(i,\bari) < 2n \leq \rho(i,j) = k$ and similarly
$\rho(\barj,j) < \rho(i,j) = k$. By
Lemma~\ref{lemma:tightencloseidempotent}
$\dbmij{m'''}{i}{\bari} + \dbmij{m'''}{\barj}{j}/2 \geq
\dbmij{m'''}{i}{j}$ and thus
$(\dbmij{m^{k}}{i}{\bari} + \dbmij{m^{k}}{\barj}{j})/2 \geq
\dbmij{m'''}{i}{j}$. Repeating the argument above it follows that $\dbmij{m^{k+1}}{i}{j} \geq \dbmij{m'''}{i}{j}$.
Now to show $\dbmij{m'''}{i}{j} \geq \dbmij{m^{k+1}}{i}{j}$ observe:
\begin{align*}
  \dbmij{m'''}{i}{j} &= \dbmij{m''}{i}{j} = \min \left (\dbmij{m'}{i}{j}, \frac{\dbmij{m'}{i}{\bari} + \dbmij{m'}{\barj}{j}}{2} \right ) \\
                     &= \min \left ( \min \left (\begin{array}{l}
                                                  \dbmij{m}{i}{j}, \\
                                                  \dbmij{m}{i}{a} + d + \dbmij{m}{b}{j}, \\
                                                  \dbmij{m}{i}{\bar{b}} + d + \dbmij{m}{\bar{a}}{j},\\
                                                  \dbmij{m}{i}{\bar{b}} + d + \dbmij{m}{\bar{a}}{a} + d + \dbmij{m}{b}{j}, \\
                                                  \dbmij{m}{i}{a} + d + \dbmij{m}{b}{\bar{b}} + d + \dbmij{m}{\bar{a}}{j} \\
                                                 \end{array} \right ), \frac{\dbmij{m'}{i}{\bari} + \dbmij{m'}{\barj}{j}}{2} \right) \\
                      &= \min \left ( \min \left (\begin{array}{l}
                                                  \dbmij{m^{k}}{i}{j}, \\
                                                  \dbmij{m^{k}}{i}{a} + d + \dbmij{m^{k}}{b}{j}, \\
                                                  \dbmij{m^{k}}{i}{\bar{b}} + d + \dbmij{m^{k}}{\bar{a}}{j},\\
                                                  \dbmij{m^{k}}{i}{\bar{b}} + d + \dbmij{m^{k}}{\bar{a}}{a} + d + \dbmij{m^{k}}{b}{j}, \\
                                                  \dbmij{m^{k}}{i}{a} + d + \dbmij{m^{k}}{b}{\bar{b}} + d + \dbmij{m^{k}}{\bar{a}}{j} \\
                                                 \end{array} \right ), \frac{\dbmij{m'}{i}{\bari} + \dbmij{m'}{\barj}{j}}{2} \right)  = \dbmij{m^{k+1}}{i}{j}\\
\end{align*}
Hence it follows
$\forall 0 \leq \ell < k + 1. \dbm{m^{k+1}_{\rho^{-1}(\ell)}} =
\dbm{m''_{\rho^{-1}(\ell)}}$. Note
$\forall k + 1 \leq \ell < 4n^2 . \dbm{m^{k+1}_{\rho^{-1}(\ell)}} =
\dbm{m_{\rho^{-1}(\ell)}}$ follows by inductive hypothesis and definition of $\dbmij{m^{k+1}}{i}{j}$.

Suppose $\dbm{m'}$ is inconsistent hence $\dbmij{m'}{i}{i} < 0$. Put
$k = \rho(i,i)$. But $\dbm{m^k} \leq \dbm{m}$ and by
Proposition~\ref{lemma-monotonicity}
$\dbmij{m^{4n^2}}{i}{i} = \dbmij{m^{k+1}}{i}{i} \leq \dbmij{m'}{i}{i}
< 0$ as required. \qed
\end{proof}
